\newif\ifarxiv
\arxivtrue
\newif\ifjournal
\journalfalse

\documentclass{article}

\ifarxiv
\usepackage{times}
\usepackage{savesym}
\usepackage{amsmath,amsthm,amsbsy}
\savesymbol{iint}
\savesymbol{openbox}
\usepackage{txfonts}
\restoresymbol{TXF}{iint}
\restoresymbol{TXF}{openbox}
\def\expspace#1{{\,#1}}
\else
\usepackage[lf]{MinionPro}
\usepackage{amsmath,amsthm,amsbsy}
\def\expspace#1{#1}
\fi

\usepackage{geometry}
\usepackage{color}
\definecolor{labelkey}{rgb}{0,0,.75}
\definecolor{MyGreen}{rgb}{0,.6,.2}
\definecolor{MyDarkBlue}{rgb}{.1,.1,.75}
\usepackage[colorlinks=true, citecolor=MyGreen, linkcolor=MyDarkBlue]{hyperref}
\usepackage{tensor}
\usepackage[all,cmtip]{xy}
\usepackage{graphicx}

\date{July 6, 2014}
\title{Initial Data in General Relativity Described by Expansion, Conformal Deformation and Drift}

\author{David Maxwell}

\DeclareMathOperator{\ck}{\bf L}
\renewcommand{\div}{\mathop{\rm div}\nolimits}
\newcommand{\Div}{\mathop{\rm Div}\nolimits}
\DeclareMathOperator{\Lap}{\Delta}
\DeclareMathOperator{\Vol}{\rm Vol}
\DeclareMathOperator{\tr}{\rm tr}
\DeclareMathOperator{\Lie}{\mathrm{Lie}}

\DeclareMathOperator{\id}{\rm id}
\DeclareMathOperator{\extd}{\mathbf d}
\renewcommand{\Im}{\mathop{\mathrm{Im}}\nolimits}
\DeclareMathOperator{\Ker}{\mathrm{Ker}}
\DeclareMathOperator{\Drift}{\mathrm{Drift}}

\def\ip<#1,#2>{\left<#1,#2\right>}
\let\ol\overline
\def\tensor{\otimes}
\newcommand{\ra}{\rightarrow}
\newcommand{\dimk}{a}

\newcommand{\Reals}{\mathbb{R}}

\newcommand{\calD}{\mathcal{D}}
\newcommand{\calE}{\mathcal{E}}
\newcommand{\calF}{\mathcal{F}}
\newcommand{\calG}{\mathcal{G}}

\newcommand{\calK}{\mathcal{K}}
\newcommand{\calL}{\mathcal{L}}
\newcommand{\calM}{\mathcal{M}}

\newcommand{\calQ}{\mathcal{Q}}

\newcommand{\calC}{\mathcal{C}}
\newcommand{\calV}{\mathcal{V}}
\newcommand{\calX}{\mathcal{X}}
\newcommand{\bfsigma}{\boldsymbol{\sigma}}
\def\dc<#1,#2,#3>{\{#1;\;#2,#3\}}
\def\define#1{{\bf #1}}
\newcommand{\stokeslap}{\mathcal{L}}

\begin{document}
\newtheorem{theorem}{Theorem}[section]
\newtheorem{conjecture}[theorem]{Conjecture}
\newtheorem{problem}[theorem]{Problem}
\newtheorem{proposition}[theorem]{Proposition}
\newtheorem{corollary}[theorem]{Corollary}
\newtheorem{lemma}[theorem]{Lemma}
\theoremstyle{definition}
\newtheorem{definition}[theorem]{Definition}
\numberwithin{equation}{section}

\maketitle
\begin{abstract}
The conformal method is a technique for finding 
Cauchy data in general relativity solving the 
Einstein constraint equations, and its parameters include a conformal class,
a conformal momentum (as measured by a densitized lapse), 
and a mean curvature.  Although the conformal method is successful
in generating constant mean curvature (CMC) solutions of the constraint
equations, it is unknown how well it applies in the non-CMC setting,
and there have been indications that it encounters difficulties there.
We are therefore motivated to investigate alternative generalizations of the CMC conformal method.

Introducing a densitized lapse into the 
ADM Lagrangian, we find that solutions of the 
momentum constraint can be described in terms of three parameters.
The first is conformal momentum as it appears in the standard conformal method.
The second is volumetric momentum, which appears as an explicit parameter in
the CMC conformal method, but not in the non-CMC formulation. We have
called the third parameter drift momentum, and it is the
conjugate momentum to infinitesimal 
motions in superspace that preserve conformal
class and volume form up to independent diffeomorphisms.  This
decomposition of solutions of the momentum constraint
 leads to extensions of the CMC conformal method
where conformal and volumetric momenta both appear as 
parameters. There is more than one way to treat drift momentum,
in part because of an interesting duality that emerges,
and we identify three candidates for
incorporating drift into a variation of the conformal method.
\end{abstract}

\section{Introduction}\label{sec:intro}

An initial data set in general relativity consists of the
geometry and matter distribution of the universe at an instant in time, along
with the instantaneous rate of change of these quantities. The associated
Cauchy problem is to determine an ambient spacetime for the initial data
set that satisfies the Einstein equations as well as the applicable matter field equations.
In contrast to Newtonian gravity, 
initial data cannot be freely specified, and must satisfy certain
underdetermined compatibility conditions known as the Einstein constraint equations.  
These constraint PDEs admit a wide variety of
solutions, and as a consequence we have enormous flexibility, 
but not complete freedom, in specifying initial conditions.  One would therefore like to 
find intrinsic parameters describing the set of solutions of the constraint equations.

This problem is already difficult, and not yet understood, for vacuum spacetimes with
a vanishing cosmological constant,
in which case an initial data set consists of a Riemannian manifold $(M^n,g_{ab})$
and a symmetric tensor $K_{ab}$ representing the second fundamental form of the embedding of $M^n$ 
into its ambient spacetime. Vacuum spacetimes
are Ricci flat, and hence the Gauss and Codazzi equations imply the following relations
between $g_{ab}$ and $K_{ab}$:
\begin{subequations}\label{eq:constraints}
\begin{alignat}{2}
R_{ g} - | K|_{ g}^2 + (\tr_{ g}  K)^2 &= 0 &\qquad&\text{\small[Hamiltonian constraint]}\label{eq:hamiltonian}\\
\div_{ g}  K- \extd (\tr_{g} K) &= 0&\qquad&\text{\small[momentum constraint]}\label{eq:momentum}
\end{alignat}
\end{subequations}
where $\mathbf d$ is the exterior derivative, $R_g$ is the scalar curvature, $\div_g$ is the divergence, and $\tr_g$ is the trace operator
of $g_{ab}$. Equations \eqref{eq:constraints} are the vacuum Einstein constraint equations, and the fact that they are underdetermined reflects the physical property that gravitational waves can propagate in
vacuum, as well as the gauge property that we have freedom to choose coordinates in spacetime.

There are a number of approaches for finding solutions of the constraint equations in specific circumstances,
and we note in particular the examples provided by gluing methods
\cite{Chrusciel:2003ug}\cite{Corvino:2006wf}\cite{Chrusciel:2005jo}\cite{Chrusciel:2011jp}
\cite{Carlotto:2014}, as well as the density and perturbation
techniques of \cite{Huang:2009gm}\cite{Huang:2010dh}. These constructions provide deep insight into the 
diversity of solutions of the constraint equations and their properties, but they do not yield parameterizations. Indeed, as far as 
concrete parameterizations are concerned,  there is presently only a single general purpose candidate,
the conformal method, and it occurs in the literature in two principal variations. The
original conformal method was initiated by Lichnerowicz \cite{Lichnerowicz:1944} 
and later extended by York to construct constant-mean curvature (CMC) solutions \cite{York:1973fla}  
and, along with O'Murchadha, to construct non-CMC solutions of the constraint equations \cite{OMurchadha:1974bf}. 
Subsequently York developed the Lagrangian conformal thin-sandwich (CTS) method \cite{YorkJr:1999jo}
and then with Pfeiffer presented the Hamiltonian form of the CTS method \cite{Pfeiffer:2003ka}.
It turns out that the standard and CTS conformal methods are two different ways to write down the 
same parameterization of the constraint equations \cite{Maxwell:2014a}, and we will 
refer to all these techniques collectively as the conformal method. Using the
language of \cite{Maxwell:2014a} that emphasizes the role of conformal geometry,
the Hamiltonian form of the conformal method has four parameters:
\begin{itemize}
\item A conformal class $\mathbf g$, represented by the choice of a metric $g_{ab}\in \mathbf g$.
\item A conformal momentum $\bfsigma$, represented by a pair $(g_{ab};\; \sigma_{ab})$
where $\sigma_{ab}$ is trace-free and divergence free.  Writing $q=2n/(n-2)$ for the critical
Sobolev exponent, if $\phi>0$ is a conformal
factor then the pair $(\phi^{q-2} g_{ab};\; \phi^{-2}\sigma_{ab})$
represents the same
conformal momentum $\bfsigma$.
\item An arbitrary function $\tau$ dictating a mean curvature.
\item A so-called densitized lapse represented by
a pair $(g_{ab};\; N)$ where $N$ is a positive function.  If
$\phi>0$ is a conformal factor, $(\phi^{q-2} g_{ab};\; \phi^q N)$
represents the same densitized lapse.
\end{itemize}
The choice of a densitized lapse $\mathbf N$ allows for a notion of conformal momentum to
be assigned to a solution of the constraint equations, and after fixing a densitized 
lapse every solution of the constraint equations uniquely determines 
conformal parameters $(\mathbf g, \bfsigma, \tau, \mathbf N)$.  The central question 
for the conformal method is the extent to which this map is a bijection. 

Suppose for concreteness that $M$ is compact.  If we restrict our attention
to CMC solutions of the constraint equations (i.e. solutions with $\tau\equiv\tau_0$
for some constant $\tau_0$) then the map from solutions of the constraint equations 
onto conformal parameters is indeed a bijection \cite{Isenberg:1995bi}, with the following caveats
based on the sign of the Yamabe invariant $Y_{\mathbf g}$ of the conformal class
$\mathbf g$:
\begin{itemize}
\item If $Y_g>0$, then $\bfsigma=0$ is impossible.
\item If $Y_g<0$, then $\tau_0=0$ is impossible.
\item If $Y_g=0$, then $\bfsigma=0$ is impossible and $\tau_0=0$ is impossible,
except that there is a homothety family of solutions corresponding to
the case where both  $\bfsigma=0$ and $\tau_0=0$.
\end{itemize}
Moreover, these same results largely
extend into the near-CMC regime: see, e.g., \cite{Isenberg:1996fia} and \cite{Allen:2008ef}
as augmented by \cite{Maxwell:2014a} 
for existence and uniqueness theorems, and see \cite{Isenberg:2004jd} for non-existence results when $Y_\mathbf g\ge 0$ and $\bfsigma =0$.  Indeed, the theory for near-CMC solutions is satisfactory and
complete, except that
existence is not understood if $\mathbf g$ admits nontrivial conformal Killing fields.

On the other hand, the properties of the conformal method when $\tau$ is far-from-CMC are largely unknown.
On compact manifolds we have a single far-from CMC existence theorem \cite{Holst:2009ce}\cite{Maxwell:2009co}:
given a Yamabe positive conformal class $\mathbf g$ and an arbitrary mean curvature $\tau$, 
if $\bfsigma\neq 0$ is close to zero 
(with closeness depending on $\tau$), there exists \textit{at least} one
associated solution of the constraint equations.
This foray into far-from-CMC territory
can, moreover, be thought of as a perturbation off of a CMC solution with $\tau_0=0$ \cite{Gicquaud:2014}.
And although the far-from-CMC existence result 
%from \cite{Holst:2009ce}\cite{Maxwell:2009co} 
is  consistent with 
the possibility that the good properties of the CMC conformal method extend to far-from-CMC
solutions, subsequent 
	case studies in \cite{Maxwell:2011if} and \cite{Maxwell:2014b} show that at least 
sometimes they do not. 

The work in \cite{Maxwell:2011if} exhibits a family 
of symmetric conformal data on the torus such that
in the far-from-CMC regime there are 
multiple solutions when $\bfsigma$ is small, no solutions with the symmetry
of the data when $\bfsigma$ is large, and certain rare cases that lead
to one-parameter families of non-CMC solutions.  
% The dividing line between small and large values of $\bfsigma$
% corresponding to the boundary between existence and non-existence 
% depends on the specific mean curvature, and it is unknown if there is some \textit{a priori} computable 
% condition that identifies this boundary.  
The mean curvatures studied in \cite{Maxwell:2011if} have $L^\infty$ regularity, and although
it not known if similar difficulties occur for smooth mean curvatures, 
the follow-up study in \cite{Maxwell:2014b} shows that at least 
the one-parameter families persist.

The conformal parameters considered in \cite{Maxwell:2014b} 
have the form $(\mathbf g, \mu \bfsigma^\flat, \tau, \mathbf N)$ where 
$\mathbf g$ is the conformal class of a flat product metric $g_{ab}$ on the torus, $\bfsigma^\flat$
is a particular conformal momentum, $\mu$ is a constant, and where $\tau$ and $\mathbf N=(g_{ab};\; N)$
are arbitrary, except that $\tau$ and $N$ depend on only one factor of the torus. 
Writing
\begin{equation}\label{eq:taustarintro}
\tau^* = \frac{\int_M N \tau\; \omega_g}{\int_M N\; \omega_g}
\end{equation}
where $\omega_g$ is the volume form of $g_{ab}$, \cite{Maxwell:2014b} 
shows that if $\mu$ and $\tau^*$ have the same sign, then the conformal
parameters generate a slice of a flat spacetime (typically a Kasner solution,
with certain other spacetimes occurring non-generically). The case where $\tau^*=0$
is special, however: if $\mu$ and $\tau^*$ both vanish, then the conformal 
parameters construct a one parameter family of solutions of the constraint equations.
Note that if $\tau=\tau_0$ for some constant $\tau_0$, then $\tau^*=\tau_0$ and the CMC
one-parameter families occur when $\tau_0=0$.  But if $\tau$ is not constant then the
computation of $\tau^*$ involves a particular choice of representative of $\mathbf g$,
and the condition $\tau^*=0$ is not readily computed in advance. Indeed, $\tau^*$ can be 
computed with respect to the physical metric that solves the constraint equations, 
but to compute $\tau^*$ when working with some other background metric, one must first 
conformally transform to a flat metric, at which point one has all but solved the 
constraint equations \cite{Maxwell:2014b}.  Hence we have an example of non-uniqueness
for certain non-CMC conformal parameters where the non-uniqueness is difficult to detect
\textit{a priori}.

The success of the conformal method in the CMC setting has 
physical consequences including, for example, Fischer and Moncrief's program of
Hamiltonian reduction \cite{Fischer:2001kh}.
In contrast, failures of the conformal method for non-CMC conformal parameters 
may not imply anything in particular about general relativity.
The set of solutions of the constraint equations has, when given a suitable topology, a manifold structure
\cite{Chrusciel:2003ug}\cite{Bartnik:2005tl} at generic points, 
and there are many possible choices of charts for this manifold.
Although the conformal method provides a useful and successful chart in a neighborhood of CMC solutions,
we interpret the evidence to date as suggesting that this chart simply breaks 
down outside of this neighborhood.  If this is indeed the case, the details of 
this breakdown may be meaningful facts about the conformal method, but perhaps not
about the constraint equations.

In this article we examine the possibility that the CMC conformal method admits
an extension, other than the standard conformal method, that potentially 
has better properties for non-CMC solutions of the constraint equations.  
In particular, we identify geometrically
and physically motivated alternatives that replace the mean curvature parameter $\tau$ with 
two independent quantities: the constant $\tau^*$ from equation \eqref{eq:taustarintro} along with 
a second parameter, described below, that we will call a drift.  The guiding principle of 
leading to these alternatives is to treat
the densitized lapse as a fundamental object, and to apply it uniformly to both
conformal and volumetric degrees of freedom.

Densitized lapses first appeared in the context of the constraint equations
in York's  development of the conformal thin sandwich method \cite{YorkJr:1999jo},
where they occur as lapses that conformally transform according to
$N\mapsto \phi^q N$ when we change $g_{ab}\mapsto\phi^{q-2}g_{ab}$.
Although densitized lapses arrived somewhat
late in the development of the conformal method, because the original conformal method 
and the  CTS methods are equivalent, densitized lapses have been a part of the conformal
method all along. In this work we represent a densitized lapse by
a choice of volume form $\alpha$ on $M$.  To every metric $g_{ab}$ we then assign a
lapse according to
\begin{equation}\label{eq:dlapseintro}
N_{g,\alpha} = \frac{\omega_g}{\alpha}
\end{equation}
where $\omega_g$ is the volume form of $g_{ab}$.  Since volume forms conformally
transform according to $\omega_g\mapsto \phi^q \omega_g$ we recover York's transformation law,
and in terms of our earlier notation the volume form $\alpha$ corresponds to the densitized lapse
$\mathbf N$ represented by $(g_{ab}; \omega_g/\alpha)$.
Note that if we interpret $\alpha$ as `coordinate area', 
then equation \eqref{eq:dlapseintro} expresses 
the lapse as the ratio of physical to coordinate area in addition to its standard interpretation
as the ratio of physical to coordinate time.  Using equation \eqref{eq:dlapseintro}
to rewrite the usual Arnowitt-Dieser-Misner (ADM) Lagrangian\cite{ADM1962} 
so that it depends on $\alpha$ instead of the standard lapse,
we find that the following features emerge.
\begin{itemize}
\item The densitized lapse assigns each pair 
$(g_{ab},K_{ab})$, regardless of whether it solves the constraint equations or not,
a conformal velocity and a conformal
momentum of motion in $\calC/\calD_0$, where $\calC$ is the set of conformal classes on $M$ and $\calD_0$ is
the connected component of the identity of the diffeomorphism group.  
These dynamical quantities
are associated with their standard ADM counterparts as described
in diagram \eqref{diag:dl-legendre-conf}, but doing so requires a densitized lapse
rather than the standard ADM lapse.
For CMC solutions of the constraints, the measurement of conformal momentum is independent
of the choice of densitized lapse, but this is not true for non-CMC solutions.
The conformal method uses conformal velocity or conformal momentum as one of its parameters 
depending on whether we use the Lagrangian or the Hamiltonian formulation,
and these quantities are connected to each via a Legendre transformation associated with a
Lagrangian (conformal kinetic energy) on the tangent bundle $T\; \calC/\calD_0$. 
Sections \ref{sec:conformalTS} and \ref{sec:conformalLegendre}
describe these results.
\item The densitized lapse assigns each pair $(g_{ab},K_{ab})$ a volumetric velocity and momentum 
of motion in
 $\calV/\calD_0$, where $\calV$ is the set of volume forms.  Volumetric velocity and momentum are associated
with ADM velocity and momentum as described in diagram \ref{diag:dl-legendre-conf-vol}, and
again this relationship uses a densitized lapse.  Volumetric momentum is a single number,
and if $g^{ab}K_{ab}=\tau_0$ for some constant $\tau_0$, the volumetric momentum is
$-2\kappa \tau_0$ where $\kappa = (n-1)/n$.  For non-constant mean curvature 
the measurement of conformal momentum depends on the choice of densitized lapse
and equals $-2\kappa\tau^*$ where $\tau^*$ is the quantity \eqref{eq:taustarintro}
identified previously in \cite{Maxwell:2014b}. In the CMC conformal method, the volumetric
momentum is one of the explicit parameters, but this is not the case for the non-CMC conformal method.
Volumetric velocity and momentum are connected to each via a Legendre transformation associated with a
Lagrangian (volumetric kinetic energy) on the tangent bundle $T\; \calV/\calD_0$.  
Sections \ref{sec:volumetric} and \ref{sec:volumetricLegendre} describe these results,
and we see in these sections that the volumetric parameters have a structure that completely parallels
that of the conformal parameters, but that is ignored in the standard conformal method where
the mean curvature is specified explicitly.
\item Conformal momentum at a metric $g_{ab}$ 
is related to the York decomposition of trace-free tensors $A_{ab}$
\begin{equation}
A_{ab} = \sigma_{ab} + \frac{1}{2N_{g,\alpha}}\ck_g W_{ab}
\end{equation}
where $\sigma_{ab}$ is transverse traceless, $\ck_g$ is the conformal Killing operator
of $g_{ab}$,
and $W^a$ is a vector field. Volumetric momentum is associated with a York-like splitting
of mean curvature functions $\tau$:
\begin{equation}
\tau = \tau^* + \frac{1}{N_{g,\alpha}} \div V
\end{equation}
where $\tau^*$ is a constant and $V^a$ is a vector field.  In this way, $\tau^*$ plays
the same role for volumetric degrees of freedom that $\sigma_{ab}$ plays for
conformal degrees of freedom.

\item Let $\calM$ be the space of metrics.
Instantaneous motion in $\calM/\calD_0$ can be decomposed into
three components: conformal, volumetric, and drift.  The decomposition depends on the choice
of a densitized lapse, and the conformal and volumetric components of this decomposition
agree with the notions
of conformal and volumetric velocity just discussed.  A drift is an instantaneous motion
in  $\calM/\calD_0$ that preserves both conformal class (modulo diffeomorphisms) and volume form
(modulo diffeomorphisms).  Although a metric is uniquely determined by its conformal class and volume
form, there are nontrivial drifts, and indeed the drifts at a metric $g_{ab}$ 
can be identified with the space of vector fields on $M$, modulo the
divergence-free vector fields and conformal Killing fields of $g_{ab}$.
Section \ref{sec:drift} contains basic results concerning drifts.

\item It is well known that solutions of the momentum constraint correspond to
the momenta of motion in $\calM/\calD_0$.  In Section \ref{sec:drift-mom} we show
that after selection of a densitized lapse,
such momenta can be decomposed into three components: conformal, volumetric,
and drift.  The conformal and volumetric momenta are the quantities identified
previously, and a drift momentum at $g_{ab}$ can be described by a pair
of linked drifts $(\mathbf W, \mathbf V)$.
The drifts $\mathbf W$ and $\mathbf V$ can be represented by vector fields
$W^a$ and $V^a$ solving the drift equation
\begin{equation}\label{eq:driftintro}
\div_g\left[ \frac{1}{2N_{g,\alpha}} \ck_g W\right] = \kappa\, \mathbf{d}\left[ \frac{1}{N_{g,\alpha}} \div_g V\right]
\end{equation}
where $\ck_g$ is the conformal Killing operator of $g_{ab}$.  
Equation \eqref{eq:driftintro} has a remarkable symmetry between the conformal
and volumetric parameters $W^a$ and $V^a$.  We can specify $V^a$ and solve for $W^a$, in which case we can add
an arbitrary divergence-free vector field to $V^a$,  
but equation \eqref{eq:driftintro} is only solvable after adding a specific choice of conformal 
Killing field to $V^a$.  Conversely, we can specify $W^a$ and solve for $V^a$, 
in which case we can add an arbitrary conformal
Killing field to $W^a$ and we must additionally add a particular divergence-free vector field to 
ensure that equation \eqref{eq:driftintro} is solvable. So a pair 
$(\mathbf W, \mathbf V)$ representing a drift momentum is uniquely determined by either its
conformal drift $\mathbf W$ or its volumetric drift $\mathbf V$. Section \ref{sec:drift-mom} describes
these results in detail.   

\item The CMC solutions of the constraint equations are the solutions with zero drift momentum.

\item  Although solutions of the momentum constraint correspond to momenta in $\calM/\calD_0$,
solutions of the constraint equations are not well described in terms of velocities in
$\calM/\calD_0$.  There exist distinct solutions of the vacuum constraint equations,
generating distinct spacetimes, that nevertheless have identical geometries and velocities 
in $\calM/\calD_0$.  This phenomenon occurs because the drift momentum of a pair $(\mathbf W, \mathbf V)$
corresponds to a velocity $ \mathbf V-\mathbf W$ in $\calM/\calD_0$, and this can vanish even if 
$\mathbf W$ and $\mathbf V$ do not.  Either the conformal drift $\mathbf W$ or the volumetric drift
$\mathbf V$ can be taken as a parameter of motion that determines the other, but using the difference
$\mathbf V-\mathbf W$ leads to non-uniqueness.  Section \ref{sec:drift-vmke}
describes how we can take either factor $\mathbf W$  or $\mathbf V$ to be the drift velocity
corresponding to drift momentum, and that in either case we can construct
a Lagrangian (conformal or volumetric drift kinetic energy) whose Legendre transformation connects 
drift velocity and momentum.

\item The kinetic energy term of the ADM Lagrangian, when restricted to solutions of the momentum
constraint, decomposes into three independent terms corresponding to conformal, volumetric, and drift
kinetic energy.
\end{itemize}

These main results effectively comprise a study of the interaction of densitized lapses
with the momentum constraint.   In Section \ref{sec:driftcm} we then propose variations
of the conformal method where the parameters include a conformal class, 
a conformal momentum, a volumetric
momentum, and a vector field determining a drift momentum.  There is more than one way
to do this, however, and we present three candidates that each include the
CMC conformal method as a special case.  The resulting
equations are technically more challenging than those of the standard conformal method, 
and we therefore postpone their analysis for future work. Although we hope that
features of the momentum constraint documented here will assist those
efforts, it remains to be seen the extent to which these drift parameterizations, 
or perhaps some variation, outperform the conformal method.  Regardless, 
drifts have the the potential to play a role in understanding 
any variation of the CMC conformal method.  For example, the one-parameter families discovered
for the standard conformal method in \cite{Maxwell:2014b} all have the property 
that they have zero conformal momentum and zero volumetric momentum, but not-necessarily zero drift
momentum. Moreover, drifts are related to past difficulties
in applying the standard conformal method to construct non-CMC solutions of the constraints 
with metrics having nontrivial conformal Killing fields, and we discuss in Section \ref{sec:driftcm}
how the standard conformal method might be adjusted to account for conformal Killing fields.

Our main goal is to find well-motivated alternatives to the conformal method,
and in order minimize distraction we work under hypotheses that reduce the number of
technical details. In particular, we work only on compact manifolds, and we work only
in the smooth category.
Smoothness comes with the attendant complexity of
Fr\'echet manifolds, and we have emphasized linear algebra over topology when working
with their tangent spaces.  For example, direct sums and isomorphisms are always
meant in the sense of linear algebra, although in many cases it is obvious that the
subspaces involved are closed and the maps involved are at least continuous.
We adopt an intuitive (but precise) approach to working with tangent and
cotangent spaces to infinite dimensional spaces such as 
$\calC$ and $\calC/\calD_0$.
Sections \ref{sec:notation}, \ref{sec:conformalTS} and \ref{sec:volumetric}
contain the related definitions and details, and it is important to note that
the simplicity of our approach comes with the
penalty that objects such as $T^*\calC/\calD_0$ appearing in the theorems
are to be understood rather formally.  We also adopt some helpful but
non-standard notations regarding the trace/trace-free decomposition of
$T\calM$ and its interaction with the numerous quotient spaces we work with.
Again,
Sections \ref{sec:notation}, \ref{sec:conformalTS} and \ref{sec:volumetric}
contain the details.

\subsection{Notation and Conventions}\label{sec:notation}

Throughout we assume that $M$ is a smooth, compact, connected, oriented 
$n$-manifold with $n\ge 3$.  The set of smooth functions on $M$ is $C^\infty(M)$
and if $E$ is a smooth bundle over $M$, then $C^\infty(M,E)$ is the set of smooth sections
of $E$.  We write $TM$ and $T^*M$ for the tangent and cotangent bundles of $M$,
$S^2M$ and $S_2 M$ for the  bundles of symmetric $(2,0)$
and $(0,2)$ tensors, and $\Lambda^n M$ for the bundle
of $n$-forms.  All tensors are assumed to be smooth unless otherwise noted;
in Section \ref{sec:drift-mom} we work with $L^2$
Sobolev spaces $W^{k,2}$ where $k\in \mathbb Z$ denotes the order of differentiability.

We have the following sets of interest:
\begin{itemize}
\item[$\calM$,] the smooth metrics on $M$,
\item[$\calC$,] the conformal classes of smooth metrics,
\item[$\calV$,] the smooth volume forms (i.e., the positively oriented
elements of $C^\infty(M,\Lambda^n)$),
\item[$\calK$,] the space $C^\infty(M,S_2 M)$ of second fundamental
forms.
\end{itemize}

Three constants derived from the dimension $n$ arise sufficiently frequently that we
use the notation
\begin{equation}
q = \frac{2n}{n-2}\qquad\qquad \kappa = \frac{n-1}{n}\qquad\qquad \dimk = 2\kappa q = \frac{4(n-1)}{n-2}.
\end{equation}
We also use the symbol $a$ as an abstract index, but there should be no confusion since the 
constant $a$ defined above is never used as an exponent.

\subsubsection{The space \texorpdfstring{$\calM$}{M} of metrics}
The set $\calM$ of smooth metrics over $M$ is the open subset of 
positive definite elements of the Fr\'echet vector space $C^\infty(M,S_2 M)$.
Hence $\calM$ is a Fr\'echet manifold, and if $g_{ab}\in \calM$, then $T_g\calM=C^\infty(M,S_2 M)$.
Note that we use abstract index notation in this paper with the understanding that indices can 
be dropped freely if they clutter notation or are otherwise intrusive.

Let $g_{ab}\in\calM$. The dual space $(T_g\calM)^*$ contains a wide variety of distributions, 
and it will be convenient to work with a smaller subspace.  We define
\begin{equation}
T^*_g\calM = C^\infty(M,S^2 M\tensor \Lambda^n M).
\end{equation}
If $h_{ab}\in T_g(M)$ and $F^{ab}\omega\in T^*_g\calM$, then $F^{ab}$ acts on $h_{ab}$ via
\begin{equation}
\ip< F^{ab}\omega, h_{ab}> = \int_{M} F^{ab} h_{ab}\; \omega.
\end{equation}
One readily verifies that with this action, $T^*_g\calM\subseteq (T_g\calM)^*$.

There is a natural $L^2$ metric $\calG$ on $\calM$ defined by 
$\calG(h_{ab},\widehat h_{ab}) = \int_M \ip<h,\widehat h>_g\;\omega_g$
for all $h_{ab}$ and $\widehat h_{ab}\in T_g\calM$.  
Here and elsewhere $\omega_g$ is the oriented volume form of $g_{ab}$.
The metric $\calG$ determines a map from $T_g\calM$ to $(T_g\calM)^*$ defined by
\begin{equation}
h_{ab}\mapsto \calG(h_{ab},\cdot)
\end{equation}
and it is easy to see that $T^*_g\calM$ is the image of $T_g\calM$ under this map.
Thus we have a natural identification of $T_g\calM$ with $T^*_g\calM$.

The trace-free and pure-trace subspaces of $T_g\calM$ play an important role
in this paper and it will be helpful to have special notation to work with them.
Suppose $\beta$ is a function and $u_{ab}$ is symmetric and trace-free
with respect to $g_{ab}$,  We define
\begin{equation}\label{eq:ttf}
(g_{ab};\; u_{ab},\beta) = u_{ab} + \frac{2}{n}\beta\; g_{ab} \in T_g\calM.
\end{equation}
It is easy to see that any $h_{ab}\in T_g\calM$ admits a unique decomposition of the form \eqref{eq:ttf}.
Similarly, if $f$ is a function and $A^{ab}$ is symmetric and trace-free with respect to $g_{ab}$
we define
\begin{equation}\label{eq:ttfstar}
(g_{ab};\; A^{ab},f)^* = (A^{ab} + \frac{1}{2} f g^{ab})\omega_{g}\in T^*_g\calM.
\end{equation}
Note that
\begin{equation}\label{eq:Mduality}
\ip< (g_{ab};\; A^{ab}, f)^*, (g_{ab};\; u_{ab}, \beta) > = \int A^{ab}u_{ab} + f\beta\; \omega_g.
\end{equation}
It is sometimes convenient to work with elements of $T^*_g\calM$ represented by covariant
tensors, so
if $B_{ab}$ is symmetric and trace-free with respect to $g_{ab}$ we define
\begin{equation}\label{eq:ttfstar-alt}
(g_{ab};\; B_{ab},f)^* = (g_{ab}; \; g^{ac}g^{bd}B_{bd},f)^*.
\end{equation}

\subsubsection{The space \texorpdfstring{$\calM/\calD_0$}{M/D0} of geometries}

Let $\calD_0$ be connected component of the identity $e$ in the group of smooth diffeomorphisms from $M$ to $M$.
Then $\calM/\calD_0$ is the set of equivalence classes of metrics where $g_{ab}$
is related to $\widehat g_{ab}$ if there exists $\Phi\in\calD_0$ with  $\widehat g_{ab} = \Phi^* g_{ab}$.
We write $\{g_{ab}\}$ for the equivalence class of $g_{ab}$ in $\calM/\calD_0$.

Recall that $\calD_0$ is a Fr\'echet manifold and $T_e \calD_0 = C^\infty( M, TM)$ \cite{Kriegl:1997hf}.
Suppose $\Phi_t$ is a path of diffeomorphisms with $\Phi_0=e$, and let $X^a$ be its infinitesimal
generator.  Given a metric $g_{ab}$, the path of metrics $\gamma_{ab}(t) = \Phi^*_t g_{ab}$ 
remains in $\{g_{ab}\}$ and satisfies
\begin{equation}
\gamma_{ab}'(0) = \Lie_g X_{ab} =  \nabla_a X_b + \nabla_b X_a
\end{equation}
where $\nabla$ is the Levi-Civita connection of $g_{ab}$.  Since $\gamma_{ab}$
is stationary in $\calM/\calD_0$, the directions $\Im\Lie_g\subseteq T_g\cal M$
become null directions in $\calM/\calD_0$, which motivates the formal definition
\begin{equation}
T_{g} \calM/\calD_0 = T_{g}\calM / \Im \Lie_g.
\end{equation}
By working formally and infinitesimally, we avoid details concerning the 
structure of $\calM/\calD_0$ as a stratified space.
However, one can often think of $T_{g} \calM/\calD_0$ as a proxy for an actual tangent space
$T_{\{g\}} \calM/\calD_0$ that we have not defined \cite{Fischer:1996gg}.

Let $(g_{ab};\; u_{ab},\beta)\in T_g\calM$. We continue the practice of denoting quotients
by $\calD_0$ using curly braces and define
\begin{equation}
\{g_{ab};\; u_{ab},\beta\} = (g_{ab};\; u_{ab},\beta) + \Im \Lie_g \in T_g \calM/\calD_0.
\end{equation}
It is helpful to think of the projection
\begin{equation}\label{eq:MtoMD0push}
(g_{ab};\; u_{ab},\beta) \mapsto \{g_{ab};\; u_{ab},\beta\}
\end{equation}
as the pushforward from $T_g\calM$ to $T_g \calM/\calD_0$.

The \define{conformal Killing operator}
of a metric $g_{ab}$ acts on vector fields $X^a$ by
\begin{equation}
\ck_g X_{ab} = \Lie_g X_{ab} -\frac{2}{n} \div_g X
\end{equation}
where $\div_g X = \nabla_a X^a$.  An element of the kernel
of $\ck_g$ is a \define{conformal Killing field}.
Note that in trace/trace-free notation
\begin{equation}
\Lie_g X_{ab} = (g_{ab};\; \ck_{g} X_{ab}, \div_g X)
\end{equation}
and hence
\begin{equation}
\dc<g_{ab}, \ck_g X_{ab}, \div_g X> = 0.
\end{equation}

Since $T_g\calM/\calD_0$ is a quotient of $T_g\calM$ by $\Im \Lie_g$
we formally define 
\begin{equation}
T^*_g \calM/\calD_0 = (\Im \Lie_g)^\perp = \left\{ A\in T^*_g \calM: A|_{\Im \Lie_g} = 0\right\}.
\end{equation}
Consequently, $T^*_g \calM/\calD_0\subseteq T^*_g \calM$ and
an integration by parts exercise
shows that 
$F^{ab}\omega_g\in T_g^*\calM/\calD_0$ if and only if $(\div_g F)^a = \nabla_{a} F^{ab}= 0$.
If $(g_{ab};\; A^{ab},f)^*\in T_g\calM$, then the divergence-free condition is
\begin{equation}\label{eq:divfreettf}
\nabla_a A^{ab} +\frac{1}{2}\nabla^b f = 0
\end{equation}
and we write
\begin{equation}
\{g_{ab};\; A^{ab},f\}^*
\end{equation}
for any $(g_{ab};\;A^{ab},f)^*\in T_g^*\calM$ that satisfies 
equation \eqref{eq:divfreettf}.
Elements $\mathbf F\in T^*_g \calM/\calD_0$
are functionals on $T_g\calM/\calD_0$ according to to the rule
\begin{equation}\label{eq:MD0duality}
\ip< \mathbf F, \{\mathbf h\} > = \ip<\mathbf F, \mathbf h>,
\end{equation}
and we see that the natural embedding $T^*_g\calM/\calD_0\hookrightarrow T^*_g\calM$
is the pullback associated with the pushforward \eqref{eq:MtoMD0push}.

\section{The ADM Lagrangian with densitized lapse}\label{sec:AMD}

In the traditional approach to the ADM $n+1$ decomposition of general relativity, 
on each slice of constant coordinate time we select
a positive function $N$ (the lapse) and a vector field $X^a$ (the shift) that
describe the layout of a coordinate system in spacetime.  
A metric and second fundamental form $(g_{ab}, K_{ab})\in \calM\times\calK$ 
determine the ADM momentum
\begin{equation}
\Pi^{ab} = \left[K^{ab}- \tr_g K g^{ab}\right] \omega_g \in T^*_g\calM
\end{equation}
and also determine, in conjunction with the lapse and shift, the ADM velocity
\begin{equation}
\dot g_{ab}  = 2NK_{ab} + \Lie_g X_{ab} \in T_g\calM.
\end{equation}
From these maps we obtain an isomorphism $i_{N,X^a}:T\calM\ra T^*\calM$
\begin{equation}
\begin{gathered}
\xymatrix{
  &\calM \times \calK \ar@{<->}[dr] \ar@{<->}[dl]_{(N,X^a)} \\
 T\;\calM 
\ar@<2pt>@{<-}[rr]
\ar@<-2pt>@{->}[rr]_{i_{N,X^a} }
 & & T^*\; \calM 
}
\end{gathered}
\end{equation}
where the notation $(N,X^a)$ denotes a nameless function that depends on the lapse and shift.
Recalling the definition of a Legendre transformation from, e.g., \cite{Mardsen:1999ab}
that relates velocities and momenta, the map $i_{N,X^a}$ is
the Legendre transformation associated with the ADM Lagrangian
\begin{equation}\label{eq:ADM}
L_{\text{ADM}}(g_{ab}, \dot g_{ab};\; N, X^a)
= \int_{M} N\left( R_g +\frac{1}{4N^2}|\dot g - \Lie_g X|^2_g - \frac{1}{4N^2} (\tr_g \dot g-2\div_g X)^2\;\right) \omega_g.
\end{equation}

Writing $\dot g_{ab} = (g_{ab};\; u_{ab}, \beta)$ in trace/trace-free notation we have
\begin{equation}\label{eq:ADMlag}
L_{\text{ADM}}(g_{ab}, u_{ab}, \beta;\; N, X^a) = 
\int_{M} N\left( R_g +\frac{1}{4N^2}|u - \ck_g X|^2_g - \frac{\kappa}{N^2} (\beta-\div_g X)^2\;\right) \omega_g
\end{equation}
and the Legendre transformation can be written
\begin{equation}\label{eq:non-dl-legendre}
i_{N,X^a}( g_{ab};\; u_{ab}, \beta ) = \left(g_{ab};\;  \frac{1}{2N}\left(u_{ab} -\ck_g X_{ab}\right), -2\kappa 
\frac{1}{N}\left(\beta-\div_g X\right)\right)^*
\end{equation}
with inverse
\begin{equation}\label{eq:i_inv}
i_{N,X^a}^{-1}( (g_{ab};\; A_{ab}, f)^* ) = \left(g_{ab};\;2NA_{ab} +\ck_g X_{ab}, -\frac{N}{2\kappa}f +\div_g X\right).
\end{equation}
It will also be helpful to have trace/trace-free expressions for the conversion from a second fundamental form
to ADM velocity or momentum.
If $K_{ab}=A_{ab}+(\tau/n)g_{ab}$ where $A_{ab}$ is trace-free,
then the ADM velocity is
\begin{equation}
(g_{ab};\;2NA_{ab}+\ck_{g}X_{ab},N\tau+\div_g X)
\end{equation}
and the ADM momentum is
\begin{equation}\label{eq:ADMmomTTF}
(g_{ab};\;A_{ab},-2\kappa\tau)^*.
\end{equation}

We will work with a modified form of the ADM Lagrangian where the lapse
is prescribed indirectly using a construct called a densitized lapse.  Densitized lapses
appear in various contexts in general relativity \cite{ChoquetBruhat:1983wu}
\cite{Anderson:1998er} \cite{Sarbach:2002ia} and were introduced to the constraint equations
in York's conformal thin sandwich method \cite{YorkJr:1999jo}. As mentioned in the
introduction, we will represent a densitized lapse by a choice of volume form $\alpha$,
and given a metric $g_{ab}$, the lapse associated with $g_{ab}$ and $\alpha$ is
\begin{equation}\label{eq:lapsedensity}
N_{g,\alpha} = \frac{\omega_g}{\alpha}.
\end{equation}
Note that if $\widehat g_{ab} = \phi^{q-2} g_{ab}$ for some
conformal factor $\phi$ then
\begin{equation}
N_{\widehat g, \alpha} = \phi^q N_{g,\alpha}.
\end{equation}
We will call $\alpha$ a \define{lapse form}.

Replacing $N$ with $N_{g,\alpha}$ and using equation \eqref{eq:lapsedensity} to
rewrite $\omega_g$ in terms of $\alpha$, the Lagrangian \eqref{eq:ADMlag} becomes
\begin{equation}\label{eq:dlAMD}
L_{\text{ADM}'}(g_{ab}, u_{ab}, \beta;\; \alpha, X^a) = 
\int_{M}\left( N^2_{g,\alpha} R_g +\frac{1}{4}|u - \ck_g X|^2_g - \kappa (\beta-\div_g X)^2\;\right) \alpha.
\end{equation}
For
the remainder of this paper we work with the densitized-lapse ADM Lagrangian \eqref{eq:dlAMD}.
The Legendre transformation associated with this Lagrangian is 
the standard transformation \eqref{eq:non-dl-legendre} with the substitution
$N=N_{g,\alpha}$
% :
% \begin{equation}
% i_{N,X^a}( g_{ab}; u_{ab}, \beta ) = \frac{1}{2N_{g,\alpha}}(g_{ab};\;u_{ab} -\ck_g X_{ab}, -4\kappa (\beta-\div_g X))^*
% \end{equation}
and we have the commutative diagram
\begin{equation}\label{diag:dl-legendre}
\begin{gathered}\xymatrix{
 &\calM \times \calK \ar@{<->}[dr] \ar@{<->}[dl]_{(\alpha,X^a)} \\
 T\;\calM 
\ar@<2pt>@{<-}[rr]
\ar@<-2pt>@{->}[rr]_{i_{\alpha,X^a} }
& & T^*\; \calM.
}
\end{gathered}
\end{equation}
The distinction between the standard and densitized-lapse Legendre transformations is perhaps
subtle.  Given a metric $g_{ab}$ and a lapse form $\alpha$, there always exists a lapse $N$
such that the maps $i_{N,X^a}$ and $i_{\alpha,X^a}$ agree as maps from
$T_g \calM$ to $T_g^* \calM$.  But if we consider a second metric $\widehat g_{ab}$
with volume form $\omega_{\widehat g}$ different from $\omega_g$, then
the two Legendre transformations as maps from $T_{\widehat g} \calM$ to $T_{\widehat g}^* \calM$
are no longer the same. This difference is important when 
thinking of the Legendre transformation as
a map between the total bundles $T\calM$ and $T^*\calM$, and we will find that
the densitized lapse is particularly compatible
with the  product structure $\calM=\calC\times\calV$.  For example, 
given a lapse form $\alpha$, we will be able
to assign a notion of conformal velocity, momentum, and kinetic energy
measured by $\alpha$ to each curve in $\calM$ 
in such a way that if two curves in $\calM$ descend to the
same curve in $\calC$ or even $\calC/\calD_0$, then these conformal quantities are
preserved.    
The next several sections make these ideas precise, and we start by 
recalling the definitions of tangent and cotangent spaces of $\calC$ and $\calC/\calD_0$
from \cite{Maxwell:2014a}.

\section{Conformal Tangent and Cotangent Spaces}\label{sec:conformalTS}

If $g_{ab}\in\calM$, we write $[g_{ab}]$ for its conformal class,
and we use bold type
to denote a conformal class when we do not wish to emphasize any particular representative.
So $[g_{ab}] = \mathbf g \in \calC$ if and only if $g_{ab}\in\mathbf g$. 
By convention we use conformal transformations of the form
$\widehat g_{ab} =  \phi^{q-2}g_{ab}$ since the exponent $q-2$
leads to a simple conformal transformation law for scalar curvature: 
\begin{equation}
R_{\widehat g} = \phi^{1-q}( -\dimk\Lap_g \phi + R_g\phi).
\end{equation}

Following \cite{Maxwell:2014a}, if $\mathbf g\in \calC$ we define 
$T_{\mathbf g} \calC$ to be the set of equivalence classes of
pairs $(g_{ab};\; u_{ab})$ where 
$g_{ab}\in\mathbf g$, $u_{ab}$ is trace-free with respect to $g_{ab}$,
and where
\begin{equation}\label{eq:CtangentER}
(g_{ab};\; u_{ab}) \sim  (\phi^{q-2} g_{ab};\; \phi^{q-2} u_{ab}).
\end{equation}
The trace-free condition arises because we identify $\calC$
with the set of metrics with a fixed volume form,
and the equivalence relation reflects the arbitrary 
choice of volume form.  We will write
\begin{equation}
[g_{ab};\; u_{ab}]
\end{equation}
for the element of $T_\mathbf g \calC$ determined by $(g_{ab};\; u_{ab})$, and
we will write
$\mathbf{u}$ for a conformal tangent vector when we do not wish
to emphasize a particular representative. At $\mathbf g\in \calC$,
we define the conformal killing operator $\ck_\mathbf g$ acting
on a vector field $X^a$
\begin{equation}
\ck_{\mathbf g} X^a = [g_{ab};\; \ck_g X_{ab}]
\end{equation}
where $g_{ab}$ is any representative of $\mathbf g$; the conformal
transformation law $\ck_{\widehat g} = \phi^{q-2} \ck_g$ if
$\widehat g_{ab} =\phi^{q-2} g_{ab}$ ensures that $\ck_\mathbf g$ is well-defined.

The cotangent space $T^*_{\mathbf g}\calC$ is also defined as 
a set of equivalence classes of pairs $(g_{ab};\;A_{ab})$ where
$g_{ab}\in \mathbf g$, $g^{ab}A_{ab}=0$, but the equivalence
relation differs. Now 
\begin{equation}
(g_{ab};\; A_{ab}) \sim (\phi^{q-2} g_{ab};\; \phi^{-2} A_{ab}).
\end{equation}
and we write
\begin{equation}
[g_{ab};\; A_{ab}]^*
\end{equation}
for the element of $T^*_{[g]} \calC$ determined by $(g_{ab};\; A_{ab})$.
As before, we use bold face when no representative is preferred.  
If $\mathbf u\in T_\mathbf g \calC$
and $\mathbf A\in T^*_\mathbf g\calC$ we define
\begin{equation}\label{eq:Cact}
\ip<\mathbf A, \mathbf u> = \int_M \ip<A,u>_g\;\omega_g
\end{equation}
where $g_{ab}$ is any representative of $\mathbf g$ and where 
$u_{ab}$ and $A_{ab}$ are the representatives such that
\begin{equation}
\mathbf A = [g_{ab};\; A_{ab}]^* \quad\text{and}\quad
\mathbf u = [g_{ab};\; u_{ab}].
\end{equation}
The equivalence relations for conformal tangent and cotangent vectors
ensures that this action is well defined.

We have a natural map from $T_g\calM$ to $T_{[g]}\calC$ given by
\begin{equation}\label{eq:MCpush}
(g_{ab};\; u_{ab},\beta) \mapsto [g_{ab};\; u_{ab}]
\end{equation}
that can be thought of as the pushforward.  From
equation \eqref{eq:Cact} we have the corresponding 
pullback $T^*_{[g]}\calC \ra T^*_g \calM$ which can be written in the notation
of equation \eqref{eq:ttfstar-alt} as
\begin{equation}\label{eq:MCpull}
\mathbf A \mapsto (g_{ab};\; A_{ab},0)^*
\end{equation}
if $\mathbf A = [g_{ab};\; A_{ab}]^*$.

Sitting below the space of
conformal classes is the space $\calC/\calD_0$ 
of conformal geometries.  Two conformal classes $\mathbf{g}$ and
$\widehat{\mathbf{g}}$ are equivalent at the level of conformal geometries 
if there is a diffeomorphism $\Phi\in\calD_0$ such that
$\Phi^* \mathbf g = \widehat{\mathbf g}$.  Concretely, two
metrics $g_{ab}$ and $\widehat g_{ab}$ determine the same conformal geometry if there
is a diffeomorphism $\Phi\in\calD_0$
and a smooth positive function $\phi$ such that 
$\widehat g_{ab} = \phi^{q-2} \Phi^* g_{ab}$. We write $\{ \mathbf g\}$ for the conformal geometry 
determined by the conformal class $\mathbf g$. 

In defining the tangent spaces to $\calM/\calD_0$ we quotiented by 
the directions $\Im \Lie_g$.  The pushforward of $\Im \Lie_g$
into $T_{[g]}\calC$ is $\Im \ck_{[g]}$ and we therefore formally define
for any $\mathbf g\in\calC$
\begin{equation}
T_{\mathbf g} \calC/\calD_0 = (T_{\mathbf g} \calC) / \Im \ck_{\mathbf g}.
\end{equation}
We write 
$\{\mathbf u\}$ for the equivalence class $\mathbf u + \Im \ck_\mathbf g\in T_{\mathbf g}\;\calC/\calD_0$.
Correspondingly, we define 
\begin{equation}
T^*_{\mathbf g} \calC/\calD_0= (\Im \ck_\mathbf g) ^\perp =\{ \mathbf A\in T_\mathbf g^*: \mathbf A|_{\Im \ck_\mathbf g}=0\}
\end{equation}
and elements $\bfsigma\in T^*\calC/\calD_0$ are then well-defined functionals on $T\,\calC/\calD_0$
according to
\begin{equation}\label{eq:calCdual}
\ip<\bfsigma, \{\mathbf u\}> = \ip<\bfsigma, \mathbf u>.
\end{equation}
An integration by parts exercise shows that $[g_{ab};\; \sigma_{ab}]^*\in T^*_{\mathbf g} \calC$
if and only if $\sigma_{ab}$ is divergence-free with respect to $g_{ab}$.
Since $\sigma_{ab}$
is trace-free as well, it is a so-called transverse traceless (TT) tensor.
We will
write $\{g_{ab};\;\sigma_{ab}\}^*$ if we wish to emphasize that 
$[g_{ab};\;\sigma_{ab}]^*$ belongs to $T^*\calC/\calD_0$.

We define the pushforward $T_{\mathbf g}\calC\ra T_{\mathbf g}\calC/\calD_0$ by
\begin{equation}\label{eq:calCpush}
\mathbf u \mapsto \{\mathbf u\} = \mathbf u + \Im \ck_\mathbf g.
\end{equation}
Its corresponding pullback is the natural embedding $T_{\mathbf g}^* \calC/\calD_0 \hookrightarrow
T_{\mathbf g}^* \calC$.

\section{Conformal Velocity, Momentum and Kinetic Energy}\label{sec:conformalLegendre}

Let $\mathbf g$ be a conformal class and let $g_{ab}$ be
any representative.
Given a lapse form $\alpha$ and a shift $X^a$ we can combine diagram \eqref{diag:dl-legendre} with
the pushforward and pullback maps described in the previous section to obtain
the following diagram:
\begin{equation}\label{diag:dl-legendre-extended}
\begin{gathered}
\xymatrix{
 &\calK \ar@{<->}[dr] \ar@{<->}[dl]_{(\alpha,X^a)} \\
 T_{g}\;\calM 
\ar@<2pt>@{<-}[rr]
\ar@<-2pt>@{->}[rr]_{i_{\alpha,X^a} }
& & T^*_{g}\; \calM\\
T_{\mathbf g}\calC/\calD_0 \ar@{<-}[u] & & T^*_{\mathbf g}\calC/\calD_0.\ar[u]
}
\end{gathered}
\end{equation}
The principal goal of this section is to show that the Legendre
transformation $i_{\alpha,X^a}$ descends to an isomorphism $j^\expspace\calC_{\alpha}:T_\mathbf g \calC/\calD_0\ra T^*_\mathbf g \calC/\calD_0$ that such that for every $g_{ab}\in\mathbf g$, the diagram
\begin{equation}\label{diag:dl-legendre-conf}
\begin{gathered}
\xymatrix{
 & \calK \ar@{<->}[dr] \ar@{<->}[dl]_{(\alpha,X^a)} \\
 T_g\calM 
\ar@<2pt>@{<-}[rr]
\ar@<-2pt>@{->}[rr]_{i_{\alpha,X^a} }
& & T^*_g \calM\\
T_\mathbf g \calC/\calD_0 \ar@{<-}[u] 
\ar@<2pt>@{<-}[rr]
\ar@<-2pt>@{->}[rr]_{j^\expspace\calC_{\alpha} }
& & T^*_\mathbf g \calC/\calD_0
% \ar@<4pt>[u]
\ar@<0pt>[u]
% \ar@<-4pt>[u]
}
\end{gathered}
\end{equation}
nearly commutes. The failure of commutativity comes from the fact that the projection
$T_{g}\calM\ra T_{\mathbf g} \calC/\calD_0$ loses information that
cannot be recovered.  Instead, we will see that traversing the lower loop of diagram
\eqref{diag:dl-legendre-conf} when starting from the middle row is a projection.

With the isomorphism $j_\alpha^\expspace\calC$ in hand, we assign to a 
pair $(g_{ab},K_{ab})\in\calM\times\calK$ a conformal velocity and momentum
as follows.
The conformal velocity is obtained by forming diagram \eqref{diag:dl-legendre-conf}
for $g_{ab}$ and  then mapping $K_{ab}$ through the left-hand side of the diagram 
starting from $\calK$ to obtain a velocity in $T_{[g]}\calC/\calD_0$.  Note that this is a velocity
modulo diffeomorphsims, and strictly speaking it is a `conformal geometric velocity'.
The conformal momentum is constructed from
the conformal velocity by applying $j^\expspace\calC_\alpha$.  In
this sense, $j^\expspace\calC_\alpha$ behaves like a Legendre transformation, and 
we show in Proposition \ref{prop:confleg}
that it arises from a Lagrangian on $T\;\calC/\calM$ that we will call conformal kinetic energy.

To construct $j_\alpha^\expspace\calC$ it turns out that it is easiest to construct $(j^\expspace\calC_{\alpha})^{-1}$ first.
Diagram \eqref{diag:dl-legendre-extended} defines a map from $T^*_\mathbf g\calC/\calD_0$
to $T_\mathbf g\calC/\calD_0$  given by traveling from the lower-right corner to the 
lower left corner.  In principle this map depends on $\alpha$, $X^a$, and the choice
$g_{ab}\in\mathbf g$, and we provisionally call this map $j^{-1}_{\alpha, X^a, g}$.
The first order of business is to show that this map is independent of $X^a$ 
(because we are reducing to a quotient modulo $\calD_0$) and $g_{ab}$
(because we are using a densitized lapse)
to obtain a map $j^{-1}_\alpha$.  We then show that $j^{-1}_\alpha$ is, as the notation suggests,
the inverse of a map $j_\alpha^\expspace\calC$.

\begin{lemma}\label{lem:jinv}  
Let $\mathbf g\in\calC$ and let $\alpha$ be a fixed lapse form.
For any pair of shifts $X^a$ and $\widehat X^a$, and any pair of representatives
$g_{ab}$ and $\widehat g_{ab}$ of $\mathbf g$,
\begin{equation}
j^{-1}_{\alpha,X^a,g} = j^{-1}_{\alpha,\widehat X,\widehat g}
\end{equation}
and we call the common map $j^{-1}_\alpha$.  
Moreover, for all $\bfsigma\in T_{\mathbf g}^*\calC/\calM$,
\begin{equation}\label{eq:jinv}
j^{-1}_{\alpha}(\bfsigma) = \{g_{ab};\;2N_{g,\alpha}\sigma_{ab}\}
\end{equation}
where $\sigma_{ab}$ is the representative of $\bfsigma$ with respect
to $g_{ab}$.
\end{lemma}
\begin{proof}
Let $\bfsigma\in T_{\mathbf g}^*\calC/\calD_0$. 
To compute $j_{\alpha, X^a,g}^{-1}(\bfsigma)$, let $\sigma_{ab}$ be the representative
of $\bfsigma$ with respect to $g_{ab}$.  
From equation \eqref{eq:MCpull} the pullback of
$\bfsigma$ to $T_g\calM$ is $(g_{ab};\;\sigma_{ab},0)^*$. Applying $i^{-1}_{\alpha,X^a}$
from equation \eqref{eq:i_inv} we arrive at $(g_{ab};\; 2N_{\alpha,g}\sigma_{ab}+\ck_g X_{ab}, 0)$.
Finally, we apply the pushforward from equation \eqref{eq:MCpush} to conclude
\begin{equation}\label{eq:jinvbody}
j^{-1}_{\alpha,X^a,g}(\bfsigma) =  \{g_{ab};\; 2N_{\alpha,g}\sigma_{ab} + \ck_gX_{ab}\} = 
\{g_{ab};\; 2N_{\alpha,g}\sigma_{ab} \}
\end{equation}
since $\{g_{ab};\;\ck_g X_{ab}\} = 0$.

At this stage, it is clear that $j^{-1}_{\alpha, g, X^a}$ is independent of the shift $X^a$.
Now suppose $\widehat g_{ab}$ is another representative of $\mathbf g$ with 
$\widehat g_{ab} = \phi^{q-2}g_{ab}$ for some conformal factor $\phi$.  The representative
of $\bfsigma$ with respect to $\widehat g_{ab}$ is $\widehat\sigma_{ab} = \phi^{-2}\sigma_{ab}$
and we have $N_{\alpha,\widehat g} = \phi^q N_{\alpha, g}$.  Recalling equation
\eqref{eq:CtangentER} we find
\begin{equation}
\begin{aligned}
j^{-1}_{\alpha,\widehat g,X^a}( \bfsigma ) &= \{\widehat g_{ab};\; 2N_{\alpha,\widehat g}\widehat \sigma_{ab}\} \\
&= \{\phi^{q-2} g_{ab};\; \phi^{q-2} 2N_{\alpha, g} \sigma_{ab}\} \\
&= \{ g_{ab};\; 2N_{\alpha, g} \sigma_{ab}\} \\
&= j^{-1}_{\alpha, g, X^a}( \bfsigma ).
\end{aligned}
\end{equation}
Hence $j^{-1}_{\alpha,X^a,g}=j^{-1}_{\alpha,X^a,\widehat g}$ as claimed, and equation \eqref{eq:jinv}
follows from equation \eqref{eq:jinvbody}.
\end{proof}

To show $j_\alpha^{-1}$ is the inverse of a function $j_\alpha^\expspace\calC$ we require
York splitting \cite{York:1973fla}, which we use in the following form.
\begin{proposition}[York Splitting]\label{prop:yorksplit}
Let $g_{ab}\in \calM$ and let $N>0$ be a lapse.

If $A_{ab}$ is symmetric and trace-free, then
there is a $g_{ab}$-TT tensor $\sigma_{ab}$ 
and a vector field $W^a$ 
such that
\begin{equation}\label{eq:yorksplitm}
A_{ab} =\sigma_{ab}+\frac{1}{2N} \ck_g W_{ab}.
\end{equation}
This decomposition is unique up to addition of a conformal Killing field to $W^a$.

Equivalently, if $u_{ab}$ is symmetric and trace-free, there is a unique
$g_{ab}$-TT tensor $\sigma_{ab}$ and a vector field $W^a$, unique up to addition of a
conformal Killing field, such that
\begin{equation}
u_{ab} = 2N\sigma_{ab} + \ck_g W_{ab}.
\end{equation}
\end{proposition}
When $N\equiv 1/2$, Proposition \ref{prop:yorksplit} is classical York splitting,
and the result for arbitrary lapses is equivalent to classical York splitting \cite{Maxwell:2014a};
see also \cite{Pfeiffer:2003ka}.  The decomposition from Proposition \ref{prop:yorksplit}
defines a projection from symmetric trace-free tensors to transverse-traceless tensors and we introduce
the following notation for it.

\begin{definition}
Let $g_{ab}$ be a metric and let $\alpha$ be a lapse form.
Given a symmetric, trace-free tensor $A_{ab}$, its \define{York projection}
is
\begin{equation}
Y_{g,\alpha}(A_{ab}) = \sigma_{ab}
\end{equation}
where $\sigma_{ab}$ is the unique $g_{ab}$-TT tensor such that equation
\eqref{eq:yorksplitm} holds with $N=N_{g,\alpha}$.
\end{definition}

We now show that the formal notation $j_\alpha^{-1}$ is justified by constructing an inverse
$j_\alpha^\expspace\calC$.
\begin{lemma}\label{lem:jinv2}
For all $\mathbf g\in\calC$, $j^{-1}_\alpha:T^*_{\mathbf g}\calC/\calD_0 \ra T_{\mathbf g}\calC/\calD_0$ 
is a linear isomorphism.  
If $\{\mathbf u\}\in T_{\mathbf g} \calC/\calD_0$ then  $j_\alpha^\expspace\calC(\{\mathbf u\})$ is computed as follows.  
Pick any $g_{ab}\in\mathbf g$ and pick any $u_{ab}$ such that
\begin{equation}
\{\mathbf u\} = \{g_{ab};\; u_{ab}\}.
\end{equation}
Let $\sigma_{ab} = Y_{g,\alpha}(1/(2N_{g,\alpha}) u_{ab})$, so $\sigma_{ab}$ is
the unique $g_{ab}$-TT tensor such that
\begin{equation}
u_{ab} = 2N_{g,\alpha} \sigma_{ab} + \ck_g W_{ab}
\end{equation}
for some vector field $W^a$.  Then
\begin{equation}\label{eq:finv}
j_\alpha^\expspace\calC( \{\mathbf u\}) = \{g_{ab};\; \sigma_{ab} \}^*.
\end{equation}
\end{lemma}
\begin{proof}
To see that $j_\alpha^{-1}$ is injective, suppose $j_\alpha^{-1}(\bfsigma)=0$ for some 
$\bfsigma=\{g_{ab};\;\sigma_{ab}\}^*$.  From Lemma \ref{lem:jinv}
it follows that
\begin{equation}
0=j_{\alpha}^{-1}(\bfsigma) = \{g_{ab};\; 2N_{\alpha,g} \sigma_{ab} \}
\end{equation}
and consequently $2N_{\alpha,g}\sigma_{ab} + \ck_g W_{ab}=0$
for some vector field $W^a$.  But $0$ also admits the decomposition $0=2N_{\alpha,g}0+\ck_g 0$ and
the uniqueness clause of Proposition \ref{prop:yorksplit} implies $\sigma_{ab}=0$. Therefore $\bfsigma=0$.

To show $j_\alpha^{-1}$ is surjective, 
consider $\{\mathbf u\}\in T_{\mathbf g}\calC/\calD_0$. Let $g_{ab}\in\mathbf g$ 
and pick any $u_{ab}$ such that
\begin{equation}
\mathbf u = \{ g_{ab};\; u_{ab}\}.
\end{equation}
Let $\sigma_{ab}$ be the unique $g_{ab}$-transverse traceless tensor given by Proposition
\ref{prop:yorksplit} such that
\begin{equation}
u_{ab} = 2N_{g,\alpha}\sigma_{ab} + \ck_g W_{ab}.
\end{equation}
for some vector field $W^a$. From Lemma \ref{lem:jinv} it follows that
\begin{equation}\label{eq:finvbody}
j^{-1}_\alpha(\{g_{ab};\; \sigma_{ab}\}) = \{g_{ab};\;2N_{\alpha,g} \sigma_{ab} \} = 
\{ g_{ab};\;2N_{\alpha,g} \sigma_{ab} + \ck_g W_{ab} \} =
\{ g_{ab};\;u_{ab}\} =\{\mathbf u\}.
\end{equation}
This establishes the claimed surjectivity, so $j_{\alpha}^{-1}$
has an inverse $j_{\alpha}^\expspace\calC$. Equation \eqref{eq:finv}
follows from applying $j_\alpha^\expspace\calC$ to both sides of equation
\eqref{eq:finvbody}.
\end{proof}

Having constructed the isomorphism $j_\alpha^\expspace\calC$, we
obtain diagram \eqref{diag:dl-legendre-conf}, which
commutes except perhaps when going around the lower loop. 
A straight forward exercise using Lemma \ref{lem:jinv2}
shows that traversing the loop starting at the level of $\calC/\calD_0$
is the identity, but traversing the loop starting at the level of $\calM$ is a projection.
In particular, if we start at $T_{g}^*\calM$, then the projection is
\begin{equation}
(g_{ab};\; A_{ab},f)^* \mapsto (g_{ab};\; Y_{g,\alpha}(A_{ab}),0)^*.
\end{equation}

As mentioned previously, we assign a conformal velocity in $T_{[g]}\calC/\calD_0$ 
to $(g_{ab}, K_{ab})\in\calM\times\calK$ by  descending the left-hand side
of diagram \eqref{diag:dl-legendre-conf}. In principle the velocity depends on both the 
lapse form $\alpha$ and the shift $X^a$, but in fact it is independent of the shift.  
To see this, let $K_{ab}$ be a second-fundamental
form which we write in trace/trace-free form as $K_{ab}=A_{ab}+(\tau/n)g_{ab}$.  Proceeding
down the left-hand side of diagram \eqref{diag:dl-legendre-conf}, 
we first obtain $(g_{ab};\; 2N_{g,\alpha}A_{ab}+\ck_g X_{ab}, N_{g,\alpha}\tau+\div X)\in T_g\calM$
and then arrive at $\{g_{ab};\; 2N_{\alpha,g}A_{ab}+\ck_g X_{ab}\}\in T_{[g]}\calC/\calD_0$.  
But $\{g_{ab};\; \ck_g X_{ab}\}=0$ and therefore
the final result is $\{g_{ab};\; 2N_{\alpha,g}A_{ab}\}$, which is independent of $X^a$.  
\begin{definition}
Let $(g_{ab},K_{ab})\in \calM\times\calK$, and let $\alpha$ be a lapse form.  The \define{conformal velocity}
of $(g_{ab},K_{ab})$, as measured by $\alpha$, is 
\begin{equation}
v^\expspace\calC_\alpha(g_{ab},K_{ab}) = \{g_{ab};\; 2N_{\alpha,g}A_{ab}\}
\end{equation}
where $A_{ab}$ is the $g_{ab}$-trace-free
part of $K_{ab}$.  
\end{definition}
To obtain the corresponding conformal momentum, we convert the velocity to a momentum via $j_\alpha^\expspace\calC$.
Starting with $\{g_{ab};\; 2N_{\alpha,g} A_{ab}\}\in T_{[g]}\calC/\calD_0$, 
let $\sigma_{ab}=Y_{g,\alpha}(A_{ab})$ be 
the York projection, so
\begin{equation}
2N_{\alpha,g} A_{ab} = 2N_{\alpha,g}\sigma_{ab} + \ck_g W_{ab}
\end{equation}
for some vector field $W^a$.  Lemma \ref{lem:jinv2} then implies
\begin{equation}
j_\alpha^\expspace\calC(\{g_{ab};\; 2N_{\alpha,g} A_{ab}\}) = \{g_{ab};\;\sigma_{ab}\}^*.
\end{equation}
\begin{definition}
Let $(g_{ab},K_{ab})\in \calM\times\calK$, and let $\alpha$ be a lapse form.  
Let $A_{ab}$ be the  $g_{ab}$-trace-free part of $K_{ab}$
and let $\sigma_{ab}=Y_{g,\alpha}(A_{ab})$ be its York projection.
The \define{conformal momentum} of $(g_{ab},K_{ab})$, as measured by $\alpha$, is 
\begin{equation}
m^\expspace\calC_\alpha(g_{ab},K_{ab}) = \{g_{ab};\; \sigma_{ab}\}^*.
\end{equation}
\end{definition}

From the maps $v^\expspace\calC_\alpha$ and $m^\expspace\calC_\alpha$ we obtain the 
diagram
\begin{equation}\label{diag:conf-vel-mom}
\begin{gathered}
\xymatrix{
& \calM\times \calK \ar[dl]_{v^\expspace\calC_\alpha} \ar[dr]^{m^\expspace\calC_\alpha} & \\
T \;\calC/\calD_0
\ar@<2pt>@{<-}[rr]
\ar@<-2pt>@{->}[rr]_{j_{\alpha}^\expspace\calC }
& &
T^*\calC/\calD_0
}
\end{gathered}
\end{equation}
which should be compared with diagram \eqref{diag:dl-legendre}.
Note in particular that although the ADM momentum is computed without
reference to the lapse or shift, both the conformal
velocity and conformal momentum depend in general the choice of a lapse form.
The CMC solutions of the constraint equations are an exception to this observation,
however. If $(g_{ab},K_{ab})$ is a CMC
solution of the constraint equations, 
then $K_{ab}=\sigma_{ab}+ (\tau_0/n)g_{ab}$ for some
transverse traceless tensor $\sigma_{ab}$ and some constant $\tau_0$.
Hence the York projection
of $K_{ab}$ is $\sigma_{ab}$ regardless of the choice of lapse form.

The map $j^\expspace\calC_\alpha$
appeared previously in \cite{Maxwell:2014a},
where it was denoted $j_\alpha$ and where it was derived from purely geometric considerations.
The approach taken here suggests that $j^\expspace\calC_\alpha$ is 
a Legendre transformation, and our next task is to identify
a Lagrangian on $T\;\calC/\calD_0$ for which $j^\expspace\calC_\alpha$ is
the associated Legendre transformation.

Consider the kinetic energy term of the densitized-lapse ADM Lagrangian:
\begin{equation}
K(g_{ab},u_{ab},\beta; X^a, \alpha) = \int \frac{1}{4}|u-\ck_g X|^2_g -\kappa (\beta-\div_g X)^2\; \alpha.
\end{equation}
The first term on the right-hand side involves the kinetic energy due to conformal deformation.
Let $\sigma_{ab}$ be the $g_{ab}$-TT tensor such that
\begin{equation}
j_{\alpha}^\expspace\calC(\{g_{ab};\;u_{ab}\}) = \{g_{ab};\; \sigma_{ab}\}^*.
\end{equation}
So there is a vector field $W^a$ such that
\begin{equation}\label{eq:cKEsplit}
u_{ab} = 2N_{\alpha,g}\sigma_{ab} + \ck_g (W+X)_{ab}.
\end{equation}
Then
\begin{equation}
\begin{aligned}\label{eq:confke}
\int \frac{1}{4}|u-\ck_g X|^2_g\alpha 
&=  \int N_{\alpha,g}^2|\sigma|_g^2 +\frac{N}{2}\ip<\sigma,\ck_g(W)>_g + \frac{1}{4}|\ck_g(W)|_g^2 \;\alpha \\
&= \int N_{\alpha,g}^2|\sigma|_g^2 +\frac{1}{4}|\ck_g(W)|_g^2 \alpha + \frac{1}{2} \int \ip<\sigma,\ck_g(W)>_g
\; \omega_g\\
&= \int N_{\alpha,g}^2|\sigma|_g^2 +\frac{1}{4}|\ck_g(W)|_g^2 \;\alpha
\end{aligned}
\end{equation}
where we have used the fact that $N_{\alpha,g}\alpha = \omega_g$ as well as the $L^2$-orthogonality of $\sigma_{ab}$ and 
$\ck_g(W)$ with respect to $g_{ab}$. The conformal kinetic energy is the first term on the right-hand side of the final expression of equation \eqref{eq:confke}.

\begin{definition}  Let $\alpha$ be a lapse form.  The \define{conformal kinetic energy}
of $(g_{ab};\; u_{ab},\beta)\in T_{g}\calM$, as measured by $\alpha$, is 
\begin{equation}
K^\expspace\calC_\alpha(g_{ab}, u_{ab}) = \int N_{\alpha,g}^2|\sigma|_g^2\;\alpha
\end{equation}
where $\sigma_{ab}$ is the $g_{ab}$-TT tensor such that
$j_{\alpha}^\expspace\calC(\{g_{ab};\;u_{ab}\}) = \{g_{ab};\; \sigma_{ab}\}^*$.
\end{definition}
The following lemma shows 
conformal kinetic energy 
descends to a well defined function on $T\;\calC/\calD_0$.
\begin{lemma}\label{lem:cke-descends}
Suppose $\{g_{ab};\;u_{ab}\}=\{\widehat g_{ab};\;\widehat u_{ab}\}$. Then 
$K^\expspace\calC_\alpha(g_{ab}, u_{ab})=K^\expspace\calC_\alpha(\widehat g_{ab}, \widehat u_{ab})$.
\end{lemma}
\begin{proof}
Suppose $\{g_{ab},u_{ab}\}=\{\widehat g_{ab},\widehat u_{ab}\}$, and let $\phi$ be
the conformal factor such that $\widehat g_{ab} = \phi^{q-2} g_{ab}$.
Since $\{g_{ab},u_{ab}\}=\{\widehat g_{ab},\widehat u_{ab}\}$, the corresponding conformal momenta
$\{g_{ab};\;\sigma_{ab}\}^*$ and $\{\widehat g_{ab};\; \widehat\sigma_{ab}\}^*$
are the same.  So $\widehat\sigma_{ab}=\phi^{-2}\sigma_{ab}$. 
Since $N_{\widehat g,\alpha} = \omega_{\widehat g}/\alpha = \phi^q \omega_g/\alpha = \phi^q N_{g,\alpha}$
we conclude
\begin{equation}
K^\expspace\calC_{\alpha}(\widehat g_{ab},\widehat u_{ab}) = \int_M N_{\widehat g,\alpha}^2 |\widehat \sigma|_{\widehat g}^2\;\alpha = 
\int_M \phi^{2q}N_{g,\alpha}^2 \phi^{4-2q}|\phi^{-2}\sigma|_g^2\;\alpha = 
\int_M N_{g,\alpha}^2 |\sigma|_g^2\;\alpha = K^\expspace\calC_{\alpha}( g_{ab},u_{ab}).
\end{equation}
\end{proof}

We will use the same notation $K^\expspace\calC_\alpha$ to denote a function on $T\,\calC/\calD_0$
rather than $T\calM$.
Thinking of it as a Lagrangian,
the following proposition shows that $j_\alpha^\expspace\calC$ is its Legendre transformation.

\begin{proposition}\label{prop:confleg}
Suppose $g_{ab}\in \calM$ is a metric and
$u_{ab}(t)$ is
a smooth path of $g_{ab}$-trace-free tensors. Then
\begin{equation}
\left.\frac{d}{dt}\right|_{t=0} K^\expspace\calC_\alpha(\{g_{ab},u_{ab}(t)\}) = 
\ip< j_\alpha^\expspace\calC(\{g_{ab};\; u_{ab}(0)\}), \{g_{ab}, u_{ab}'(0)\}>.
\end{equation}
\end{proposition}
\begin{proof}
For each $t$, let $\sigma_{ab}(t)$ be the transverse-traceless tensor with
\begin{equation}
j_{\alpha}^\expspace\calC( \{g_{ab};\; u_{ab}(t)\}) = \{g_{ab};\; \sigma_{ab}(t)\}^*;
\end{equation}
since $\sigma_{ab}(t)=Y_{g,\alpha}(u_{ab}(t)/(2N_{g,\alpha}))$, the curve $\sigma_{ab}(t)$ is smooth.
For each $t$ let $W^a(t)$ be a vector field such that
\begin{equation}
u_{ab}(t) = 2N_{g,\alpha} \sigma_{ab}(t) + \ck_g W_{ab}(t).
\end{equation}
Then
\begin{equation}
K^\expspace\calC_\alpha(\{g_{ab},u_{ab}(t)\}) = \int_M N_{g,\alpha}^2 |\sigma(t)|_g^2 \; \alpha = 
\int_M N_{g,\alpha} |\sigma(t)|_g^2 \; \omega_g
\end{equation}
and
\begin{equation}
\begin{aligned}
\left.\frac{d}{dt}\right|_{t=0} K^\expspace\calC_\alpha(\{g_{ab},u_{ab}(t)\}) &= 
\int_M 2N_{g,\alpha} \ip<\sigma(0),\sigma'(0)>_g \; \omega_g \\
&= \int_M \ip<\sigma(0),u'(0)-(\ck_g W)'(0)>_g \; \omega_g \\
&= \int_M \ip<\sigma(0),u'(0)>_g \; \omega_g.
\end{aligned}
\end{equation}
since $\sigma_{ab}$ is transverse traceless.  But from equations
\eqref{eq:Cact} and \eqref{eq:calCdual} this last expression is precisely
\begin{equation}
\ip< [g_{ab};\; \sigma_{ab}(0)], \{g_{ab};\; u_{ab}'(0)\}> =
\ip< j_\alpha^\expspace\calC(\{g_{ab},u_{ab}(0)\}), \{g_{ab};\; u_{ab}'(0)\}>.
\end{equation}
\end{proof}

\section{The Conformal Method}\label{sec:confmeth}

As presented in \cite{Maxwell:2014a}, the conformal
method can be understood in terms of the conformal
parameters discussed in the previous section.
We have the following two formulations.

\begin{problem}[Lagrangian Conformal Method]\label{prob:lag}
Let $\mathbf g$ be a conformal class, 
let $\alpha$ be a lapse form, let $\{\mathbf u\}\in T_{\mathbf g}\calC/\calM$ be
a conformal velocity,
and let $\tau$ be a mean curvature.  Find all solutions $(\ol g_{ab}, \ol K_{ab})$
of the vacuum constraint equations \eqref{eq:constraints} such that
\begin{equation}
\begin{aligned}\strut
[\ol g_{ab}] &= \mathbf g\\
v^\expspace\calC_\alpha(\ol g_{ab},\ol K_{ab}) &= \{\mathbf u\} \\
\ol g^{ab}\ol K_{ab}&=\tau.
\end{aligned}
\end{equation}
\end{problem}

\begin{problem}[Hamiltonian Conformal Method]\label{prob:ham}
Let $\mathbf g$ be a conformal class, 
let $\alpha$ be a lapse form, let $\bfsigma\in T_{\mathbf g}^*\calC/\calM$ be
a conformal momentum,
and let $\tau$ be a mean curvature.  Find all solutions $(\ol g_{ab}, \ol K_{ab})$
of the vacuum constraint equations \eqref{eq:constraints} such that
\begin{equation}
\begin{aligned}\strut
[\ol g_{ab}] &= \mathbf g\\
m^\expspace\calC_\alpha(\ol g_{ab},\ol K_{ab}) &= \bfsigma \\
\ol h^{ab}\ol K_{ab}&=\tau.
\end{aligned}
\end{equation}
\end{problem}

The two problems differ only in whether the conformal velocity or momentum is prescribed, and
they are equivalent: $(\ol g_{ab}, \ol K_{ab})$ is a solution
of Problem \ref{prob:lag} for parameters $(\mathbf g, \{\mathbf u\}, \tau, \alpha)$
if and only if it is a solution of Problem \ref{prob:ham} for parameters 
$(\mathbf g, \bfsigma, \tau, \alpha)$ with $\bfsigma = j_{\alpha}^\expspace\calC(\{\mathbf u\})$.

In order to write down PDEs corresponding to these problems we choose representative tensors
of the conformal parameters. In the Hamiltonian case, we can take conformal parameters to
be a metric $g_{ab}$, a transverse traceless tensor $\sigma_{ab}$, a mean curvature $\tau$,
and a lapse $N$.  These prescribe Hamiltonian conformal parameters
\begin{equation}
\begin{aligned}
\mathbf g&=[g_{ab}]\\
\bfsigma &= \{g_{ab};\; \sigma_{ab}\}^*\\
\tau&=\tau\\
\alpha &= \omega_g/N
\end{aligned}
\end{equation}
and the constraint equations become the CTS-H equations
\begin{equation}
\begin{aligned}\label{eq:CTS-H}
-\dimk\Lap_g\phi + R_g\phi - \left|\sigma +\frac{1}{2N} \ck_g W\right|^2_g\phi^{-q-1} + \kappa\tau^2\phi^{q-1} &= 0 \qquad&\text{\small[CTS-H Hamiltonian constraint]} \\
\div_g\left(\frac{1}{2N}\ck_g W\right) -\kappa \phi^q \mathbf d\tau&=0
\qquad&\text{\small[CTS-H momentum constraint]}
\end{aligned}
\end{equation}
which first appeared, in a slightly different form, in \cite{Pfeiffer:2005iz}.
These equations are to be solved for a conformal factor $\phi$ and a vector field $W^a$, and
if a solution exists then
\begin{equation}
\begin{aligned}
\ol g_{ab} &= \phi^{q-2} g_{ab}\\
\ol K_{ab} &= \phi^{-2}\left(\sigma_{ab}+ \frac{1}{2N}\ck_g W_{ab}\right)+\frac{\tau}{n}\ol g_{ab}
\end{aligned}
\end{equation}
solve the vacuum constraint equations.  Note that in York's original
formulation of the conformal method, there are three parameters $(g_{ab}, \sigma_{ab}, \tau)$ and $N$
is implicitly $1/2$.  This is not an essential restriction since one can control 
the lapse form $\alpha$ by moving $g_{ab}$ within its conformal class while suitably
adjusting $\sigma_{ab}$, but the requirement of
tying the conformal class representative to $\alpha$ leads to some inflexibility.  Hence we prefer
the CTS-H equations to those of York's original conformal method.  In the Lagrangian case, 
the parameter $\sigma_{ab}$ is replaced with an arbitrary symmetric, trace-free tensor
$u_{ab}$ which determines a conformal velocity
\begin{equation}
\{\mathbf u\} = \{ g_{ab};\; u_{ab} \}  = [g_{ab};\; u_{ab}] +\Im\ck_{\mathbf g}.
\end{equation}
and the CTS-H equations become the CTS-L equations found in, e.g., \cite{Maxwell:2014a}.

Since the conformal method specifies a conformal velocity or momentum (modulo diffeomorphisms), 
we would like to understand how the mean curvature is related
to volumetric velocity or momentum (modulo diffeomorphisms). 
We have seen that if $K_{ab}$ has trace/trace-free decomposition
$K_{ab}=A_{ab}+(\tau/n)g_{ab}$, then the conformal momentum is obtained from a lapse-dependent
York projection of $A_{ab}$.  It turns out that volumetric momentum is a single number,
and is obtained from an analogous lapse-dependent York-like projection of $\tau$.  
Indeed, there is a way to treat the volumetric degrees of freedom 
in a fashion completely in parallel to the manner in which the conformal method treats
the conformal degrees of freedom, and we described this in the next two sections.

\section{Volumetric Tangent Spaces}\label{sec:volumetric}

The space $\calV$ of volume forms is an open subset of
$C^\infty(M,\Lambda^n M)$, so the tangent space at $\omega\in\calV$ is
\begin{equation}
T_\omega \calV = C^\infty(M,\Lambda^n M).
\end{equation}
% There is a natural identification of $T_\omega \calV$ with $\C^\infty M$,
% however:
% \begin{equation}
% \eta \mapsto \frac{\eta}{\omega}
% \end{equation}
% and we henceforth define $T_\omega M=C^\infty M$ with this identification.

We define
\begin{equation}
T^*_\omega \calV = C^\infty(M)
\end{equation}
and identify $T^*_\omega\calV$ as a subset of $(T_\omega\calV)^*$
by defining the action of $f\in T^*_\omega \calV$ on $\eta\in T_\omega \calV$ by
\begin{equation}
\ip<f,\eta> = \int_M f\;\eta.
\end{equation}

Suppose $\gamma(t)$ is a path of metrics with $\gamma(0)=g_{ab}$
and $\gamma'(0)=(g_{ab};\; u_{ab}, \beta)$.  A standard computation
shows that the associated
path of volume forms $\omega(t)$ satisfies
\begin{equation}
\omega'(0) = \beta \omega_g.
\end{equation}
Hence the pushforward $T_g \calM \ra T_{\omega_g} \calV$
is
\begin{equation}\label{eq:MtoVpush}
(g_{ab};u_{ab},\beta) \mapsto \beta\omega_g.
\end{equation}
To compute the pullback we note
that if $f\in T^*_{\omega_g}$, equation \eqref{eq:Mduality} implies
\begin{equation}
\ip<f, \beta\omega_g> = \int_M f\beta\;\omega_g = \ip<(g_{ab};\;0,f)^*,(g_{ab};\;u_{ab},\beta)>
\end{equation}
and hence the pullback $T^*_{\omega_g}\calV \ra T^*_g\calM$ is
\begin{equation}\label{eq:MtoVpull}
f \mapsto (g_{ab};\;0,f)^* = \frac{f}{2} g^{ab}\omega_g.
\end{equation}

We now consider volume forms modulo diffeomorphisms, $\calV/\calD_0$.
Suppose $\Phi_t$ is a path of diffeomorphisms starting at the identity
with infinitesimal generator $X^a$.  If $\omega$ is a volume form
and $\gamma(t)=\Phi_t^*\omega$, then
\begin{equation}
\gamma'(0) = \Div_\omega(X)
\end{equation}
where the divergence operator $\Div_\omega$
applied to $X^a$ is the Lie derivative $\calL_X \omega$. 
Note that if $g_{ab}$ is a metric then
\begin{equation}
\Div_{\omega_g}(X) = \div_g(X)\; \omega_g.
\end{equation}
Since $\gamma$ is stationary in $\calV/\calD_0$, 
the directions $\Div_{\omega_g} X$ are null
directions in $\calV/\calD_0$ and we make the formal definition
\begin{equation}
T_\omega \calV/\calD_0 = T_\omega \calV / \Im \Div_\omega.
\end{equation}
The space $T_\omega \calV / \Im \Div_\omega$ is much simpler
than its conformal counterpart, and indeed is one dimensional.
\begin{lemma}\label{lem:voldot}
The map $\dot\Vol: T_\omega \calV/\calD_0 \ra \Reals$  given by
\begin{equation}
\dot\Vol( \eta + \Im\Div_\omega) = \int_M \eta
\end{equation}
is well defined and is an isomorphism.
\end{lemma}
\begin{proof}
We claim that if $\eta$ is an $n$-form, then $\int_M\eta=0$ if and only if
$\eta\in \Im\Div_\omega$. 

To see this, let $g_{ab}$ be any metric such that $\omega_g=\omega$.  Now
if $\eta\in\Im\Div_\omega$, then there is a vector field $X^a$ such that
$\eta = \div_g X\;\omega_g$ and hence $\int\eta=0$.

Conversely, suppose $\int_M \eta = 0$. Then $\eta =f \omega_g$ for some zero-mean function.
Since $M$ is connected, there exists a unique zero-mean solution $u$ of $\Delta_g u = f$.
Setting $X^a=\nabla^a u$ we find that $\eta = \div_g (X)\omega_g = \Div_\omega(X)$.

Since the kernel of $\eta\mapsto \int_M\eta$ is $\Im\Div_\omega$, we conclude
that integration descends to a map $\dot\Vol$ on the quotient space 
$T_\omega\calV/ \Im\Div_\omega = T_\omega\calV/\calD_0 $.  And since $\dot\Vol$ is surjective,
the claimed isomorphism holds.
\end{proof}

We will henceforth identify $T_\omega\calV/\calD_0$ with $\Reals$ using $\dot\Vol$.  Note that
with this identification, the pushforward $T_\omega \calV\ra T_\omega\calV/\calD_0$ is
\begin{equation}\label{eq:VtoVD0push}
\eta\mapsto \int_M\eta.
\end{equation}
Since $\Reals$ is its own dual space (acting on itself by multiplication) we
define $T_\omega^*\calV/\calD_0=\Reals$.  The pullback 
$T_\omega^*\calV/\calD_0\ra T_\omega^*\calV$ takes the constant $c\in\Reals$ to
the constant function $c\in C^\infty(M)$ since
\begin{equation}
c \int_M \eta = \int_M c\eta = \ip<c,\eta>.
\end{equation}
Note that the constant functions in $T_\omega^*\calV$ are
the annihilator of $\Im\Div_\omega$, and hence we could
have equivalently defined $T_\omega^*\calV/\calD_0 = (\Im \Div_\omega)^\perp$
in an approach analogous to that of Section \ref{sec:conformalTS}.

It will be helpful to have notation for the composite pushforward
$T_g\calM\ra T_{\omega_g}\calV/\calD_0$. If $g_{ab}\in\calM$
and $\beta\in C^\infty(M)$ we define
\begin{equation}\label{eq:MtoVD0push}
\{g_{ab};\; \beta\} = \int_M \beta \omega_g.
\end{equation}
From composition we obtain the following
pushforwards and pullbacks
associated with the projection $\calM\ra \calV/\calD_0$.
\begin{lemma}\label{lem:volpushpull}
The pushforward $T_g\calM\ra T_{\omega_g}\calV/\calD_0$ is the map
\begin{equation}\label{eq:volpush}
(g_{ab};\;u_{ab},\beta) \mapsto \{g_{ab};\; \beta\}.
\end{equation}
The pullback $T_{\omega_g}^*\calV/\calD_0\ra T_g\calM^*$ is
\begin{equation}\label{eq:volpull}
c \mapsto (g_{ab};\; 0,c)^*.
\end{equation}
\end{lemma}
\begin{proof}
Equation \eqref{eq:volpush} is a consequence of equations
\eqref{eq:MtoVpull}, \eqref{eq:MtoVpush} and \eqref{eq:MtoVD0push}.
Equation \eqref{eq:volpull} follows from the formula $c\mapsto c$
for the pullback $T^*_\omega \calV/\calD_0\ra T^*_\omega \calV$
and equation \eqref{eq:MtoVpull}.
\end{proof}

\section{Volumetric Velocity, Momentum, and Kinetic Energy}\label{sec:volumetricLegendre}

Let $\omega$ be a volume form, and let $g_{ab}$
be any metric with $\omega_g=\omega$.
Starting from the diagram \eqref{diag:dl-legendre} and the
pushforward/pullback maps from Lemma \ref{lem:volpushpull} we have the following diagram:
\begin{equation}\label{diag:dl-vol-extended}
\begin{gathered}
\xymatrix{
 &\calK \ar@{<->}[dr] \ar@{<->}[dl]_{(\alpha,X^a)} \\
 T_{g}\;\calM 
\ar@<2pt>@{<-}[rr]
\ar@<-2pt>@{->}[rr]_{i_{\alpha,X^a} }
& & T^*_{g}\; \calM\\
T_{\omega}\calV/\calD_0 \ar@{<-}[u] & & T^*_{\omega}\calV/\calD_0.\ar[u]
}
\end{gathered}
\end{equation}
We wish to construct an isomorphism $j^\expspace\calV_\alpha:T_{\omega}\calV/\calD_0\ra T_{\omega}^*\calV/\calD_0$,
analogous to $j^\expspace\calC_\alpha$, such that for every metric $g_{ab}$ with $\omega_g=\omega$, the diagram 
\begin{equation}\label{diag:dl-legendre-conf-vol}
\begin{gathered}
\xymatrix{
 & \calK \ar@{<->}[dr] \ar@{<->}[dl]_{(\alpha,X^a)} \\
 T_g\calM 
\ar@<2pt>@{<-}[rr]
\ar@<-2pt>@{->}[rr]_{i_{\alpha,X^a} }
& & T^*_g \calM\\
T_\omega \calV/\calD_0 \ar@{<-}[u] 
\ar@<2pt>@{<-}[rr]
\ar@<-2pt>@{->}[rr]_{j^\expspace\calV_\alpha }
& & T^*_\omega \calV/\calD_0
% \ar@<4pt>[u]
\ar@<0pt>[u]
% \ar@<-4pt>[u]
}
\end{gathered}
\end{equation}
commutes (with the exception that traversal of the bottom loop starting at the middle row
is a projection).

Recalling Lemma \ref{lem:voldot} and our identification of $T_\omega \calV/\calD_0$ and 
$T_\omega^* \calV/\calD_0$ with $\Reals$, we claim that
\begin{equation}\label{eq:jvol}
j^\expspace\calV_\alpha( v ) = - \left(\frac{2\kappa}{\int_M N_{g,\alpha} \omega_g}\right)v
\end{equation}
is the desired isomorphism. Evidently, $j^\expspace\calV_\alpha$ is invertible, and
\begin{equation}\label{eq:jvolinv}
(j^\expspace\calV_\alpha)^{-1}(p) = -\left(\frac{1}{2\kappa} \int_M N_{g,\alpha} \omega_g\right)p.
\end{equation}
So to establish diagram \eqref{diag:dl-legendre-conf-vol} we need only show
that traveling from the lower-right corner to the lower-left corner of diagram
\eqref{diag:dl-vol-extended} is the same map as $(j^\expspace\calV_\alpha)^{-1}$, regardless
of the choice of $g_{ab}$ with $\omega_g=\omega$.
To this end, let $p\in T^*_\omega \calV/\calD_0$. 
From equation \eqref{eq:MtoVpull} its pullback is $(g_{ab};\;0,p)^*\in T_{g}^*\calM$,
and we apply $i_{\alpha,X^a}^{-1}$ from equation \eqref{eq:i_inv} to obtain
$(g_{ab};\; L_g X_{ab}, -pN_{\alpha,g}/(2\kappa) + \div_g X)$.  Finally, applying the pushforward 
from equations \eqref{eq:volpush} and \eqref{eq:MtoVD0push} we arrive at
\begin{equation}
\int_M -(N_{\alpha,g}/(2\kappa)p + \div_g X\; \omega_g = 
-\left(\frac{1}{2\kappa}\int_M N_{g,\alpha}\;\omega_g\right)p = (j^\expspace\calV_\alpha)^{-1}(p)
\end{equation}
as desired.  This establishes diagram \eqref{diag:dl-legendre-conf-vol}, which
evidently commutes except possibly when traversing the lower loop starting at the middle
row.  As in the conformal case, such a traversal is a projection, and 
to describe concisely it we introduce the volumetric equivalent of York splitting.

\begin{lemma}[Volumetric York Splitting]\label{lem:yorkvol}
Let $g_{ab}\in\calM$ and let $N$ be a positive function.

If $\tau\in C^\infty(M)$, there is constant $\tau^*$ and a smooth vector field $V^a$
such that
\begin{equation}\label{eq:yorkvol1}
\tau = \tau^* + \frac{1}{N}\div_g V.
\end{equation}
The constant $\tau^*$ is uniquely given by
\begin{equation}\label{eq:taustar}
\tau^* = \frac{\int_M N \tau\;\omega_g}{\int_M N\;\omega_g}
\end{equation}
and $V^a$ is unique up to addition of a (smooth) divergence-free vector field.

Equivalently, if $\beta \in C^\infty(M)$, there is a unique constant 
\begin{equation}
\tau^* = \frac{\int_M \beta\;\omega_g} {\int_M N\omega_g}
\end{equation}
and a smooth vector field $V^a$, unique up to addition of a (smooth) divergence-free vector field,
such that
\begin{equation}\label{eq:yorkvol2}
\beta = N \tau^* + \div_g V
\end{equation}
\end{lemma}
\begin{proof}
Let $\tau \in C^\infty$ and let $\tau^*$ be given by equation \eqref{eq:taustar}.
So
\begin{equation}
\int_{M} N\tau - N\tau^*\;\omega_g = 0
\end{equation}
and Lemma \ref{lem:voldot} implies there is a smooth vector field $V^a$ such that
\begin{equation}\label{eq:Ntaudiff}
N\tau - N\tau^* = \div_g V.
\end{equation}
This establishes equation \eqref{eq:yorkvol2}.  

The uniqueness of $\tau^*$
follows from multiplying equation \eqref{eq:yorkvol1} by $N\omega_g$ and integrating. Moreover,
we see that we can write $\tau  = \tau^* +(1/N)\div_g \widehat V$ for some other smooth
 vector field $\widehat V^a$
if and only if the difference $V^a-\widehat V^a$ is smooth and divergence free.
Finally, we note that the decomposition \eqref{eq:yorkvol2} is a trivial (but useful) rephrasing of 
equation \eqref{eq:yorkvol1}.
\end{proof}

\begin{definition}\label{def:volyorkproj} Let $g_{ab}$ be a metric and let $\tau\in C^\infty(M)$.
The \define{volumetric York projection} of $\tau$ is
\begin{equation}\label{eq:volyorkproj}
Y_{g,\alpha}(\tau) = \frac{\int_M N_{g,\alpha} \tau\; \omega_g}{\int_M N_{g,\alpha}\;\omega_g}.
\end{equation}
Equivalently, $Y_{g,\alpha}(\tau)$ is the unique constant $\tau^*$ given by Lemma \ref{lem:yorkvol}
such that 
\begin{equation}\label{eq:volyorkprojalt}
\tau = \tau^* + \frac{1}{N_{g,\alpha}} \div_g V
\end{equation}
for some vector field $V^a$.  Note that we use the same notation $Y_{g,\alpha}$ as conformal
York projection, with the difference being that the argument is a function rather than a symmetric $(0,2)$-tensor.
\end{definition}

Using the notation of Definition \ref{def:volyorkproj}, a short computation shows that
the projection obtained by traversing the lower loop of diagram \eqref{diag:dl-legendre-conf-vol}
starting from $T^*_g\calM$ is the map
\begin{equation}
(g_{ab};\; A_{ab},f)^* \mapsto (g_{ab};\; 0, Y_{g,\alpha}(f))^*.
\end{equation}
We can also express $j^\expspace\calV_{\alpha}$ in terms of volumetric York projection.
\begin{lemma}\label{lem:jvol-alt}
Suppose $\{g_{ab};\; \beta\}\in T_{\omega_g}\calV/\calD_0$.
Then
\begin{equation}\label{eq:jvolalt}
j^\expspace\calV_\alpha( \{g_{ab};\; \beta\} ) = -2\kappa\tau^*
\end{equation}
where $\tau^*=Y_{g,\alpha}(\beta/N_{g,\alpha})$, or equivalently where
$\tau^*$ is the unique constant such that
\begin{equation}\label{eq:jvolaltbeta}
\beta = N_{g,\alpha} \tau^* + \div_g V
\end{equation}
for some vector field $V^a$.
\end{lemma}
\begin{proof}
Let $\{g_{ab};\;\beta\}\in T_{\omega_g}\calV/\calD_0$. From equations \eqref{eq:jvol}
and \eqref{eq:MtoVD0push} we find
\begin{equation}\label{eq:jvolalt1}
j^\expspace\calV_\alpha( \{g_{ab};\; \beta\} ) = -\frac{2\kappa}{\int_M N_{g,\alpha}\;\omega_g}\{g_{ab};\; \beta\}
=-2\kappa \frac{\int_M \beta\;\omega_g}{\int_M N_{g,\alpha}\;\omega_g}.
\end{equation}
Now let $\tau^*=Y_{g,\alpha}(\beta/N)$.  Equation 
\eqref{eq:volyorkprojalt} implies 
equation \eqref{eq:jvolaltbeta} and
integrating with respect to $\omega_g$ we find
\begin{equation}\label{eq:jvolalt2}
\tau^* = \frac{\int_M \beta \omega_g}{\int_M N_{g,\alpha} \omega_g}.
\end{equation}
Equation \eqref{eq:jvolalt} now follows from equations \eqref{eq:jvolalt1}
and \eqref{eq:jvolalt2}.
\end{proof}

Given $(g_{ab},K_{ab})\in \calM\times\calK$, the volumetric velocity and momentum
measured with respect to a lapse form $\alpha$ are defined analogously to their conformal counterparts.  For the velocity
we send $K_{ab}$ down the left-hand side of diagram \eqref{diag:dl-legendre-conf-vol}
starting at $\calK$, and we convert the velocity into a momentum by applying
$j^\expspace\calV_\alpha$ in the form of Lemma \ref{lem:jvol-alt}.  
This leads to the following definitions.

\begin{definition}
Let $(g_{ab},K_{ab})\in \calM\times\calK$, and let $\alpha$ be a lapse form.  
Writing $\tau=g^{ab}K_{ab}$, the \define{volumetric velocity}
of $(g_{ab},K_{ab})$, as measured by $\alpha$, is 
\begin{equation}
v^\expspace\calV_\alpha(g_{ab},K_{ab}) = \{g_{ab};\; N_{g,\alpha}\tau\}=\int_M N_{g,\alpha} \tau\;\omega_g.
\end{equation}
The \define{volumetric momentum} of $(g_{ab},K_{ab})$, as measured by $\alpha$, is
\begin{equation}
m^\expspace\calV_\alpha(g_{ab},K_{ab}) = -2\kappa \tau^*
\end{equation}
where $\tau^* = Y_{g,\alpha}(\tau)$.
\end{definition}

Note that the volumetric velocity is the rate of change of slice volume,
as measured with respect to coordinate time.  We also note that if $\tau\equiv\tau_0$
for some constant $\tau_0$, then equation \eqref{eq:volyorkproj} shows
that the volumetric momentum is simply $-2\kappa\tau_0$.

The volumetric kinetic energy is derived in a parallel fashion to  conformal kinetic
energy. Consider the kinetic energy terms of the densitized-lapse ADM Lagrangian:
\begin{equation}
K(g_{ab},u_{ab},\beta; X^a, \alpha) = \int \frac{1}{4}|u-\ck_g X|^2_g -\kappa (\beta-\div_g X)^2\; \alpha.
\end{equation}
The second term on the right-hand side involves the kinetic energy due to expansion.
Define $\tau^*$ by 
\begin{equation}
-2\kappa\tau^* = j^\expspace\calV_\alpha(\{g_{ab};\;\beta\}).
\end{equation}
From Lemma \ref{lem:jvol-alt} we see that we can write
\begin{equation}\label{eq:vKEsplit}
\beta = N_{g,\alpha}\tau^* + \div_g (V+X)
\end{equation}
for some vector field $V^a$.
Then, since $N_{g,\alpha}\alpha=\omega_g$, we find
\begin{equation}
\begin{aligned}\label{eq:volke}
-\kappa\int_M(\beta-\div_g X)^2\; \alpha &= 
-\kappa\int_M(N\tau^*+\div_g V)^2\; \alpha \\
&= -\kappa\int_M N_{g,\alpha}^2(\tau^*)^2+(\div_g V)^2\; \alpha -2\kappa\int_M \tau^* \div_g V\;\omega_g \\
&= -\kappa\int_M N_{g,\alpha}^2(\tau^*)^2+(\div_g V)^2\; \alpha.
\end{aligned}
\end{equation}
The volumetric kinetic energy, as measured by $\alpha$,
is the first term on the final right-hand side of equation \eqref{eq:volke}.

\begin{definition}  Let $\alpha$ be a lapse form.  The \define{volumetric kinetic energy}
of $(g_{ab};\; u_{ab},\beta)\in T_{g}\calM$, as measured by $\alpha$, is 
\begin{equation}
K^\expspace\calV_\alpha(g_{ab}, \beta) = -\kappa\int_M N_{g,\alpha}^2(\tau^*)^2 \;\alpha
\end{equation}
where 
\begin{equation}\label{eq:taustardef}
-2\kappa\tau^* = j^\expspace\calV_\alpha(\{g_{ab};\;\beta\}).
\end{equation}
\end{definition}
From equation \eqref{eq:taustardef} and the definition of $j^\expspace\calV_\alpha$
we see that
\begin{equation}
\tau^* =
\frac{1}{\int_M N_{g,\alpha}\;\omega_g} \{g_{ab};\;\beta\} = 
\frac{1}{\int_M N_{g,\alpha}^2\;\alpha} \{g_{ab};\;\beta\}
\end{equation}
and hence we can also write
% and hence
\begin{equation}
K^\expspace\calV_\alpha(g_{ab}, \beta) = -\frac{\kappa}{\int_M N_{g,\alpha}^2 \;\alpha}	\left(\{g_{ab};\;\beta\}\right)^2
 = -\frac{\kappa}{\int_M N_{g,\alpha} \;\omega_g}	\left(\{g_{ab};\;\beta\}\right)^2.
\end{equation}
So $K^\expspace\calV_\alpha$ descends to a Lagrangian on $T\;\calV/\calD_0$ (which we also call $K^\expspace\calV_\alpha$), and the associated Legendre transformation of $\{g_{ab};\;\beta\}$ is the linearization
\begin{equation}
(K^\expspace\calV_\alpha)'(\{g_{ab};\; \beta\})=
-\frac{2\kappa}{\int_M N_{g,\alpha} \;\omega_g}	\{g_{ab};\;\beta\} = j^\expspace\calV_\alpha( \{g_{ab};\;\beta\} ).
\end{equation}

% From the definition of $K^\expspace\calV_\alpha$, we see that its
% dependence on $(g_{ab},\beta)$ is really only
% on $\{g_{ab};\;\beta\}\in T_{\omega_g}\calV/\calD_0$
% and therefore it descends to a function, still denoted by $K^\expspace\calV_\alpha$,
% on $T\;\calV/\calM$.  Recalling that
% \begin{equation}
% \{g_{ab};\;\beta\} = \int_{M} \beta\;\omega_g
% \end{equation}

% \begin{proposition}

% \begin{equation}
% K^\expspace\calV_\alpha( \{g_{ab}})
% \end{equation}
% \end{proposition}

% -\frac{\kappa}{	} \int N_{\alpha,g}^2|\sigma|_g^2\;\alpha
% \end{equation}
% where $\sigma_{ab}$ is the $g_{ab}$-TT tensor such that
% $j_{\alpha}(\{g_{ab};\;u_{ab}\}) = \{g_{ab};\; \sigma_{ab}\}^*$.
% \end{definition}

% and admits the following equivalent forms:
% \begin{equation}
% \begin{aligned}
% K^\expspace\calV_\alpha(g_{ab};\; u_{ab},\beta) &= -\kappa (\tau^*)^2 \int_M N_{g,\alpha}^2\alpha\\
% & = -\kappa (\tau^*)^2 \int_M N_{g,\alpha} \;\omega_g\\
% &= -\frac{\kappa }{\int_M N_{g,\alpha}\;\omega_g} \left[\int_M \beta\; \omega_g\right]^2\\
% &= -\frac{\kappa }{\int_M N_{g,\alpha}\;\omega_g} \left[v^\expspace\calV_\alpha(g_{ab};\; u_{ab},\beta)\right]^2.
% \end{aligned}
% \end{equation}
% The final expression shows that the volumetric kinetic energy is really a Lagrangian on $T\calV/\calD_0$,
% and it is easy to see that $j^\expspace\calV_\alpha$ is its associated Legendre transformation.

\section{Volumetric Momentum and the Conformal Method}\label{sec:volmom}

Consider
Hamiltonian conformal method parameters $(\mathbf g, \bfsigma, \tau, \alpha)$
and suppose $(\ol g_{ab}, \ol K_{ab})$ is a solution of the vacuum 
Einstein constraint equations generated by it.  So
\begin{equation}\label{eq:cmsol}
\begin{aligned}\strut
[\ol g_{ab}] &= \mathbf{g}\\
m^\expspace\calC_\alpha( \ol g_{ab}, \ol K_{ab}) &= \bfsigma \\
\ol g^{ab} \ol K_{ab} &= \tau.
\end{aligned}
\end{equation}
The conformal momentum of the solution, as measured by $\alpha$, is specified directly
via $\bfsigma$.  There is only an indirect
connection, however, between the conformal data and the volumetric momentum measured by $\alpha$.  
Indeed, suppose
$g_{ab}$ is a representative of $\mathbf g$ and let $\sigma_{ab}$ be the representative
of $\bfsigma$ with respect to $g_{ab}$. Equation \eqref{eq:cmsol} is equivalent to
the existence of a conformal factor $\phi$ and a vector field $W^a$ such that
\begin{equation}
\begin{aligned}
\ol g_{ab} &= \phi^{q-2} g_{ab}\\
\ol K_{ab} &= \phi^{-2}\left( \sigma_{ab} + \frac{1}{2N_{\alpha,g}} \ck_g W_{ab}\right) + \frac{\tau}{n}\ol g_{ab}.
\end{aligned}
\end{equation}
and such that $\phi$ and $W^a$ solve the CTS-H equations \eqref{eq:CTS-H}.
The volumetric momentum of $(\ol g_{ab},\ol K_{ab})$ measured by $\alpha$ is
$-2\kappa\tau^*$ where
\begin{equation}\label{eq:taustar2}
\tau^* = \frac{\int_M N_{\ol g, \alpha} \tau\;\omega_{\ol g}}{\int_M N_{\ol g, \alpha} \;\omega_{\ol g}}
=\frac{\int_M \phi^{2q} N_{g, \alpha} \tau\;\omega_{g}}{\int_M \phi^{2q} N_{g, \alpha} \;\omega_{g}}.
\end{equation}
Notice from the right-hand side of equation \eqref{eq:taustar2} that 
the computation of $\tau^*$ from $(g_{ab},\sigma_{ab},\tau,\alpha)$ appears to involve the unknown
conformal factor $\phi$ in an essential way.  Although we need not know $\phi$ exactly
(one can compute $\tau^*$ from $c\phi$ for any positive constant $c$), 
it seems unlikely that one can compute $\tau^*$ without at least determining
at least $c\phi$ and thereby effectively solving the CTS-H equations.
Moreover, if the conformal data generates
more than one solution of the constraints, as happens at least in some cases involving an $L^\infty$ mean curvature that changes sign \cite{Maxwell:2011if}, there is no reason to believe that the volumetric momenta of the two solutions will agree.

Hence the conformal method treats the conformal and volumetric degrees of freedom differently,
with the conformal degrees respecting a kind of diffeomorphism invariance, but not the volumetric
degrees.  This discrepancy seems to negatively impact the
effectiveness of the conformal method as a parameterization in the far-from-CMC setting.
As mentioned in the introduction, the recent study in \cite{Maxwell:2014b} presented a family $\calF$ of 
smooth, non-CMC conformal data sets that generate certain $U^{n-1}$-symmetric slices of flat Kasner spacetimes. 
Given $(\mathbf g, \bfsigma, \tau, \alpha)\in \calF$, it either generates a single $U^{n-1}$-symmetric slice
of a flat Kasner spacetime, or it generates a homothety family of $U^{n-1}$-symmetric slices.  The
homothety families appear precisely when $\tau^*=0$, as computed with respect to one (and consequently any)
of the generated $U^{n-1}$-symmetric solutions of the Einstein constraint equations.
So the quantity $\tau^*$ that
we seem to be unable to control directly from the conformal parameters determines, in the setting
of \cite{Maxwell:2014b}, the multiplicity of solutions generated by the conformal parameters.

From the evidence of the role of $\tau^*$ from \cite{Maxwell:2014b}, along with the naturality
of treating the conformal and volumetric degrees of freedom in the parallel ways discussed
in Sections \ref{sec:conformalLegendre} and \ref{sec:volumetricLegendre}, we 
are therefore lead to consider conformal-like methods where the parameters include
\begin{enumerate}
\item a conformal class $\mathbf g$,
\item a lapse form $\alpha$,
\item either a conformal velocity $\{\mathbf u\}$ or a conformal momentum $\bfsigma$,
with $\bfsigma = j_\alpha^\expspace\calC(\{\mathbf u\})$, and
\item either a volumetric velocity $v\in\Reals$ or a volumetric momentum $-2\kappa\tau^*\in\Reals$
with $-2\kappa\tau^* = j_\alpha^\expspace\calV(v)$.
\end{enumerate}
This list is evidently not comprehensive; the standard conformal method is successful in the near-CMC
case but we have now replaced a function $\tau$ with a scalar $\tau^*$.  In the remainder of the paper
we examine alternatives for augmenting this list with geometrically natural degrees of freedom.

% these parameters (along with the requirement that the solution be
% CMC, and noting that the choice of $\alpha$ has no bearing in the CMC case).  From the
% success of the standard conformal method for near-CMC data, we also see that
% the list is not comprehensive: we have replaced a function $\tau$ with a scalar $\tau^*$,
% so there must be other degrees of freedom still to be identified.
% For example, let $M$
% be the torus $S^1\times\cdots\times S^1$ with coordinates
% $(s^1,\ldots,s^n)$ and let $g_{ab}$ be the product metric:
% \begin{equation}
% g = (ds^1)^2 + \cdots + (ds^n)^2
% \end{equation}
% Define a transverse-traceless tensor $\sigma^\flat_{ab}$ by
% \begin{equation}
% \sigma^\flat = (ds^1)^2 - \frac{1}{n} g,
% \end{equation}
% let $\tau$ be any mean curvature on $M$ that depends only on $s^1$,
% let $N$ be a lapse that depends only on $s^1$, and let $\alpha$
% be the lapse form $\omega_g/N$.  Now consider the conformal
% data sets $(g_{ab},\mu\sigma_{ab},\tau_a=a+\tau,\alpha)$
% where $\mu$ and $a$ are constants.

\section{Drifts}\label{sec:drift}

Consider a path of metrics $g_{ab}(t)$ such that the diffeomorphism class of
the conformal class of $g_{ab}(t)$ is constant along the curve, and 
such that the diffeomorphism class of the volume form 
of $g_{ab}(t)$ is also constant along the curve. 
By applying an appropriate path of diffeomorphisms, we could ensure
that either the conformal class or the volume form is constant along the curve,
but in general we cannot ensure both are constant.  For example, suppose we 
apply diffeomorphisms to fix the conformal class.  Since the diffeomorphism class of the
volume form is constant, the volume will also remain constant along the curve,
but we are free to smoothly reallocate the fixed volume. 
So although the conformal geometry and volume
are constant, the conformal class and volume form can move relative to one another.
Since the conformal class is a more rigid object than the volume form
(e.g., the space of conformal Killing fields is finite dimensional, but the space of divergence-free
vector fields is not), we visualize the volume form as drifting
relative to the landmarks provided by the fixed conformal geometry.
With this intuition in mind, we call an infinitesimal
motion in $\calM/\calD_0$ that preserves conformal geometry and volume a drift.

To formalize these ideas, we first observe that the pushforwards from $T_{g}\calM$ to $T_{[g]}\calC/\calD_0$ and 
$T_{\omega_g}\calV/\calD_0$ can be factored through $T_{g}\calM/\calD_0$
to obtain the maps $\pi^\expspace\calC_*$ and $\pi^\expspace\calV_*$ in the diagram
\begin{equation}\label{diag:factor}
\begin{gathered}
\xymatrix{ & T_g \calM \ar[d]\ar[ldd]\ar[rdd]&\\
& T_g \calM/\calD_0\ar[rd]_{\pi^\expspace\calV_*}\ar[ld]^{\pi^\expspace\calC_*} &\\
T_{[g]} \calC/\calD_0 & & T_{\omega_g} \calV/\calD_0.
}
\end{gathered}
\end{equation}
Indeed, we claim that
\begin{equation}\label{eq:picalCstar}
\pi^\expspace\calC_*( \{ g_{ab};\; u_{ab},\beta\}) =  \{ g_{ab};\; u_{ab}\}.
\end{equation}
First, note that the map $\pi^\expspace\calC$ is well defined, for if $X^a$ is a vector field,
$\Lie_g X= (g_{ab};\; \ck_g X_{ab}, \div_g X)$ and
\begin{equation}
\pi^\expspace\calC_*( \{ g_{ab};\; \ck_g X,\div_g X\}) = \{ g_{ab};\; \ck_g X_{ab}\} = 0.
\end{equation}
Moreover, from equations \eqref{eq:MtoMD0push} and \eqref{eq:calCpush} 
we see that equation \eqref{eq:picalCstar} is exactly the statement
that the left-hand triangle of diagram \ref{diag:factor} commutes.
Similar considerations show that 
\begin{equation}\label{eq:picalVstar}
\pi^\expspace\calV_*( \{ g_{ab};\; u_{ab},\beta\} ) = \{ g_{ab};\; \beta\} = \int_M \beta\;\omega_g.
\end{equation}

\begin{definition}
Let $g_{ab}\in\calM$. A \textbf{drift} at $g_{ab}$ is an element $\mathbf U\in T_g \calM/\calD_0$
such that $\pi^\expspace\calC_*(\mathbf U)= 0$ and $\pi^\expspace\calV_*(\mathbf U)= 0$.  We denote the
collection of drifts at $g_{ab}$ by
$\Drift_g$.
\end{definition}

\begin{lemma}\label{lem:drift_id}  Suppose $\mathbf U\in T_g\calM$. Then $\mathbf U\in\Drift_g$ 
if and only 
if there is a vector field $R^a$ such that
\begin{equation}\label{eq:drift-form}
\mathbf U = \{g_{ab}; 0, \div_g R\}.
\end{equation}
Moreover, if $\widehat R^a$ is another vector field, then
\begin{equation}\label{eq:RRhat}
\{g_{ab};\; 0, \div_g R\} = \{g_{ab};\; 0, \div_g \widehat R\}
\end{equation}
if and only if there is a divergence-free vector field $E^a$ and a conformal Killing field $Q^a$
such that
\begin{equation}\label{eq:driftQE}
\widehat R^a = R^a + E^a + Q^a.
\end{equation}
\end{lemma}
\begin{proof}
Suppose $\mathbf U = \{g_{ab};\; u_{ab},\beta\}\in T_g\calM/\calD_0$.  
From equation \eqref{eq:picalCstar} we see that 
$\pi^\expspace\calC_*(\mathbf U) = 0$ if and only if $u_{ab}\in \Im \ck_g$.  Hence
there is a vector field $W^a$ such that $u_{ab} = \ck_g W_{ab}$.  Similarly,
from equation \eqref{eq:picalVstar} and
Lemma \ref{lem:voldot} we see that
$\pi^\expspace\calV_*(\mathbf U) = 0$ if and only if 
there is a vector field $V^a$ such that $\beta = \div_g V$.
Thus $\mathbf U$ is a drift if and only if there are vector
fields $W^a$ and $V^a$ such that
\begin{equation}
\mathbf U = \{g_{ab};\; \ck_g W_{ab}, \div_g V\}.
\end{equation}
Moreover,
\begin{equation}
\begin{aligned}
\{g_{ab};\; \ck_g W_{ab}, \div_g V\} &= \{g_{ab};\; \ck_g W_{ab}, \div_g V\} - \Lie_g W \\
&=  \{g_{ab};\; \ck_g W_{ab}, \div_g V\} - \{g_{ab};\; \ck_g W_{ab}, \div_g W\} \\
&=  \{g_{ab};\; 0, \div_g (V-W)\}.
\end{aligned}
\end{equation}
Setting $R^a=V^a-W^a$ we see that $\mathbf U$ is a drift if and only if there is a vector field $R^a$
such that equation \eqref{eq:drift-form} holds.

Now suppose $R^a$ and $\widehat R^a$ are vector fields such that
\begin{equation}\label{eq:RRhatsame}
\{g_{ab};\; 0, \div_g R\} =  \{g_{ab};\; 0, \div_g \widehat R\}.
\end{equation}
Hence 
\begin{equation}
(g_{ab};\; 0, \div_g (\widehat R- R) ) \in \Im \Lie_g.
\end{equation}
and there is a vector field $Q^a$ such that 
\begin{equation}\label{eq:driftrans}
(g_{ab};\; 0, \div_g (\widehat R- R) ) = \Lie_g Q = (g_{ab};\; \ck_g Q_{ab}, \div_g Q )
\end{equation}
Equation \eqref{eq:driftrans} implies $\ck_g Q_{ab}=0$ and hence $Q^a$ is a conformal Killing field.  
Defining $E^a=\widehat R^a-R^a -Q^q$, equation \eqref{eq:driftrans} also
implies that $E^a$ is divergence free. Since
\begin{equation}
\widehat R^a = R^a + Q^a + E^a
\end{equation}
we see that if equation \eqref{eq:RRhat} holds then so does equation \eqref{eq:driftQE}.
Conversely, if $R^a$ and $\widehat R^a$ are related via \eqref{eq:driftQE}
we can reverse the previous argument to conclude \eqref{eq:RRhatsame}.	
\end{proof}

Given a metric $g_{ab}$, let $\calQ_g$ be the subgroup of $\calD_0$ that preserves 
the conformal class of $g_{ab}$ 
and let $\calE_g$ be the subgroup that preserves the volume form of $g_{ab}$.
We define $T_e \calQ_g$ to be
the set of conformal Killing fields of $g_{ab}$ and $T_e\calE_g$ to
be the set of $\omega_g$-divergence free vector fields.  Lemma \ref{lem:drift_id}
provides an isomorphism
\begin{equation}
\Drift_g \approx T_e \calD_0 / (T_e \calQ_g \oplus T_e \calE_g).
\end{equation}

We wish to show that motion in $\calM/\calD_0$ can be completely described
in terms of volume expansion, conformal deformation, and drift. 
If $\mathbf U\in T_g\calM/\calD_0$, assigning a a conformal velocity
and volumetric velocity is straightforward: simply apply $\pi^\expspace\calC_*$ and $\pi^\expspace\calV_*$.
Assigning a drift to $\mathbf U$ requires, however, a choice of projection 
\begin{equation}
T_g\calM/\calD_0 \ra \Drift_g
\end{equation}
and we now construct a family of such projections that depend on the choice of a lapse
form $\alpha$. 

Consider the lower loop of diagram \ref{diag:dl-legendre-conf} where
we additionally factor the pushforward $T_g\calM\ra T_{[g]}\calC/\calD_0$
through $T_{g}\calM/\calD_0$:
\begin{equation}\label{diag:fourfiveC}
\begin{gathered}
\xymatrix{
 T_g\calM 
\ar@<2pt>@{<-}[rr]
\ar@<-2pt>@{->}[rr]_{i_{\alpha,X^a} }
\ar[d]
& & T^*_g \calM\\
T_{g}\calM/\calD_0 & & \\
T_{[g]} \calC/\calD_0 \ar@{<-}[u]^{\pi^\expspace\calC_*} 
\ar@<2pt>@{<-}[rr]
\ar@<-2pt>@{->}[rr]_{j_{\alpha}^\expspace\calC }
& & T^*_{[g]} \calC/\calD_0
\ar@<0pt>[uu]
}
\end{gathered}
\end{equation}
Let $\iota^\expspace\calC:T_{[g]}\calC/\calD_0\ra T_{g}\calM/\calD_0$ 
be the map obtained by nearly completing the loop in diagram \eqref{diag:fourfiveC}.
From Lemma \ref{lem:jinv2} and equations \eqref{eq:i_inv} and
\eqref{eq:MtoMD0push} we find
\begin{equation}\label{eq:iotaC}
\iota^\expspace\calC_\alpha(\{g_{ab};\;u_{ab}\} ) = \{g_{ab};\; 2N_{g,\alpha}\sigma_{ab},0\}
\end{equation}
where $\sigma_{ab}$ is the unique $g_{ab}$-TT tensor such that
\begin{equation}
u_{ab} =2N_{g,\alpha}\sigma_{ab} + \ck_g W_{ab}
\end{equation}
for some vector field $W^a$.  The following lemma shows that
$\iota_{\alpha}^\expspace\calC$ selects an $\alpha$-dependent representative
in $T_g\calM/\calD_0$ of each conformal motion in $T_{[g]}\calC/\calD_0$.
\begin{lemma}\label{lem:piCViotaCV}
The map $\iota^\expspace\calC_\alpha$ satisfies
\begin{equation}\label{eq:piiotaC}
\begin{aligned}
\pi^\expspace\calC_*\circ \iota^\expspace\calC_\alpha &= \id \\
\pi^\expspace\calV_*\circ \iota^\expspace\calC_\alpha &= 0.
\end{aligned}
\end{equation}
\end{lemma}
\begin{proof}
Note that $\pi^\expspace\calC_*\circ \iota^\expspace\calC_\alpha$ is the map obtained
by traversing the bottom loop of diagram \ref{diag:dl-legendre-conf}
starting at $T_{[g]}\calC/\calD_0$.  In Section \ref{sec:conformalLegendre}
we showed that this map is the identity.  On the other hand, from
equations \eqref{eq:iotaC} and \eqref{eq:picalVstar} we see that
$\pi^\expspace\calV_*\circ \iota^\expspace\calC_\alpha = 0$.
\end{proof}

Similarly, from the diagram
\begin{equation}\label{diag:fourfiveV}
\begin{gathered}
\xymatrix{
 T_g\calM 
\ar@<2pt>@{<-}[rr]
\ar@<-2pt>@{->}[rr]_{i_{\alpha,X^a} }
\ar[d]
& & T^*_g \calM\\
T_{g}\calM/\calD_0 & & \\
T_{\omega_g} \calV/\calD_0 \ar@{<-}[u]^{\pi^\expspace\calV_*} 
\ar@<2pt>@{<-}[rr]
\ar@<-2pt>@{->}[rr]_{j_{\alpha}^\expspace\calV }
& & T^*_{\omega_g} \calV/\calD_0
\ar@<0pt>[uu]
}
\end{gathered}
\end{equation}
we obtain a map $\iota^\expspace\calV_\alpha:T_{\omega_g}\calV/\calD_0 \ra T_{g}\calM/\calD_0$
given by
\begin{equation}\label{eq:iotaV}
\iota^\expspace\calV_\alpha(\{g_{ab};\;\beta\} ) = \{g_{ab};\; 0, N_{g,\alpha}\tau^*\}
\end{equation}
where $\tau^*$ is the unique constant given by volumetric York splitting (Lemma \ref{lem:yorkvol}) such that
\begin{equation}
\beta = N_{g,\alpha}\tau^* + \div_g V
\end{equation}
for some vector field $V^a$.  We have an analogue of Lemma \ref{lem:piCViotaCV}
that shows that $\iota_\alpha^\expspace\calV$ selects
an $\alpha$-dependent representative in $T_g\calM/\calD_0$ 
of each volumetric motion in $T_{\omega_g}\calV/\calD_0$; we omit the proof.
\begin{lemma}\label{lem:piViotaV}
The map $\iota^\expspace\calV_\alpha$ satisfies
\begin{equation}\label{eq:piiotaV}
\begin{aligned}
\pi^\expspace\calC_*\circ \iota^\expspace\calV_\alpha &= 0 \\
\pi^\expspace\calV_*\circ \iota^\expspace\calV_\alpha &= \id.
\end{aligned}
\end{equation}
\end{lemma}

Writing
\begin{equation}
\iota^{\rm Drift}: \Drift_g \ra T_g\calM/\calD_0
\end{equation}
for the natural embedding we define
\begin{equation}
\iota_\alpha: T_{[g]}\calC/\calD_0\; \oplus\; T_{\omega_g}\calV/\calD_0\; \oplus\; \Drift_g \ra  T_g \calM/\calD_0
\end{equation}
by
\begin{equation}
\iota_\alpha = \iota_\alpha^\expspace\calC \oplus \iota_\alpha^\expspace\calV \oplus \iota^{\Drift}.
\end{equation}

\begin{proposition}\label{prop:CVD}
Let $g_{ab}$ be a metric and let $\alpha$ be a lapse form.  Then
$\iota_\alpha$ is an isomorphism and the following diagram commutes:
\begin{equation}\label{diag:CVDiso}
\begin{gathered}
\xymatrix{
& T_{[g]}\calV/\calD_0\\
T_{[g]}\calC/\calD_0 \oplus T_{\omega_g}\calV/\calD_0 \oplus \Drift_g
\ar@<2pt>[r]^-{\iota_\alpha}
\ar[rd]\ar[ru]
\ar@<-2pt>@{<-}[r]\ar[rd]
 &
T_{g} \calM/\calD_0
\ar[u]_{\pi^\expspace\calV_\alpha}
\ar[d]^{\pi^\expspace\calC_\alpha}\\
& T_{[g]}\calC/\calD_0.
}
\end{gathered}
\end{equation}
Moreover, if $\mathbf R$ is a drift,
\begin{equation}\label{eq:ialphainvR}
\iota_{\alpha}^{-1}(\mathbf R) = (0,0,\mathbf R).
\end{equation}
\end{proposition}
\begin{proof}
Note that $\Drift_g = \Ker \pi^\expspace\calC_*\cap \Ker \pi^\expspace\calV_*$, so $\pi^\expspace\calC_*\circ \iota^{\Drift}=0$
and $\pi^\expspace\calV_*\circ \iota^{\Drift}=0$. Using this fact and Lemmas \ref{lem:piCViotaCV} and 
\ref{lem:piViotaV} we conclude
\begin{equation}
\pi^\expspace\calC_*\circ \iota_\alpha = \pi^\expspace\calC_*\circ\iota^\expspace\calC_\alpha\; +\;
\pi^\expspace\calC_*\circ\iota^\expspace\calC_\alpha \;+\;
\pi^\expspace\calC_*\circ\iota^{\Drift} = \id \;+\; 0 \;+\; 0.
\end{equation}
This establishes the lower triangle of diagram \eqref{diag:CVDiso} up to
showing $\iota_\alpha$ has an inverse. The upper triangle is established similarly,
and we turn our attention to the invertibility of $\iota_\alpha$.

To see that $\iota_\alpha$ is injective, notice that from the facts established thus far
for diagram \eqref{diag:CVDiso}
that anything in the kernel of $\iota_\alpha$ must be of the form $(0,0,\mathbf R)$
for some drift $\mathbf R$.  But $\iota_\alpha(0,0,\mathbf R)=\mathbf R$,
so $\iota_\alpha$ has trivial kernel.

To show that $\iota_\alpha$ is surjective, let $\{g_{ab};\;u_{ab},\beta\}\in T_g\calM/\calD_0$.
Let $\sigma_{ab}$ be the $g_{ab}$-TT tensor such that
\begin{equation}
u_{ab} = 2 N_{g,\alpha} \sigma_{ab} + \ck_g W_{ab}
\end{equation}
for some vector field $W^a$, and let $\tau^*$ be the constant such that
\begin{equation}
\beta = N_{g,\alpha} \tau^* + \div_g V
\end{equation}
for some vector field $V^a$.  Let 
\begin{equation}
\mathbf R = \{g_{ab};\; \ck_g W_{ab}, \div_g V\}
\end{equation}
and observe that $\mathbf R$ is a drift.  Then
\begin{equation}
\begin{aligned}
\iota^\expspace\calC_\alpha(\{g_{ab};\; u_{ab}\}) &= \{g_{ab};\; 2N_{g,\alpha}\sigma_{ab},0\} \\
\iota^\expspace\calV_\alpha(\{g_{ab};\; \beta\}) &= \{g_{ab};\; 0,N_{g,\alpha}\tau^*\}\\
\iota^{\Drift}( \mathbf R)& = \{g_{ab};\; \ck_g W_{ab}, \div_g V\}
\end{aligned}
\end{equation}
so
\begin{equation}
\begin{aligned}
\iota_\alpha(\{g_{ab};\; u_{ab}\}, \{g_{ab};\; \beta\}, \mathbf R) &= 
\{g_{ab};\; 2N_{g,\alpha}\sigma_{ab},0\} + \{g_{ab};\; 0, N_{g,\alpha}\tau^*\}
+\{g_{ab};\; \ck_g W_{ab}, \div_g V\}\\
&= \{g_{ab};\; 2N_{g,\alpha}\sigma_{ab} + \ck_g W_{ab}, N_{g,\alpha}\tau^* + \div_g V\}\\
&= \{g_{ab};\; u_{ab}, \beta\}
\end{aligned}
\end{equation}
as desired.

Finally, we note that equation \eqref{eq:ialphainvR} follows from the invertibility of
$\iota_\alpha$ and the fact that $\iota_\alpha(0,0,\mathbf R)=\mathbf R$ for any drift $\mathbf R$.
\end{proof}

Proposition \ref{prop:CVD} is the formal assertion that motion in $\calM/\calD_0$ is 
characterized by conformal deformation, volume expansion, and drift.  The conformal and volumetric
velocities are unambiguously associated with $\mathbf U\in T_g\calM/\calD_0$ via $\pi^\expspace\calC_*$ and $\pi^\expspace\calV_*$, and Proposition \ref{prop:CVD} provides a lapse-form-dependent 
map from $T_g\calM/\calD_0$ to $\Drift_g$: compute $\iota_\alpha^{-1}$, and
extract the drift component.  Let us call this map $\pi^{\Drift}_\alpha$.

\begin{proposition}\label{prop:pidrift}
The map $\pi^{\Drift}_\alpha: T_g\calM/\calD_0\ra \Drift_g$ is a projection and
\begin{equation}\label{eq:pidrift}
\pi^{\Drift}_\alpha( \{g_{ab};\; u_{ab}, \beta \}) = \{g_{ab};\; \ck_g W_{ab}, \div_g V\}
\end{equation}
where $W^a$ and $V^a$ are any vector fields 
obtained from York splitting
\begin{equation}
\begin{aligned}
u_{ab} &= 2N_{g,\alpha} \sigma_{ab} +\ck_g W \\
\beta  &= 2N_{g,\alpha}\tau^* + \div_g V
\end{aligned}
\end{equation}
for some $g_{ab}$-TT tensor $\sigma_{ab}$ and some constant $\tau^*$.

Moreover,
\begin{equation}\label{eq:iotainv}
\iota_\alpha^{-1} = \pi^\expspace\calC_\alpha\oplus \pi^\expspace\calV_\alpha \oplus \pi^{\Drift}_\alpha.
\end{equation}

\end{proposition}
\begin{proof}
That $\iota_\alpha$ is a projection follows from equation \eqref{eq:ialphainvR},
and formula \eqref{eq:pidrift} was established in the body of the proof of 
Proposition \ref{prop:CVD}. Equation \eqref{eq:iotainv} follows from Proposition \ref{prop:CVD}
and the definition of $\pi^{\Drift}_\alpha$.
\end{proof}

\section{Drifts and the Momentum Constraint}\label{sec:drift-mom}

Consider a metric $g_{ab}$ and a lapse form $\alpha$.
From diagram \eqref{diag:dl-legendre} and the pushforwards and pullbacks
described in Section \ref{sec:intro} we obtain the diagram
\begin{equation}\label{diag:dl-legendre-extended-M}
\begin{gathered}
\xymatrix{
 &\calK \ar@{<->}[dr] \ar@{<->}[dl]_{(\alpha,X^a)} \\
 T_{g}\;\calM 
\ar@<2pt>@{<-}[rr]
\ar@<-2pt>@{->}[rr]_{i_{\alpha,X^a} }
& & T^*_{g}\; \calM\\
T_g \calM/\calD_0 \ar@{<-}[u] & & T^*_g \calM/\calD_0.\ar[u]
}
\end{gathered}
\end{equation}
which is the analog of the conformal
and volumetric equivalents \eqref{diag:dl-legendre-extended} and 
\eqref{diag:dl-vol-extended}.  In the conformal and volumetric
cases, the Legendre transformation $i_{\alpha,X^a}$ descended to 
a Legendre transformation after quotienting by diffeomorphisms.
This is not the case for diagram \eqref{diag:dl-legendre-extended-M}, however.
There is certainly a map $k_\alpha:T^*_g \calM/\calD_0\ra T_g\calM/\calD_0$ obtained by traveling
from the lower-right to the lower-left of diagram \eqref{diag:dl-legendre-extended-M}:
\begin{equation}
\begin{gathered}
\xymatrix{
 &\calK \ar@{<->}[dr] \ar@{<->}[dl]_{(\alpha,X^a)} \\
 T_{g}\;\calM 
\ar@<2pt>@{<-}[rr]
\ar@<-2pt>@{->}[rr]_{i_{\alpha,X^a} }
& & T^*_{g}\; \calM\\
T_g \calM/\calD_0 \ar@{<-}[u] & & T^*_g \calM/\calD_0\ar[u]\ar[ll]^{k_\alpha}.
}
\end{gathered}
\end{equation}
But it turns out that $k_\alpha$ can fail to be an isomorphism, and this gives
some insight about the configuration space for the Einstein equations.
The elements of $T^*_g \calM/\calD_0$ are precisely the solutions of the momentum
constraint, but is is not quite correct to think of these as momenta corresponding
to the velocity of the system in $T_g \calM/\calD_0$, and the notion of drifts seems
to play a key role here.

An arbitrary element of $T_g^*\calM/\calD_0$ can be written as
\begin{equation}
(g_{ab};\; A_{ab},-2\kappa\tau)^*
\end{equation}
for some trace-free $A_{ab}$ and some function $\tau$ that satisfy the momentum constraint
\begin{equation}\label{eq:momTTF}
\nabla^a A_{ab} = \kappa \nabla_b\tau.
\end{equation}
From equations \eqref{eq:i_inv} and \eqref{eq:MtoMD0push} we find
\begin{equation}\label{eq:kalpha}
\begin{aligned}
k_\alpha( (g_{ab};\; A_{ab},-2\kappa\tau)^* ) &= 
\{g_{ab}, 2N_{g,\alpha}A_{ab}+\ck_{g}X_{ab}, N_{g,\alpha}\tau+\div_g X\}\\
&=\{g_{ab}, 2N_{g,\alpha}A_{ab}, N_{g,\alpha}\tau\}.
\end{aligned}
\end{equation}
Applying York splitting we can write
\begin{equation}
\begin{aligned}
A_{ab} &= \sigma_{ab} + \frac{1}{2N_{g,\alpha}} \ck_g W_{ab}\\
\tau &= \tau^* + \frac{1}{N_{g,\alpha}}\div_g V
\end{aligned}
\end{equation}
for a $g_{ab}$-TT tensor $\sigma_{ab}$, a constant $\tau^*$, and vector fields $W^a$ and $V^a$.
Writing $\mathbf U$ for $k_\alpha( (g_{ab};\;A_{ab}, -2\kappa\tau)^* )$ equation
\eqref{eq:kalpha} becomes
\begin{equation}
\mathbf U = 
\{g_{ab}, 2N_{g,\alpha}\sigma_{ab},0\} +
\{g_{ab}, 0,N_{g,\alpha}\tau^*\} + \{ g_{ab};\; \ck_g W_{ab}, \div_g V\},
\end{equation}
so in the language of Proposition \ref{prop:pidrift}
\begin{equation}\label{eq:momsplit}
\begin{aligned}
\pi^\expspace\calC_\alpha(\mathbf U) &= \{g_{ab};\; 2N_{g,\alpha}\sigma_{ab}\}\\
\pi^\expspace\calV_\alpha(\mathbf U) &= \{g_{ab};\; N_{g,\alpha}\tau^*\}\\
\pi^{\Drift}_\alpha(\mathbf U) &= \{g_{ab};\; \ck_g W_{ab}, \div_g V\} = \{g_{ab};\; 0, \div_g(V-W)\}.
\end{aligned}
\end{equation}
The potential for difficulty lies in the cancellation in 
$\pi^{\Drift}_\alpha(\mathbf U)$: although $W^a$ and $V^a$ might not be zero, it
might be that $\pi^{\Drift}_\alpha(\mathbf U)$ is zero (and hence $k_\alpha$ may have nontrivial kernel). 

The momentum constraint \eqref{eq:momTTF} can be written in terms of the York-projected variables as
\begin{equation}
\div_g\left(\sigma+\frac{1}{2N_{g,\alpha}}\ck_{g} W \right) = \kappa \extd\left( \tau^* 
+\frac{1}{N_{g,\alpha}}\div_g V\right)
\end{equation}
or more simply
\begin{equation}\label{eq:drift}
\begin{aligned}
\div_g\left(\frac{1}{2N}\ck_{g} W \right) &= \kappa \extd\left( \frac{1}{N}\div_g V\right)
%&&\qquad\text{\small[drift equation]}
\end{aligned}
\end{equation}
where $N=N_{g,\alpha}$.
So $\sigma_{ab}$ and $\tau^*$ are not involved in the momentum constraint, and we have only the relationship
\eqref{eq:drift} between $W^a$ and $V^a$ that, for reasons explained below, we call the \define{drift equation}.  
One might hope that equation \eqref{eq:drift}
prevents cancellation in $\pi^{\Drift}_\alpha(\mathbf U)$, but this is not always the case.

Suppose $M^n$ is the torus $(S^1)^n$
equipped with the flat product metric $\ol g_{ab}$, and let $s$ be
the coordinate of the first factor 
of the torus.
Consider vector fields $W^a=(w(s),0,\ldots,0)$ and $V^a=(v(s),0,\ldots,0)$, and 
suppose $N$ is a lapse that depends only on $s$.  
A brief computation shows that equation \eqref{eq:drift} can be written
\begin{equation}\label{eq:momcircle}
\kappa\left(\frac{1}{2N} 2 w' \right)' = \kappa\left(\frac{1}{N} v'\right)'
\end{equation}
where primes denote differentiation with respect to $s$. Since \eqref{eq:momcircle}
is an equation on the circle, $w(s)$ and $v(s)$ solve equation \eqref{eq:momcircle}
if and only if $w=v+c$ for some constant $c$. Hence $W^a=V^a+K^a$ for some Killing field
$K^a$ and the associated drift
from equation \eqref{eq:momsplit} vanishes identically. Thus, in this case,
$k_\alpha$ has nontrivial kernel and is not an isomorphism.

The thin-sandwich conjecture
\cite{Baierlein:1962tn} states that initial data is characterized by a metric $g_{ab}$ and
the projection of its ADM velocity into $T_g \calM/\calD_0$. It
is not expected to hold in general \cite{Bartnik:1993jl},
and the observation from the preceding paragraph appears to be another facet of
its failure.
Indeed, from \cite{Maxwell:2014b} Proposition 6.2 and the discussion above it follows that
there exist distinct solutions of the constraint equations, that generate distinct spacetimes,
that nevertheless have the same metric and such that for some lapse form $\alpha$ 
\begin{itemize}
\item the the conformal velocities measured by $\alpha$ for both solutions are the same,
\item the the volumetric velocities measured by $\alpha$ for both solutions are the same, and
\item the complementary drifts for both solutions are zero.
\end{itemize}
Hence the projections of the ADM velocities in to $T\calM/\calD_0$
for these distinct solutions of the constraint equations are identical.

Although $k_\alpha$ is not an isomorphism, it turns out that solutions of the momentum constraint
can nevertheless be parameterized in terms of a conformal momentum, a volumetric momentum, and 
a drift.  The key idea is to identify equation \eqref{eq:drift} as representing a relationship between
two drifts, and we start by looking at the role of $V^a$ .

\begin{theorem}\label{thm:drift}
Suppose $g_{ab}$ is a metric, $\alpha$ is a lapse form, and $\mathbf V\in \Drift_g$.
Let $V^a$ be any vector field such that
\begin{equation}
\mathbf V  = \{g_{ab};\; 0, \div_g V \}.
\end{equation}
Then there is a 
conformal Killing field $Q^a$, unique up to addition of a proper Killing field, such that
\begin{equation}\label{eq:drifteq}
\div_g\left( \frac{1}{2N_{g,\alpha}}\ck_g W\right) = \kappa\, \mathbf{d}\left( \frac{1}{N_{g,\alpha}}\div_g(V+Q)\right)
\end{equation}
admits a solution $W^a$.  The solution $W^a$ is unique up to addition of a conformal
Killing field, and the set of solutions does not depend on the choice of $V^a$
representing $\mathbf V$ or on the subsequent
choice of conformal Killing field $Q^a$ such that equation \eqref{eq:drifteq} is solvable.
\end{theorem}
\begin{proof}
Let $V^a$ be any representative of $\mathbf V$ and for brevity let $N=N_{g,\alpha}$.
From elliptic theory the equation 
\begin{equation}
\div_g\left( \frac{1}{2N}\ck_g W\right) = \kappa \extd\left( \frac{1}{N}\div_g(V)\right)
\end{equation}
admits a solution $W^a$ if and only if
\begin{equation}\label{eq:ortho}
\int_M \frac{1}{N} \div_g(V)\div_g(Q)\;\omega_g = 0
\end{equation}
for all conformal Killing fields $Q^a$,
in which case the solution $W^a$ is unique up to addition of a conformal Killing field.
Although $V^a$ need not satisfy condition \eqref{eq:ortho}, we claim that there
is a conformal Killing field $\widehat Q^a$ such that 
\begin{equation}\label{eq:ortho2}
\int_M \frac{1}{N} \div_g(V+\widehat Q)\div_g(Q)\;\omega_g = 0
\end{equation}
for all conformal Killing fields $Q^a$, and that 
$Q^a$ is unique up to addition of a proper Killing field.  Since proper
Killing fields are divergence-free, the right-hand side of \eqref{eq:drift}
is independent of the choice of admissible conformal Killing fields, 
as is the set of solutions of $\eqref{eq:drift}$.

Consider the functional
\begin{equation}
F(Q^a) = \int_M \frac1N \div_g(V+Q)^2\;\omega_g 
\end{equation}
on the finite-dimensional space $T_e \calQ_g$ of conformal Killing fields,
and observe that $\widehat Q^a$ is stationary for $F$ if and only if equation \eqref{eq:ortho}
holds. Moreover, since the highest order term of $F$ is quadratic and non-negative definite, the 
stationary points of $F$ correspond with its minimizers.

First suppose that $g_{ab}$ does not admit any (nontrivial) proper Killing fields.
Then every nontrivial conformal Killing field satisfies $\div Q \not\equiv 0$ and 
the quadratic term of $F$ is positive definite.  Hence $F$ admits a unique minimizer.  
On the other hand, if $g_{ab}$ admits a nontrivial space $\calX$ of proper Killing fields,
then $F$ descends to a functional on the quotient $T_e\calQ/\calX$ and its quadratic
order term is again positive definite.  Hence we pick up a minimizer of $F$ over
the conformal Killing fields, and it is unique up
to addition of a proper Killing field.

This establishes the main result up to independence of the solution set
with respect to the choice of representative of $\mathbf V$. 
Let $V^a$ and $\widetilde V^a$ be two representatives, so Lemma \ref{lem:drift_id}
implies that
\begin{equation}
V^a = \widetilde V^a + \widetilde Q^a + \widetilde E^a
\end{equation}
for some conformal Killing field $\widetilde Q^a$ and some divergence-free vector field $\widetilde E^a$.
Let $Q^a$ be a conformal Killing field  and let $W^a$ be a vector field such that
\begin{equation}
\div_g\left( \frac{1}{2N}\ck_g W\right) = \kappa\, \mathbf{d}\left( \frac{1}{N}\div_g( V+ Q)\right).
\end{equation}
We wish to show that there is a conformal Killing field $\ol Q^a$ such that
\begin{equation}\label{eq:drifthat}
\div_g\left( \frac{1}{2N}\ck_g W\right) = \kappa\, \mathbf{d}\left( \frac{1}{N}\div_g( \widetilde V+ \ol Q)\right)
\end{equation}
as well. Since
\begin{equation}
\div_g( V+ Q) = \div_g( \widetilde V+ \widetilde E + \widetilde Q + Q) = 
\div_g( \widetilde V + \widetilde Q + \widehat Q)
\end{equation}
we conclude that equation \eqref{eq:drifthat} holds with $\ol Q^a = \widetilde Q^a + Q^a$.
\end{proof}

Theorem \ref{thm:drift} provides a map $j^{\Drift}_\alpha$ from $\Drift_g$ to $T_g^*\calM/\calD_0$ as
follows. Given a drift $\mathbf V$, let $V^a$ be a representative and let  $Q^a$ and $W^a$ be a conformal Killing
field and vector field respectively solving \eqref{eq:drift}. We define 
\begin{equation}\label{eq:jalphadrift}
j^{\Drift}_\alpha (\mathbf V ) =  \left(g_{ab};\; \frac{1}{2N_{g,\alpha}}\ck_g W_{ab}, -\frac{2\kappa}{N_{g,\alpha}}\div_g(V+Q)\right)^*.
\end{equation}
Note that $j^{\Drift}_\alpha$ is well defined since $\div_g(V+Q)$ and $\ck_g W_{ab}$ are uniquely determined
by $\{g;\; 0, \div_g V\}$ even though $V^a$, $Q^a$ and $W^a$ need not be, and that 
equation \eqref{eq:drift} ensures that $j^{\Drift}_\alpha$ maps into $T^*_g\calM/\calD_0$,
not just $T^*_g\calM$.  The map $j^{\Drift}_\alpha$ is injective for if
\begin{equation}
\left(g_{ab};\; \frac{1}{2N_{g,\alpha}}\ck_g W_{ab}, -\frac{2\kappa}{N_{g,\alpha}}\div_g(V+Q)\right)^* = 0
\end{equation}
then $\div_g(V+Q)=0$ and hence the source drift 
$\mathbf V=\{g_{ab};\; 0,\div_g V\}$ satisfies
\begin{equation}
\{g_{ab};\; 0, \div_g V\} = \{g_{ab};\; \ck_g Q, \div_g(V+Q)\} =  \{g_{ab};\; 0, \div_g (V+Q)\} = 0.
\end{equation}
So $\Im j^{\Drift}_\alpha$ is isomorphic to $\Drift_g$.  The following result
shows that $\Im j^{\Drift}_\alpha$ complements the conformal and volumetric momenta,
which formalizes our earlier claim that
solutions of the momentum constraint can be parameterized by the selection
of a conformal momentum, a volumetric momentum, and a drift.
\begin{proposition}\label{prop:tstar}
$$\label{eq:tstar}
T^*_g\calM/\calD_0 = T^*_{[g]} \calC/\calD_0 \; \oplus\; T^*_{\omega_g} \calV/\calD_0 \;\oplus\; \Im(j^{\Drift}_\alpha).
$$
\end{proposition}
\begin{proof}
Suppose $(g_{ab};\; A_{ab}, -2\kappa\tau)^*\in T^*_{g}\calM/\calD_0$.
From York splitting there are vector fields $W^a$ and $V^a$ solving equation \eqref{eq:drift} as well as
a TT-tensor $\sigma_{ab}$ and a constant $\tau^*$ such that
\begin{equation}
\begin{aligned}
A_{ab} &= \sigma_{ab} + \frac{1}{2N_{g,\alpha}} \ck_g W_{ab}\\
\tau &= \tau^* + \frac{1}{N_{g,\alpha}} \div_g V.
\end{aligned}
\end{equation}
So
\begin{equation}
(g_{ab};\; A_{ab}, -2\kappa\tau)^* =
(g_{ab};\; \sigma_{ab},0)^*+(g_{ab};\; 0,\tau^*)^*+
(g_{ab};\; (1/2N_{g,\alpha})\ck_g W_{ab},(1/N_{g,\alpha})\div_g V)^*.
\label{eq:ttdecomp}
\end{equation}
Since $(g_{ab};\;A_{ab},-2\kappa\tau)^*$ solves the momentum constraint \eqref{eq:momTTF},
$W^a$ and $V^a$ solve equation \eqref{eq:drift} and hence
\begin{equation}
(g_{ab};\; (1/2N_{g,\alpha})\ck_g W_{ab},(1/N_{g,\alpha})\div_g V)^* \in \Im j^{\Drift}_\alpha.
\end{equation}
Equation \eqref{eq:ttdecomp} therefore exhibits $(g_{ab};\; A_{ab}, -2\kappa\tau)^*$
as the sum of a conformal momentum, a volumetric momentum, and term in
the image of $j^{\Drift}_\alpha$.

To establish the direct sum decomposition \eqref{eq:tstar} we need only show that
the summands are mutually transverse.  Now if $W^a$ and $V^a$ solve equation \eqref{eq:drift}
and either of $\ck_g W_{ab}$ or $\div_g V$ vanishes, an integration by parts 
exercise shows the other must as well. Hence $\Im j^{\Drift}_\alpha$ is transverse
to $T^*_{[g]}\calC/\calD_0$ and $T^*_{\omega_g}\calV/\calD_0$, which are obviously
transverse to each other.
\end{proof}

We now show that Theorem \ref{thm:drift} can be understood as describing a map $R_\alpha$ 
from drifts to drifts.
Given a drift $\mathbf V$, let $V^a$ be any representative and let $W^a$ be a solution 
of equation \eqref{eq:drifteq}.  We define
\begin{equation}\label{eq:Jalpha}
R_\alpha(\mathbf V) = \{g_{ab};\; \ck_g W_{ab}, 0\} = \{g_{ab};\;0, -\div_g W\} \in \Drift_g
\end{equation}
and note that $R_\alpha$ is well-defined since $W^a$ is uniquely determined up to
adding a conformal Killing field.
Proposition \ref{prop:tstar} shows that solutions of the momentum
constraint are parameterized by a conformal momentum, a volumetric
momentum, and a pair $(\mathbf W, \mathbf V)$ of drifts that are joined at
the hip by $\mathbf W = R_\alpha(\mathbf V)$.  

It turns out that $R_\alpha$ is invertible, and $\mathbf W$
determines $\mathbf V$ as well. The reverse process proceeds as follows: 
let $W^a$ be a vector field such that $\mathbf W = \{g_{ab};\; \ck_g W_{ab},0\}$
and attempt to solve
\begin{equation}\label{eq:driftrev}
\kappa\, \mathbf{d} \left( \frac{1}{N_{g,\alpha}} \div_g V \right) = \div_g\left(\frac{1}{2N_{g,\alpha}}\ck_g W\right)
\end{equation}
for $V^a$.  Now if equation \eqref{eq:driftrev} admits a solution, we can multiply the equation
by an arbitrary divergence free vector field $E^a$ and integrate by parts to find
\begin{equation}\label{eq:compat}
\int_M \frac{1}{2N_{g,\alpha}} \ip<\ck_g W,\ck_g E>\;\omega_g =0,
\end{equation}
which poses a compatibility condition on $W^a$.
We will show that equation \eqref{eq:driftrev} admits a solution $V^a$ if and
only if the compatibility equation is satisfied,
and the solution is unique up to adding a divergence free vector field.
Hence $\div_g V$ and the drift
\begin{equation}
\mathbf V = \{g_{ab};\; 0,\div_g V\}
\end{equation}
is uniquely determined by $\mathbf W$ once the compatibility condition $\eqref{eq:compat}$ is met.  
In general an arbitrary representative $W^a$ of $\mathbf W$ will fail the compatibility condition, 
but we will show that we can adjust $W^a$ by adding a divergence-free vector field to remedy this
deficiency. Note that adding a divergence-free vector field does not change the drift represented by $W^a$.
The following three propositions contain the technical tools needed to carry out this procedure;
we start by showing that
equation \eqref{eq:driftrev} is solvable if the compatibility condition is met.

\begin{proposition}\label{prop:revsolve}
Suppose $g_{ab}$ is a metric, $N$ is a positive function and $\eta_a$ is a 1-form.
The equation
\begin{equation}\label{eq:driftrevform}
\kappa\, \mathbf{d}\left( \frac{1}{2N} \div_g V\right) = \eta
\end{equation}
admits a smooth solution $V^a$ if and only if
\begin{equation}\label{eq:compat2}
\int_M \ip<\eta,E>\;\omega_g =0
\end{equation}
for all $g$-divergence-free vector fields $E^a$, in which case
$V^a$ is determined up to addition of a (smooth) divergence-free vector field.
\end{proposition}
\begin{proof}
Multiplying equation \eqref{eq:compat2} by a divergence free vector field and integrating by parts
shows that equation \eqref{eq:compat2} is necessary for a solution to exist, and we henceforth
assume $\eta_a$ satisfies this condition.  Applying the Helmholtz-Hodge decomposition
we can write
\begin{equation}\label{eq:helm}
\eta_a = \nabla_a f + \mu_a
\end{equation}
where $f$ is a function, $\mu_a$ is divergence-free, and both are smooth.  
Multiplying equation \eqref{eq:helm}
by $\mu^a$, integrating, and using the compatibility condition we find that 
$\mu_a\equiv 0$ and hence
\begin{equation}
\eta_a = \nabla_a f.
\end{equation}

Since $\eta = \mathbf d f$,
to solve equation \eqref{eq:driftrevform} it suffices to find a smooth vector field $V^a$ 
and a constant $c$ such that
\begin{equation}\label{eq:prelap}
\div_g V = \frac{2N}{\kappa} (f+c).
\end{equation}
We pick $c$ so that
\begin{equation}
\int_M \frac{2N}{\kappa} (f+c)\omega_g = 0
\end{equation}
and find a function $u$ so that
\begin{equation}
\Delta_g u = \frac{2N}{\kappa} (f+c).
\end{equation}
Equation \eqref{eq:prelap} is then solved taking $V^a = \nabla^a u$,
and we see that $V^a$ is smooth.  If
we add a smooth divergence-free vector field to $V^a$ we obtain another
solution, and we now show that all smooth solutions are obtained this way.

Suppose that $V^a$ and $\widehat V^a$ are two solutions. It follows that
\begin{equation}
\mathbf d \left( \frac{1}{N} \div_g(V-V^a)\right) = 0
\end{equation}
and hence
\begin{equation}
\div_g(V-\widehat V) = c N
\end{equation}
for some constant $c$.  Integrating over the manifold we conclude $c=0$ and hence
$V^a$ and $\widehat V^a$ differ by a divergence-free vector field.  And if
$V^a$ and $\widehat V^a$ are both smooth, so is the difference.
\end{proof}

Adjusting the right-hand of equation \eqref{eq:driftrev} to meet
the compatibility condition involves adding a suitable divergence-free
vector field $E^a$, and we will see that $E^a$ is the solution
of a certain Stokes-like PDE. Let $\Lie_g$ be the Killing operator
of $g_{ab}$, so $\Lie_g X_{ab} = \nabla_a X_b + \nabla_b X_a$,
let $\Lie^*_g=-2\div_g$ be its adjoint, and let 
\begin{equation}
\stokeslap_{g,N} = \Lie^*_g\;\; 1/(2N)\; \Lie_g.
\end{equation}
Given forcing terms $\eta_a$ and $h$ we form the Stokes equations
\begin{subequations}\label{eq:stokes}
\begin{alignat}{2}
\stokeslap_{g,N} E &= \eta + \extd p \label{eq:stokes1}\\
\div_g E &= h
\end{alignat}
\end{subequations}
where the unknowns are $E^a$ and the pressure $p$.  In practice
we will usually take $h=0$ so that $E^a$ is divergence-free, but it will
aid a regularity bootstrap to consider a non-homogeneous term here.

Each of $\eta_a$ and $h$ must satisfy a compatibility condition in
order for system \eqref{eq:stokes} to be solvable.
Multiplying the first equation of system \eqref{eq:stokes}
by a Killing field $K^a$ we find 
\begin{equation}
\int_M \eta_a K^a \;\omega_g=0,
\end{equation}
and integrating the second equation of system \eqref{eq:stokes}
we have
\begin{equation}
\int_M h\;\omega_g = 0.
\end{equation}
These compatibility conditions are sufficient for there to
exist a solution of the Stokes system.

\begin{proposition}\label{prop:revcompat}
Let $g_{ab}$ be a smooth metric and let $N$ be a smooth positive function.
Let $\eta_a$ be a 1-form in $W^{-1,2}$ that is $L^2$ orthogonal
to the proper Killing fields and let $h$ be a function in $L^2$
that is $L^2$ orthogonal to the constants.  Then there
exists a vector field $E^a\in W^{1,2}$ and a function $p\in L^2$
solving the Stokes system \eqref{eq:stokes} in the sense of
distributions, and the solution is unique up to adding
a Killing field to $E^a$  and a constant to $p$.
\end{proposition}
\begin{proof}
We can reduce to the case where $h=0$ by solving
\begin{equation}
\Lap f = h
\end{equation}
for a function $f\in W^{2,2}$, which is possible since $h$ is orthogonal to the constants,
and writing $E^a = \widehat E^a + \nabla^a f$ where $\widehat E^a$ is an
unknown divergence-free function.  Since $\nabla f\in W^{1,2}$
we have $\stokeslap_{g,N}\nabla f\in W^{-1,2}$ and we see that $(E^a,p)$
solves the original system if and only if $(\widehat E^a,p)$ solves
the system with $\eta$ replaced with $\eta - \stokeslap_{g,N} \nabla f$
and $h$ replaced with $0$.  Henceforth we assume $h=0$, and
we seek a divergence-free vector field $E^a$ and a pressure $p$
solving \eqref{eq:stokes1}.

First suppose $g_{ab}$ has no nontrivial proper Killing fields, and let
$J^{1,2}$ be subspace of divergence-free $W^{1,2}$ vector fields.  
For $E^a$ and $F^a\in J^{1,2}$, define
\begin{equation}
A(E^a,F^a) = \int_M \frac{1}{2N} \ip<\Lie_g E, \Lie_g F>_g\; \omega_g.
\end{equation}
We claim that there is a constant $c$ such that $A(E^a,E^a)\ge c \int_M |E|_g^2\;\omega_g$
for all $E^a\in J^{1,2}$.  Suppose not.  Then we can find a sequence of divergence-free vector fields
$\{E^a_{(k)}\}_k$ such that $A(E^a_{(k)},E^a_{(k)})\ra 0$ 
and such that each $E^a_{(k)}$ has norm 1 in in $L^2$.
Recall Korn's inequality \cite{Chen:2002cz}, which implies
that there is a constant $C$ such that
\begin{equation}\label{eq:kornmod}
||E||_{W^{1,2}}^2 \le C\left[ \int_M \ip<\Lie_g E, \Lie_g E>_g\; \omega_g + ||E||_{L^2}^2\right].
\end{equation}
Since $N$ is bounded above, a similar inequality holds replacing 
$\int_M \ip<\Lie_g \cdot, \Lie_g \cdot>_g\omega_g$ with $A$ and hence the sequence $\{E^a_{(k)}\}_k$ 
is bounded in $W^{1,2}$.  So a subsequence converges weakly
in $W^{1,2}$ and strongly in $L^2$ to a limit $\widetilde E^a$.  The quadratic
form $E\mapsto A(E,E)$ is non-negative definite, so it is weakly lower semicontinuous.
Hence the weak limit satisfies $A(\widetilde E,\widetilde E)=0$ and is a Killing field.
Since $||\widetilde E||_{L^2} = 1$ as well, $g_{ab}$ admits a nontrivial Killing field, 
which is a contradiction.

We have now established that $||E||_{L^2}^2$ is controlled by $A(E,E)$, and
it then follows from inequality \eqref{eq:kornmod}
that there is a constant $c$ such that
\begin{equation}
A(E^a,E^a) \ge c||E||_{W^{1,2}}^2
\end{equation}
for all $E\in J^{1,2}$.  So $A$ is coercive over $J^{1,2}$ and the Lax-Milgram
theorem implies there is a unique $E^a\in J^{1,2}$ such that
\begin{equation}\label{eq:stokesweak}
\int_M \frac{1}{2N} \ip<D E, D F>_g\;\omega_g = \int_M \ip<\eta, F>\;\omega_g
\end{equation}
for all $F^a\in J^{1,2}$.  Now
\begin{equation}\label{eq:isap}
F^a\mapsto \int_M \frac{1}{2N} \ip<D E, D F>_g\;\omega_g -\int_M \ip<\eta, F>\;\omega_g
\end{equation}
is a continuous functional on $W^{1,2}$ that vanishes on $J^{1,2}$.  From
the Helmholtz-Hodge decomposition of $W^{-1,2}$ there is a unique weakly divergence free 
$G^a$ in $W^{-1,2}$ and function $p$ in $L^2$, uniquely determined up to a constant,
such that functional \eqref{eq:isap} is equal to
\begin{equation}
G^a + \nabla^a p.
\end{equation}
But since this functional vanishes on $J^{1,2}$ we conclude that $G^a=0$
and hence
\begin{equation}\label{eq:pressure}
\int_M \frac{1}{2N_{g,\alpha}} \ip<D E, D F>_g\;\omega_g -\int_M \ip<\eta, F>\;\omega_g=
\int_M p \div_g F\; \omega_g
\end{equation}
for all $F^a\in W^{1,2}$, which proves existence if $g_{ab}$ admits no nontrivial Killing fields.
Moreover, if $(\widehat E^a, \widehat p)$ is any solution of the Stokes system, we see that $\widehat E^a$
also satisfies \eqref{eq:stokesweak} and therefore equals $E^a$.  But then $\widehat p$ satisfies
the equation \eqref{eq:pressure} for the pressure and is therefore equal to $p$ plus a constant.

If $g_{ab}$ admits nontrivial Killing fields, we replace the space $J^{1,2}$ in the proof above with
the $L^2$ orthogonal complement in $J^{1,2}$ of the Killing fields; call this new space $\widehat J^{1,2}$.  
The proof above
then finds $E^a\in \widehat J^{1,2}$ such that equation \eqref{eq:stokesweak} holds for all $F^a$ in 
$\widehat J^{1,2}$. Since $\eta^a$ is $L^2$ orthogonal to the proper Killing fields, 
equation \eqref{eq:stokesweak} holds
for all $F^a\in J^{1,2}$ and the remainder of the proof continues without change.
\end{proof}
Proposition \ref{prop:revcompat} establishes existence of weak solutions 
of the Stokes system, and we now show that when the forcing terms are smooth,
so are the solutions.
\begin{proposition}\label{prop:stokesreg}
In Proposition \ref{prop:revcompat}, if $(\eta_a,h)\in W^{k-2,2}\times W^{k-1,2}$
for some integer $k\ge 2$, then $(E^a,p)\in W^{k,2}\times W^{k-1,2}$.  In particular,
if $\eta_a$ and $h$ are smooth, so are $E^a$ and $p$.
\end{proposition}
\begin{proof}
Suppose $(\eta_a,h)\in L^2\times W^{1,2}$.  Applying the divergence to
equation \eqref{eq:stokes1} we find that $p$ is a weak solution of
\begin{equation}\label{eq:lapP}
\begin{aligned}
\Lap p &= [\div_g,\stokeslap_{g,N}] E + \stokeslap_{g,N} \div_g E - \div_g \eta \\
&= [\div_g,\stokeslap_{g,N}] E + \stokeslap_{g,N} h - \div_g \eta.
\end{aligned}
\end{equation}
Since $[\div_g,\stokeslap_{g,N}]$ is a second-order operator and since $E^a\in W^{1,2}$,
the first term on the right-hand side of equation \eqref{eq:lapP} belongs to $W^{-1,2}$.
It is easy to see that the remaining terms on the right-hand side of equation \eqref{eq:lapP}
also belong to $W^{-1,2}$ as well and hence $p\in W^{1,2}$.  But then the right-hand side of
equation \eqref{eq:stokes1} is in $L^2$ and since $\stokeslap_{g,N}$ is elliptic, we conclude that
$E^a\in W^{2,2}$.

To obtain higher regularity, we proceed by a bootstrap.  For example, suppose $\eta_a\in W^{1,2}$
and $h\in W^{2,2}$.  Let $\partial$ be a first order operator.  Then $\widehat E^a=\partial E^a$ and 
$\widehat p=\partial p$ belong to $W^{1,2}$ and $L^2$ respectively and satisfy
\begin{equation}\label{eq:stokeshat}
\begin{aligned}
\stokeslap_{g,N} \widehat E &= \partial \eta +[\stokeslap_{g,N},\partial]\, E -[\extd,\partial]\, p+ \extd \widehat p \\
\div_g \widehat E &= [\div_g, \partial]\, E + \partial h
\end{aligned}
\end{equation}
where $[\cdot,\cdot]$ is the commutator.
Since $(E^a,p)\in W^{2,2}\times W^{1,2}$, and since $(\eta,h)\in W^{1,2}\times W^{2,2}$,
we see that the right-hand sides of equations \eqref{eq:stokeshat} belong
to $L^2$ and $W^{1,2}$ respectively.  Hence by our previous result, 
$(\partial E^a,\partial p) \in W^{2,2}\times W^{1,2}$.  
So $(E^a,p)\in W^{3,2}\times W^{2,2}$,
and the remainder of the bootstrap continues similarly.
\end{proof}

From Propositions \ref{prop:revsolve} and \ref{prop:revcompat} we obtain the following
analogue of Theorem \ref{thm:drift}.
\begin{theorem}\label{thm:driftrev}
Suppose $g_{ab}$ is a metric, $\alpha$ is a lapse form, and $\mathbf W\in \Drift_g$.
Let $W^a$ be any vector field such that
\begin{equation}
\mathbf W  = \{g_{ab};\; \ck_g W_{ab}, 0 \} = \{g_{ab};\; 0, -\div_g W\}.
\end{equation}
Then there is a divergence-free vector field $E^a$, 
unique up to addition of a proper Killing field, such that
\begin{equation}\label{eq:drifteqrev}
\kappa\, \mathbf{d}\left( \frac{1}{N_{g,\alpha}}\div_g(V)\right) = 
\div_g\left( \frac{1}{2N_{g,\alpha}}\ck_g (W+E)\right).
\end{equation}
admits a solution $V^a$.  The solution $V^a$ is unique up to addition of a divergence-free
vector field, and this space of solutions does not depend on the choice of $W^a$
or on the choice of divergence free vector field $E^a$ such that equation \eqref{eq:drifteqrev} is solvable.
\end{theorem}
\begin{proof}
From Proposition \ref{prop:revsolve} we know that equation \eqref{eq:drifteqrev} admits a solution 
so long as
\begin{equation}
\int_M \ip< \div_g\left( \frac{1}{2N}\ck_g (W+E)\right), F>_g\;\omega_g=0
\end{equation}
for all divergence-free vector fields $F^a$.  Thinking of $W^a$ as fixed and $E^a$
as an unknown vector field we see that $E^a$ satisfies
\begin{equation}
\int_M \ip< \div_g\left( \frac{1}{2N}\ck_g E\right), F>_g\;\omega_g=
-\int_M \ip< \div_g\left( \frac{1}{2N}\ck_g W\right), F>_g\;\omega_g=0
\end{equation}
for all divergence-free vector fields.  Since $E^a$ is divergence-free,
$\ck_g E_{ab} = \Lie_g E$ and $E^a$ is a weak solution of the Stokes equation
\begin{equation}\label{eq:Ecompat}
D^*_g\left(\frac{1}{2N}\Lie_g E\right) = -\div_g\left(\frac{1}{2N} \ck_g W\right) + \nabla p
\end{equation}
for some function $p$.  Proposition \ref{prop:revcompat} shows that there is a solution
of \eqref{eq:Ecompat}, and that it is unique up to addition of a proper Killing field.

Thus we have shown there is a divergence-free vector field $E^a$, unique up to addition of
a proper Killing field, such that equation \eqref{eq:drifteqrev} admits a solution $V^a$, 
and Proposition \ref{prop:revsolve} shows that the solution is unique up to addition
of a divergence-free vector field.  The proof that this space of solutions
is independent of the choice of $W^a$ is analogous to the same step of Theorem \ref{thm:drift}.
\end{proof}

\section{Drift Velocity, Drift Momentum, and Drift Kinetic Energy}\label{sec:drift-vmke}

We saw in Section \ref{sec:drift-mom} that the map 
$k_\alpha: T^*_g\calM/\calD_0 \ra T_g\calM/\calD_0$ described in  
diagram \ref{diag:dl-legendre-extended-M} can fail to be an isomorphism.
In terms of the decompositions
\begin{equation}
T^*_g\calM/\calD_0 = T^*_{[g]} \calC/\calD_0\; \oplus\; T^*_{\omega_g} \calV/\calD_0\; \oplus\; \Im(j^{\Drift}_\alpha)
\end{equation}
and 
\begin{equation}
T_g\calM/\calD_0 \approx
T_{[g]}\calC/\calD_0 \;\oplus\; T_{\omega_g}\calV/\calD_0\; \oplus\; \Drift_g.
\end{equation}
given by Propositions \ref{prop:tstar} and \ref{prop:CVD} respectively, 
a computation shows
\begin{equation}
k_\alpha(\bfsigma, -2\kappa\tau^*, j^{\Drift}_\alpha(\mathbf V)) = ((j_\alpha^\expspace\calC)^{-1}(\bfsigma), (j_\alpha^\expspace\calV)^{-1}(-2\kappa\tau^*), \mathbf V-\mathbf W)
\end{equation}
where $\mathbf W = R_\alpha(\mathbf V)$ and $R_\alpha$ is the map described 
in equation \eqref{eq:Jalpha}.  Since the Legendre transformations
$j_\alpha^\expspace\calC$ and $j_\alpha^\expspace\calV$ are isomorphisms, we see that $k_\alpha$ fails to be an isomorphism
precisely when $\mathbf V \mapsto \mathbf V - R_\alpha(\mathbf V)$ fails
to be an isomorphism.  We address this difficulty 
by treating the pair $(\mathbf W, \mathbf V)$, linked by the equation $\mathbf W=R_\alpha(\mathbf V)$,
as the drift component of motion rather than the
difference $\mathbf V -\mathbf W$.  Since $R_\alpha$ is invertible, we can parameterize 
the pairs $(\mathbf W, \mathbf V)$ in terms of either component, 
and we will refer to $\mathbf W$ as \define{conformal drift} and $\mathbf V$ as \define{volumetric drift}.

Suppose that we parameterize drift pairs in terms of their volumetric component.
To this end, 
given a lapse form $\alpha$ and a shift $X^a$ we define a projection
\begin{equation}\label{eq:piaXD1}
\pi_{\alpha,X^a}^{\Drift} T_g \calM \ra \Drift_g
\end{equation}
as follows.  Given $(g_{ab};\; u_{ab},\beta)\in T_g\calM$ we
apply volumetric York decomposition to write
\begin{equation}\label{eq:piaXD2}
\beta = N_{g,\alpha} \tau^* + \div_g(V+X)
\end{equation}\label{eq:piaXD3}
for a unique constant $\tau^*$ and a vector field $V^a$ that is unique up to
adding a divergence-free vector field.  Then
\begin{equation}
\pi_{\alpha,X^a}^{\Drift}( (g_{ab};\; u_{ab},\beta) ) = \{g_{ab};\; 0, \div_g V\}.
\end{equation}
We claim that the diagram
\begin{equation}\label{diag:dl-legendre-drift}
\begin{gathered}
\xymatrix{
 & \calK \ar@{<->}[dr] \ar@{<->}[dl]_{(\alpha,X^a)} \\
 T_g\calM 
\ar@<2pt>@{<-}[rr]
\ar@<-2pt>@{->}[rr]_{i_{\alpha,X^a} }
& & T^*_g \calM\\
\Drift_g \ar@{<-}[u]^{\pi^{\Drift}_{\alpha,X^a}}
\ar@<2pt>@{<-}[rr]
\ar@<-2pt>@{->}[rr]_{j^{\Drift}_{\alpha} }
& & \Im(j^{\Drift}_{\alpha})
\ar@<0pt>[u]
}
\end{gathered}
\end{equation}
commutes, except that traversing the bottom loop starting from
the middle row is a projection.  It is enough to establish the following.
\begin{proposition}
The map obtained from diagram \eqref{diag:dl-legendre-drift} starting
at $\Drift_g$ and traversing the bottom loop back to $\Drift_g$ is the identity.
\end{proposition}
\begin{proof}
Let $\mathbf V =\{g_{ab};\;0,\div_g V\}\in \Drift_g$.  From the definition of $j_\alpha^{\Drift}$
in equation \eqref{eq:jalphadrift} we find
\begin{equation}
j^{\Drift}_\alpha (\mathbf V ) =  
\left(g_{ab};\; \frac{1}{2N_{g,\alpha}}\ck_g W_{ab}, \frac{1}{N_{g,\alpha}}\div_g(V+Q)\right)
\end{equation}
where the vector field $W^a$ and the conformal Killing field 
$Q^a$ are determined by Theorem \ref{thm:drift}. Applying $i^{-1}_{\alpha,X^a}$
from equation \eqref{eq:i_inv} we arrive at
\begin{equation}
(g_{ab};\ck_g(W+X)_{ab}, \div_g(V+Q+X))
\end{equation}
and applying $\pi^{\Drift}_{\alpha,X^a}$ we complete the loop at
\begin{equation}
\{g_{ab};\;0,\div_g(V+Q)\}.
\end{equation}
Since $Q^a$ is a conformal Killing field,
\begin{equation}
\{g_{ab};\;0,\div_g(Q)\}=\{g_{ab};\;\ck_g Q,\div_g(Q)\} = 0
\end{equation}
and hence
\begin{equation}
\{g_{ab};\;0,\div_g(V+Q)\}= \{g_{ab};\;0,\div_g(V)\} = \mathbf V.
\end{equation}
Thus traversing the loop is the identity.
\end{proof}

Drift velocity and momentum are defined following the pattern seen previously
for conformal and volumetric quantities.  Given $(g_{ab},K_{ab})\in \calM\times \calK$
we form diagram \eqref{diag:dl-legendre-drift} and descend the left hand side 
from $\calK$ to $\Drift_g$.  If $K_{ab}$ has mean curvature $\tau$, we apply
volumetric York splitting to write
\begin{equation}\label{eq:tausplit}
\tau =  \tau^* + \frac{1}{N_{g,\alpha}}\div V,
\end{equation}
and equations \eqref{eq:ADMmomTTF}, \eqref{eq:piaXD2} and \eqref{eq:piaXD3}
show that the drift velocity is
\begin{equation}
\{g_{ab};\; 0, \div_g V\}.
\end{equation}
Note that although both maps on the left-hand side of diagram \eqref{diag:dl-legendre-drift}
depend on the shift $X^a$, their composition does not and only depends  on the lapse form $\alpha$.  
Drift momentum is obtained from drift velocity by applying $j_\alpha^{\Drift}$.

\begin{definition} Let $(g_{ab},K_{ab})\in \calM\times \calK$ and let $\alpha$
be a lapse form.  The \define{drift velocity} of $(g_{ab},K_{ab})$, 
as measured by $\alpha$, is 
\begin{equation}
v^{\Drift}_\alpha (g_{ab}, K_{ab}) = \{g_{ab};\; 0, \div_g V\}
\end{equation}
where $V^a$ is obtained by the splitting \eqref{eq:tausplit}
of $\tau=g^{ab}K_{ab}$.  The \define{drift momentum}
of $(g_{ab},K_{ab})$,  as measured by $\alpha$, is 
\begin{equation}
m^{\Drift}_\alpha (g_{ab}, K_{ab}) = \left\{g_{ab};\; \frac{1}{2N_{g,\alpha}}\ck_g W_{ab}, -\frac{2\kappa}{N_{g,\alpha}}\div_g (V+Q) \right\}^*
\end{equation}
where $W^a$ and $Q^a$ are the vector field and conformal Killing field 
obtained from Theorem \ref{thm:drift}.
\end{definition}

Since $\Drift_g\subseteq T_g \calM/\calD_0$, every element of $T_g^* \calM/\calD_0$
defines a functional on $\Drift_g$.  Since $j_\alpha^{\Drift}: \Drift_g \ra  \calM/\calD_0$,
we can therefore consider $j_\alpha^{\Drift}$ as a map into $(\Drift_g)^*$ and it is then
natural to identify a Lagrangian such that $j_\alpha^{\Drift}$ is its Legendre transformation.

\begin{definition}
Given $\mathbf V \in \Drift_g$ the \define{drift kinetic energy} of $\mathbf V$, as measured
by $\alpha$, is
\begin{equation}
\begin{aligned}
K^{\Drift}_\alpha(\mathbf V) = -\int \kappa(\div_g V)^2\; \alpha \\
\end{aligned}
\end{equation}
where $V^a$ is any representative of $\mathbf V$ 
such that
\begin{equation}\label{eq:driftcompat2}
\int \div_g(V)\div_g Q\;\alpha=0
\end{equation}
for all conformal Killing fields $Q^a$; note that
Theorem \ref{thm:drift} ensures that $\div_g V$ (and hence drift kinetic energy)
is uniquely determined by $\mathbf V$.
\end{definition}

To show that $K^{\Drift}_\alpha$ is a Lagrangian, one ought to demonstrate
a configuration space such that $K^{\Drift}_\alpha$ is a function on its tangent bundle.
Clearly each tangent space must be isomorphic to $\Drift_g$, but the right choice of 
base space is not clear. So we content ourselves by restricting our attention 
to each fibre $\Drift_g$ of the presumed total space and show that $j_\alpha^{\Drift}$
is the Legendre transformation of $K^{\Drift}_\alpha$ on that fibre.
Consider a path $\mathbf V(t) = \{g_{ab};\; 0, \div_g V(t)\}$ 
of drifts where $V^a(t)$ is a path of vector fields satisfying 
condition \eqref{eq:driftcompat2}.  Since $V^a(t)$ satisfies the compatibility
condition, there exists a path of vector fields $W^a(t)$ with
\begin{equation}
\div_g\left(\frac{1}{2N_{g,\alpha}}\ck_g W\right) = \kappa \extd\left(\frac1{N_{g,\alpha}}\div_g V\right).
\end{equation}
Then, recalling equations \eqref{eq:MD0duality} and \eqref{eq:Mduality}, we find
\begin{equation}\label{eq:driftleg}
\begin{aligned}
\frac{d}{dt}K^{\Drift}_\alpha(\mathbf V) &= -2\kappa\int_M \div_g V \div_g \dot V\;\alpha\\
&= \int_M \left(\frac{-2\kappa}{N_{g,\alpha}} \div_g V \right) \div_g \dot V\;\omega_g\\
&= \left< 
\left\{g_{ab};\;  \frac{1}{2N_{g,\alpha}}\ck_g W_{ab},
-2\kappa\frac{1}{N_{g,\alpha}} \div_g V \right\}^*,
\{g_{ab};\; 0, \div_g \dot V\} \right>\\
&= \ip< j^{\Drift}_\alpha(\mathbf V), \dot{\mathbf V}>
\end{aligned}
\end{equation}
where the various dependencies on $t$ in equation \eqref{eq:driftleg} have been suppressed.
Thus $K_\alpha^{\Drift}$ is the desired Lagrangian.

The preceding discussion was based on parameterizing pairs $(\mathbf W, \mathbf V)$
with $\mathbf W = R_{\alpha} (\mathbf V)$ in terms of volumetric drift
$\mathbf V$.  If we use conformal drift $\mathbf W$
instead we obtain a dual notion of drift velocity and kinetic energy
which we now summarize briefly.
The drift velocity of $(g_{ab},K_{ab})$ is 
\begin{equation}
\widehat v^{\Drift}_{\alpha}(g_{ab},K_{ab}) =  \{g_{ab};\; \ck_g W_{ab}, 0\}
\end{equation}
where $W^a$ is any vector field arising from the conformal York decomposition
of the trace free part $A_{ab}$ of $K_{ab}$:
\begin{equation}
A_{ab} = \sigma_{ab}+\frac{1}{2N_{g,\alpha}} \ck_g W_{ab}.
\end{equation}
The drift momentum is
\begin{equation}
\widehat m^{\Drift}_{\alpha}(g_{ab},K_{ab}) =  
\left\{g_{ab};\; \frac{1}{2N_{g,\alpha}}\ck_g (W+E)_{ab}, -\frac{2\kappa}{N_{g,\alpha}}\div_g (V) \right\}^*
\end{equation}
where the vector field $V^a$ and the divergence-free vector field $E^a$
are provided by Theorem \ref{thm:driftrev}.  
The drift kinetic energy of $\mathbf W$ is
\begin{equation}
\widehat K_{\alpha}^{\Drift}(\mathbf W) = \int_M \frac{1}{4}|\ck_g W|^2 \;\alpha.
\end{equation}
where $W^a$ is any representative of $\mathbf W$ satisfying the compatibility condition \eqref{eq:compat}
It is easy to see that
\begin{equation}
m_\alpha^{\Drift}(g_{ab},K_{ab}) = \widehat m_\alpha^{\Drift}(g_{ab},K_{ab})
\end{equation}
if and only if the pair $(g_{ab},K_{ab})$ satisfies the momentum constraint, so the drift momentum
of a solution of the constraint equations is well-defined regardless of which factor $\mathbf W$ or $\mathbf V$
we use to  parameterize drift velocity.  
The choice of using $\mathbf W$ or $\mathbf V$ is one of emphasis
between conformal and volumetric motion: we can apparently measure drift of the conformal class
relative to the normal direction, or we can measure drift of the volume form relative to the normal direction, 
but these two motions are linked by the momentum constraint and are not independent.  This class of linked motion can be parameterized in terms of $\Drift_g$, but we have two distinct natural parameterizations.

Continuing with our choice to parameterize drift motion by its volumetric component
we have the following.
\begin{theorem} \label{thm:mc}
Let $g_{ab}\in\calM$, let $\alpha$ be a lapse form, and let $X^a$ be a shift.
The map 
\begin{equation}
j_\alpha: T_{[g]}\calC/\calD_0\;\oplus\; T_{\omega_g}\calV/\calD_0\;\oplus\; \Drift_g\ra T^*_g \calM/\calD_0
\end{equation}
defined by 
\begin{equation}
j_\alpha = j_\alpha^\expspace\calC \oplus j_\alpha^\expspace\calV \oplus j_\alpha^{\Drift}
\end{equation}
is an isomorphism.  Moreover, consider the diagram
\begin{equation}\label{diag-j}
\begin{gathered}
\xymatrix{
  &\calK \ar@{<->}[drr] \ar@{<->}[dl]_{(\alpha,X^a)} \\
 T_g \calM \ar[d]_{\pi_*}
\ar@<2pt>@{<-}[rrr]
\ar@<-2pt>@{->}[rrr]_{i_{\alpha,X^a} }
& & & T^*_g \calM \\
T_{[g]}\calC/\calD_0\;\oplus\; T_{\omega_g}\calV/\calD_0\;\oplus\; \Drift_g
\ar@<2pt>@{<-}[rrr]
\ar@<-2pt>@{->}[rrr]_{j_{\alpha} }
& & &
T^*_g \calM/\calD_0 \ar[u]
}
\end{gathered}
\end{equation}
where the first two components of $\pi_*$ are the natural pushforwards
and the third component is $\pi_{\alpha,X^a}^{\Drift}$.  
Traversing the bottom loop of diagram \eqref{diag-j} starting
on the bottom row is the identity, and traversing the outermost loop starting at $\calK$ is
a projection onto second fundamental forms $K_{ab}$ such that $(g_{ab},K_{ab})$ solves
the momentum constraint.  
\end{theorem}
\begin{proof}
That $j_\alpha$ is an isomorphism follows from the fact that $j_\alpha^\expspace\calC$, $j_\alpha^\expspace\calV$, and
$j_\alpha^{\Drift}$ are isomorphisms, along with Proposition \ref{prop:tstar}.  That traversing
the bottom loop starting from the bottom row is the identity follows from the same fact
for diagrams \eqref{diag:dl-legendre-conf}, \eqref{diag:dl-legendre-conf-vol} and \eqref{diag:dl-legendre-drift}.
As a consequence, traversing the bottom loop starting from the middle row must be a projection.
Since the maps $i_{\alpha,X^a}^{-1}$, $\pi_*$, and $j_\alpha$ are surjective, the image of
$T_g^*\calM$ after traversing the bottom loop is the image of $T_g^*\calM/\calD_0$ under
the natural pullback, i.e., the divergence-free elements.  Hence traversing the outermost loop starting at 
$\calK$ is a projection onto the second fundamental forms with divergence-free ADM momenta, 
i.e., the solutions of the momentum constraint.
\end{proof}

Note that although Proposition \ref{prop:CVD} implies
\begin{equation}
T_g \calM/\calD_0 \approx T_{[g]}\calC/\calD_0\;\oplus\; T_{\omega_g}\calV/\calD_0\;\oplus\; \Drift_g,
\end{equation}
the map $\pi_*$ from Theorem \ref{thm:mc} is \textit{not}
the pushforward from $T_g\calM$ to $T_g\calM/\calD_0$.
Indeed, if $W^a$ and $V^a$ are vector fields, the pushforward of $(g_{ab};\; \ck_g W_{ab}, \div_g V)$ is the drift
$\{ g_{ab};\;0,\div_g(V-W)\}$, but the drift component of $\pi_*((g_{ab};\; \ck_g W_{ab}, \div_g V))$
is $\{ g_{ab};\;0,\div_g(V)\}$. This is the key
distinction between diagrams \eqref{diag:dl-legendre-extended-M} and \eqref{diag:dl-legendre-drift}
and is what ensures that $j_\alpha$ is always an isomorphism even though $k_\alpha$
from diagram \eqref{diag:dl-legendre-extended-M} can fail to be one.  As always, the choice
to make the drift component of $\pi_*((g_{ab};\; \ck_g W_{ab}, \div_g V))$
equal to $\{ g_{ab};\;0,\div_g(V)\}$ rather than 
$\{ g_{ab};\;\ck_g W,0\}=\{ g_{ab};\;0,-\div_g W\}$ is arbitrary, and a 
result analogous to Theorem \ref{thm:mc} holds when using conformal drift.

On the other hand, if we
do identify $T_{[g]}\calC/\calD_0\,\oplus\, T_{\omega_g}\calV/\calD_0\,\oplus\, \Drift_g$
with $T_g\calM/\calD_0$ (thinking of $\pi_*$ as a projection into $T_g\calM/\calD_0$
that is not the pushforward),  we can 
interpret $j_\alpha$ as being the Legendre transformation of the total kinetic energy
\begin{equation}\label{eq:KE}
\calK_\alpha(\mathbf u, v, \mathbf V) = \calK^\expspace\calC_\alpha(\mathbf u) + 
\calK^\expspace\calV_\alpha(v) + \calK^{\Drift}_\alpha(\mathbf V)
\end{equation}
where $(\mathbf u, v, \mathbf V)\in T_{[g]}\calC/\calD_0\;\oplus\; T_{\omega_g}\calV/\calD_0\;\oplus\; \Drift_g
\approx T_g\calM/\calD_0$.  
Equation \eqref{eq:KE} can be obtained from the ADM kinetic energy, but we
must account for the fact that we are representing drift velocity in terms of $\mathbf V$
not $\mathbf W$. Recall that if $(g_{ab};\; u_{ab},\beta)\in T_{g}\calM$, the ADM kinetic energy is
\begin{equation}
\int \frac{1}{4}|u-\ck_g X|^2_g -\kappa (\beta-\div_g X)^2\; \alpha
\end{equation}
Decomposing $u_{ab}$ and $\beta$ according to equations \eqref{eq:cKEsplit} and \eqref{eq:vKEsplit}
we can rewrite the kinetic energy as
\begin{equation}\label{eq:ADMKEsplit}
\int N_{g,\alpha}^2|\sigma|^2_g - N_{g,\alpha}^2(\tau^*)^2 +\kappa\frac{1}{4}|\ck_g W|_g^2 -\kappa (\div_g V)^2\;\alpha
\end{equation}
The first two terms are the conformal and volumetric kinetic energies.  If the momentum constraint
is satisfied, then $V^a$ will satisfy the compatibility condition \eqref{eq:driftcompat2}
and hence the term involving $\div_g V$ in expression \eqref{eq:ADMKEsplit} 
is the drift kinetic energy. So the total kinetic energy
\eqref{eq:KE} is obtained from the ADM kinetic energy by dropping the $\ck_g W$ term.  
In the dual treatment of drift velocity, the total kinetic energy would be obtained by
keeping the $\ck_g W$ term and dropping the $\div_g V$ term instead.

\section{Drifts and the Conformal Method}\label{sec:driftcm}

Theorem \ref{thm:mc} and Proposition \ref{prop:tstar} show that given a choice of lapse form $\alpha$,
solutions of the momentum constraint can be parameterized in terms
of their conformal, volumetric, and drift momenta. Hence drift momentum naturally
complements the candidate parameters for conformal-like methods discussed at the end
of Section \ref{sec:volmom}, and we consider in this section variations of the conformal method
that incorporate drift as a parameter.

Suppose $(\ol g_{ab}, \ol K_{ab})$ is a solution of the momentum constraint with
conformal momentum $\bfsigma = \{\ol g_{ab};\; \ol \sigma_{ab}\}^*$,
volumetric momentum $-2\kappa\tau^*$, and drift momentum
\begin{equation}
\left\{g_{ab};\; 
\frac{1}{2N_{\ol g}}\ck_{\ol g}W_{ab}, 
-\frac{2\kappa}{N_{\ol g, \alpha}} \div_{\ol g} V
\right\}
\end{equation}
where $W^a$ and $V^a$ are vector fields solving the drift equation
\begin{equation}\label{eq:driftconf-a}
\div_{\ol g}\left[ \frac{1}{2N_{\ol g, \alpha}} \ck_{\ol g} W\right] = \kappa\, \mathbf{d} \left[ \frac{1}{N_{\ol g,\alpha}} \div_{\ol g} V\right].
\end{equation}
Then 
\begin{equation}
\ol K_{ab} = \ol \sigma_{ab} + \frac{1}{2N_{\ol g,\alpha} }\ck_{\ol g} W + \frac{1}{n}\left(\tau^* + \frac{1}{N_{\ol g, \alpha}} \div_{\ol g}V \right)
\end{equation}
and $K_{ab}$ has mean curvature
\begin{equation}\label{eq:taudrift}
\tau = \ol g^{ab}\ol K_{ab} = \tau^* + \frac{1}{N_{\ol g, \alpha}} \div_{\ol g} V.
\end{equation}

Now suppose $\ol g_{ab} = \phi^{q-2}g_{ab}$ for some conformally related metric $g_{ab}$.  The conformal momentum
$\bfsigma$ is represented at $g_{ab}$ by $\sigma_{ab}=\phi^2\ol \sigma_{ab}$ and the volumetric momentum 
is still $-2\kappa\tau^*$.  Using the transformation law $N_{\ol g, \alpha} = \phi^q N_{g,\alpha}$
and the conformal transformation laws for divergences and for the conformal Killing operator, 
equation \eqref{eq:driftconf-a} can be written in terms of $g_{ab}$ as 
\begin{equation}\label{eq:driftconf}
\div_{g} \left[ \frac{1}{2N_{g,\alpha}}\ck_g W\right] = \kappa \phi^q \extd \left[ \frac{\phi^{-2q}}{N_{g,\alpha}} \div_g(\phi^q V) \right].
\end{equation}
Note that since
\begin{equation}\label{eq:tauconf}
\tau =  \tau^* + \frac{\phi^{-2q}}{N_{ g, \alpha}} \div_{g} (\phi^q V),
\end{equation}
equation \eqref{eq:driftconf} is simply the 
CTS-H momentum constraint after substituting
equation \eqref{eq:tauconf}.  The Hamiltonian constraint for $(\ol g_{ab}, \ol K_{ab})$
can also be written in terms of $g_{ab}$ making this same substitution and
we arrive at two conformal methods depending on whether we specify $V^a$ or $W^a$
in equation \eqref{eq:driftconf}.
First, using Theorem \ref{thm:drift} we can specify $V^a$ up to a conformal Killing field
$Q^a$ and we obtain the following modification of the CTS-H method.
\begin{problem}[CTS-H with Volumetric Drift] \label{prob:CTS-H-V}
Let $g_{ab}$ be a metric, $\sigma_{ab}$ a transverse traceless
tensor with respect to $g_{ab}$, $\tau^*$ a constant, $V^a$ a vector field, and
$\alpha$ a lapse form.  Setting $N=\omega_g/\alpha$,
find a conformal factor $\phi$, a vector field $W^a$ and a conformal Killing field
$Q^a$ such that
\begin{equation}
\begin{aligned}\label{eq:CTS-H-V}
-\dimk\Lap_g\phi + R_g\phi - \left|\sigma +\frac{1}{2N} \ck_g W\right|^2_g\phi^{-q-1} + \kappa\left(\tau^*+ \frac{\phi^{-2q}}{N}\div(\phi^q (V+Q)\right)^2\phi^{q-1} &= 0 \\
\div_g\left(\frac{1}{2N}\ck_g W\right) -\kappa \phi^q \extd \left( \frac{\phi^{-2q}}{2N} \div_g(\phi^q(V+Q))\right)&=0.
\end{aligned}
\end{equation}
\end{problem}
Alternatively, we can apply Theorem \ref{thm:driftrev} and specify $W^a$ up to a $\ol g_{ab}$
divergence-free vector field.  Since $\phi^{-q} E^a$ is divergence-free with respect to $\ol g_{ab}$
if and only if $E^a$ is divergence-free with respect to $g_{ab}$ we obtain the following.
\begin{problem}[CTS-H with Conformal Drift]  \label{prob:CTS-H-C}
Let $g_{ab}$ be a metric, $\sigma_{ab}$ a transverse traceless
tensor with respect to $g_{ab}$, $\tau^*$ a constant, $W^a$ a vector field, and
$\alpha$ a lapse form.  Setting $N=\omega_g/\alpha$,
find a conformal factor $\phi$, a vector field $V^a$ and a divergence-free vector field
$E^a$ such that
\begin{equation}\label{eq:CTS-H-C}
\begin{aligned}
-\dimk\Lap_g\phi + R_g\phi - \left|\sigma +\frac{1}{2N} \ck_g (W+\phi^{-q} E)\right|_g^2\phi^{-q-1} + \kappa\left(\tau^*+ \frac{\phi^{-2q}}{N}\div_g(\phi^q V)\right)^2\phi^{q-1} &= 0 \\
\div_g\left(\frac{1}{2N}\ck_g (W+\phi^{-q}E)\right) -\kappa \phi^q \extd \left( \frac{\phi^{-2q}}{2N} \div_g(\phi^q(V))\right) &=0.
\end{aligned}
\end{equation}
\end{problem}

The drift parameterizations in systems \eqref{eq:CTS-H-V} and \eqref{eq:CTS-H-C}
pose significant analytical challenges beyond those of the standard
conformal method.  For example, both equations of both systems are second order in $\phi$, and the Hamiltonian constraint is no longer semilinear in $\phi$.  Although the equations
for the standard conformal method are technically simpler, and therefore more attractive
at first glance, it may be that more sophisticated 
equations are required to effectively parameterize non-CMC solutions of the constraint equations.
We will return to the analysis systems \eqref{eq:CTS-H-V} and \eqref{eq:CTS-H-C} in future work. 
For now, we make some observations
to suggest that these systems are not intractable.  First, for fixed $\phi$,
the problem for the momentum constraint is equivalent
to solving one of the variations \eqref{eq:drifteq}
or \eqref{eq:driftrev} of the drift equation \eqref{eq:drift}
with respect to the metric $\phi^{q-2} g_{ab}$.  These are 
well-posed problems and hence it is natural to consider iteration schemes, not
unlike those for the standard conformal method, that alternate between solving 
the Hamiltonian and momentum constraints.  Semilinearity of the Hamiltonian constraint
could be restored in such an iteration scheme by constructing a sequence of mean 
curvatures according to equation \eqref{eq:taudrift}.  Moreover, since the CMC conformal method arises as the 
special case $V^a=0$ in system \eqref{eq:CTS-H-V} or $W^a=0$ in system \eqref{eq:CTS-H-C},
a natural first step is to consider the near-CMC case where
$V^a$ or $W^a$ is small.  It seems likely that near-CMC results similar to those
of the standard conformal method are feasible, and the harder work will be
determining the extent to which the geometric and physical structures that motivate the
drift parameterizations 
are sufficient to address the shortcomings of the standard conformal method for
non-constant mean curvature.

There is also the possibility that Problems \ref{prob:CTS-H-V} and \ref{prob:CTS-H-C}
require further refinement. We are representing drifts by vector fields, and this
introduces a degeneracy not present in the standard conformal method.
A solution of the constraint equations uniquely determines a conformal class and,
after selecting a lapse form, a conformal, volumetric and drift momentum.  But the
drift momentum determines a subspace of vector fields:
if $(g_{ab},\sigma_{ab},\tau^*,V^a,\alpha)$ is a tuple of conformal parameters for system \eqref{eq:CTS-H-V}
generating a solution $(\ol g_{ab}, \ol K_{ab})$ of the constraints, this same solution 
will be generated by $(g_{ab},\sigma_{ab},\tau^*,V^a+\ol E^a+\ol Q^a,\alpha)$ whenever $\ol Q^a$
is a conformal Killing field for $\ol g_{ab}$ and $\ol E^a$ is divergence-free with respect to $\ol g_{ab}$.
The conformal Killing field is not problematic since the set of conformal Killing fields is a 
conformal invariant, but the divergence-free vector fields for $\ol g_{ab}$ are not known 
\textit{a priori}, and
this poses a difficulty in determining
if two tuples of conformal parameters determine the same solution of the constraints. A successful
analysis system \eqref{eq:CTS-H-V} should exhibit an identifiable subset of vector 
fields $V^a$ such that solutions of the constraint equations determine only one vector field
from the subset, with a similar requirement holding for system \eqref{eq:CTS-H-C}.  The main difficulty is that
of uniquely representing drifts at $\ol g_{ab}$ using the conformally related
metric $g_{ab}$, but without knowing the connecting conformal factor.  

If $g_{ab}$ does not admit nontrivial conformal Killing fields, there is 
a way to uniquely identify the drifts at $g_{ab}$ with the drifts the conformally
related metric $\widehat g_{ab} = \phi^{q-2} g_{ab}$, and this leads to third, alternative,
parameterization. Let $\mathbf V\in \Drift_g$
and let $V^a$ be any representative.  We then send $\mathbf V$ to
\begin{equation}
\widehat {\mathbf V} = \{\widehat g_{ab};\; 0, \div_{\widehat g} (\phi^{-q} V)\} \in \Drift_{\widehat g}.
\end{equation}
The map is well defined, for if $U^a$ is another representative of $\mathbf V$,
there is a divergence-free vector field $E^a$ such that $U^a=V^a+E^a$; this uses
the fact that there are no nonzero conformal Killing fields.  But then
$\phi^{-q} U^a = \phi^{-q} V^a + \phi^{-q} E^a$, and since $\phi^{-q} E^a$
is divergence-free with respect to $\widehat g_{ab}$, 
\begin{equation}
\{\widehat g_{ab};\; 0, \div_{\widehat g} \phi^{-q} U\} =\{\widehat g_{ab};\; 0, \div_{\widehat g} \phi^{-q} V\}.
\end{equation}
Hence the map is well defined, and since it has an analogous inverse we have established an 
identification of $\Drift_{ g}$
with $\Drift_{\widehat g}$.  Using this identification we make the substitution $V^a\mapsto \phi^{-q} V^a$
into equation \eqref{eq:taudrift} to obtain
\begin{equation}
\tau = \tau^* + \frac{\phi^{-2q}}{2N_{g,\alpha}}\div_g V
\end{equation}
and then substitute this mean curvature into the CTS-H equations.  Note, however, that $\div_g V$ is
a zero-mean function with respect to $g_{ab}$ and one can dispense with the vector field entirely.
\begin{problem}[CTS-H with Lapse-Scaled Mean Curvature]  
Let $g_{ab}$ be a metric with no nontrivial conformal Killing fields, $\sigma_{ab}$ a transverse traceless
tensor with respect to $g_{ab}$, $\tau^*$ a constant, $\xi$ a zero-mean function, and
$\alpha$ a lapse form.  Setting $N=\omega_g/\alpha$,
find a conformal factor $\phi$ and a vector field $W^a$ such that
\begin{equation}
\begin{aligned}\label{eq:CTS-H-lapsetau}
-\dimk\Lap_g\phi + R_g\phi - \left|\sigma +\frac{1}{2N} \ck_g W\right|^2_g\phi^{-q-1} + \kappa\left(\tau^*+ \frac{\phi^{-2q}}{N}\xi\right)^2\phi^{q-1} &= 0 \\
\div_g\left(\frac{1}{2N}\ck_g W\right) - \kappa \phi^q d\; \left( \frac{\phi^{-2q}}{2N} \xi \right)
&= 0.
\end{aligned}
\end{equation}
\end{problem}
One could also work with the substitution $W^a\mapsto \phi^{-q} W^a$ in system \eqref{eq:CTS-H-C}
and obtain an analogous version of system \eqref{eq:CTS-H-lapsetau}, but this seems somewhat unnatural.

The drift parameterization also has the potential to inform the standard conformal method
when the background metric has nontrivial conformal Killing fields.  Very little is
known in this case: we have near-CMC existence under the very strong restriction that the mean 
curvature is constant along each
flow line of every conformal Killing field \cite{Bruhat:1992}, and we have near-CMC nonexistence
on Yamabe-non-negative manifolds if the conformal momentum is zero\cite{Isenberg:2004jd}.  Moreover, one can show
that conformal Killing fields pose a genuine obstacle for some near-CMC seed data,
but not others \cite{Maxwell:2014c}. The difficulty with conformal Killing fields
arises since the CTS-H momentum constraint is not always solvable when conformal Killing fields 
are present. Using the ideas that led to system \eqref{prob:CTS-H-V} 
one can adjust the standard conformal method to include a correction term involving 
a conformal Killing field to restore solvability of the momentum constraint, and
we will address this modification of the CTS-H equations in future work.

\section{Conclusion}

In hindsight, York's original CMC conformal method can be thought of as having three parameters: 
\begin{itemize}
\itemsep 0pt
\item 
a conformal class $\mathbf g$ in $\calC$, 
\item 
a conformal momentum $\bfsigma$ in $T^*\calC/\calD_0$, and 
\item 
a volumetric momentum $-2\kappa\tau_0$ in $T^*\calV/\calD_0$. 
\end{itemize}
CMC data sets are special, however: their conformal and volumetric momenta
are unambiguously defined, intrinsic properties.  The extension of the
conformal method to non-CMC initial data sets employs a 
fourth parameter, a densitized lapse, which is used to measure
conformal momentum in a way that only depends on conformal properties
of the solution.  The conformal momentum is  
is compatible with the ADM Lagrangian, as seen in diagram \eqref{diag:dl-legendre-conf},
and the resulting non-CMC conformal method has four parameters:
\begin{itemize}
\itemsep 0pt
\item 
a conformal class $\mathbf g$ in $\calC$, 
\item 
a densitized lapse, represented by a lapse form $\alpha$,
\item
a conformal momentum $\bfsigma$ in $T^*\calC/\calD_0$ as measured by $\alpha$, and 
\item 
a mean curvature $\tau$.
\end{itemize}
In this formulation the mean curvature no longer directly controls the
volumetric momentum of the solution.  We saw
in Section \ref{sec:volmom}, however, that the mechanism used by the
standard conformal method to interpret conformal momentum
can be applied to the volumetric degrees of freedom, 
and volumetric momentum, as measured by a densitized lapse,
emerges as a property of a non-CMC initial data set. The parallels
between conformal and volumetric momenta are striking, and 
indeed the volumetric theory described in Section \ref{sec:volumetricLegendre}
is completely analogous to the conformal theory of Section \ref{sec:conformalLegendre}.
We have therefore considered alternatives to the conformal method where
the mean curvature is determined indirectly by specification
of a volumetric momentum and some other ingredient, and we have
identified drifts as playing a role in understanding these
alternatives.  

Indeed, every solution of the momentum constraint is a sum of a conformal momentum,
a volumetric momentum,  and a drift momentum, which is represented by
a pair of vector fields $W^a$ and $V^a$ solving the
drift equation
\begin{equation}\label{eq:concDrift}
\div_g\left(\frac{1}{2N_{g,\alpha}} \ck_g W\right) = \kappa\; \extd \left(\frac{1}{N_{g,\alpha}} \div_g V\right).
\end{equation}
Section \ref{sec:drift-mom} showed that equation \eqref{eq:concDrift} 
is not really a relationship between vector fields, but is
a relationship between a pair of drifts $(\mathbf W,\mathbf V)$.  Moreover, the relationship is symmetric:
either of $\mathbf W$ or $\mathbf V$ determines the other, and each of $\mathbf W$ or $\mathbf V$
can be taken as the velocity representing drift motion. Section \ref{sec:drift-vmke}
described dual theories, depending on the choice of using $\mathbf W$ or $\mathbf V$, 
in which the ADM kinetic energy descends to  
a kinetic energy Lagrangian without constraints on a 
tangent space decomposed into conformal, volumetric, and drift motion.
We were obligated, however, to pick either conformal or volumetric drift as representing drift velocity
because the difference $\mathbf V-\mathbf W$, which is the drift component of ADM velocity projected 
into $T\calM/\calD_0$, is not always sufficient to detect distinct solutions  of the constraint equations.

These results show that the introduction of a densitized lapse into the ADM Lagrangian
leads to a rich structure.
Although some of this structure is employed by the standard conformal method,
some of it is ignored, and in Section \ref{sec:driftcm} we saw
that there are alternative extensions of the CMC conformal method
that incorporate volumetric momentum and drift as parameters instead of
mean curvature.  Indeed there are a number of ways to do this, and it
is not yet clear how to best work with drift.  Nevertheless, 
future progress in applying the conformal method, or some variation,
in the non-CMC setting will require new ideas.  An improved understanding 
of the geometry of the conformal method,
of the type sought here, may well assist with these efforts.

% The preceding facts detail the interaction of the momentum constraint with
% densitized lapses, and as such they are statements about the constraint equations
% regardless of whether one is interested in the conformal method or not.  One would like,
% however, to capitalize on these observations to construct non-CMC solutions of the constraint
% equations and perhaps generate a superior extension of the CMC conformal method into
% the non-CMC regime.  The variations of the conformal method presented
% in Section \ref{sec:driftcm} naturally arise as candidates for 
% parameterizations of the constraint equations that incorporate drift, but it 
% remains to be seen how these alternatives compare with the standard conformal method.
% Regardless, it seems likely that progress in understanding the standard
% conformal method or its variations in the non-CMC setting will require
% a better understanding of the geometry of the constraint equations,
% and the results described here may be helpful in these efforts.

\section*{Acknowledgment}
This work was supported by NSF grant 0932078 000 while I was a resident at 
the Mathematical Sciences Research Institute in Berkeley, California,
and was additionally supported by NSF grant 1263544.

\bibliographystyle{amsalpha-abbrv}
\bibliography{Drift,Drift-manual}

\end{document}